\documentclass[]{article}

\usepackage{pslatex}
\usepackage{bussproofs}
\usepackage{listings}
\usepackage{fancybox}
\usepackage{mathtools}
\usepackage{tikz}
\usetikzlibrary{arrows.meta, shapes, arrows, automata, fit, matrix,
  positioning, decorations.pathreplacing, decorations.pathmorphing, calligraphy}
\usepackage{amssymb}
\usepackage{amsmath}
\usepackage{thm-restate}
\usepackage{caption}
\usepackage{subcaption}
\usepackage[hidelinks]{hyperref}
\usepackage{xcolor}

\newif\ifLongVersion\LongVersiontrue
\newif\ifArticle\Articletrue

\newtheorem{fact}{Fact}
\newtheorem{assumption}{Assumption}

\newcommand{\N}{\mathbb{N}}
\newcommand{\universe}{\mathbb{C}}
\newcommand{\vars}{\mathbb{V}}

\newcommand{\preds}{\mathbb{A}}
\newcommand{\isdef}{\stackrel{\scriptscriptstyle{\mathsf{def}}}{=}}
\newcommand{\interv}[2]{[{#1},{#2}]}
\newcommand{\tuple}[1]{\langle {#1} \rangle}

\newcommand{\set}[1]{\{ {#1} \}}
\newcommand{\Set}[1]{\left\{ {#1} \right\}}
\newcommand{\pow}[1]{\mathrm{pow}({#1})}
\newcommand{\dom}[1]{\mathrm{dom}({#1})}

\newcommand{\finsubseteq}{\subseteq_{\mathit{fin}}}

\newcommand{\bigO}{\mathcal{O}}
\newcommand{\twoexptime}{$2\mathsf{EXP}$}

\newcommand{\comps}{\mathcal{C}}
\newcommand{\acomp}{\mathsf{C}}

\newcommand{\interacs}{\mathcal{I}}
\newcommand{\ainterac}{\mathsf{I}}
\newcommand{\interac}[4]{\tuple{{#1}.{\mathit{#2}}, {#3}.{\mathit{#4}}}}

\newcommand{\beh}{\mathcal{B}}

\newcommand{\states}{Q}

\newcommand{\ports}{P}

\newcommand{\smallstate}{q}
\newcommand{\arrow}[2]{\xrightarrow{{\scriptstyle #1}}_{{\scriptstyle #2}}}
\newcommand{\Arrow}[2]{\xRightarrow{{\scriptstyle #1}}_{\raisebox{4pt}{\!$\scriptstyle{#2}$}}}

\makeatletter
\newcommand*{\da@rightarrow}{\mathchar"0\hexnumber@\symAMSa 4B }
\newcommand*{\da@leftarrow}{\mathchar"0\hexnumber@\symAMSa 4C }
\newcommand*{\xdashrightarrow}[2][]{%
  \mathrel{%
    \mathpalette{\da@xarrow{#1}{#2}{}\da@rightarrow{\,}{}}{}%
  }%
}
\newcommand{\xdashleftarrow}[2][]{%
  \mathrel{%
    \mathpalette{\da@xarrow{#1}{#2}\da@leftarrow{}{}{\,}}{}%
  }%
}
\newcommand*{\da@xarrow}[7]{%
  \sbox0{$\ifx#7\scriptstyle\scriptscriptstyle\else\scriptstyle\fi#5#1#6\m@th$}%
  \sbox2{$\ifx#7\scriptstyle\scriptscriptstyle\else\scriptstyle\fi#5#2#6\m@th$}%
  \sbox4{$#7\dabar@\m@th$}%
  \dimen@=\wd0 %
  \ifdim\wd2 >\dimen@
    \dimen@=\wd2 %
  \fi
  \count@=2 %
  \def\da@bars{\dabar@\dabar@}%
  \@whiledim\count@\wd4<\dimen@\do{%
    \advance\count@\@ne
    \expandafter\def\expandafter\da@bars\expandafter{%
      \da@bars
      \dabar@ 
    }%
  }%
  \mathrel{#3}%
  \mathrel{%
    \mathop{\da@bars}\limits
    \ifx\\#1\\%
    \else
      _{\copy0}%
    \fi
    \ifx\\#2\\%
    \else
      ^{\copy2}%
    \fi
  }%
  \mathrel{#4}%
}
\makeatother

\newcommand{\Open}[2]{\xdashrightarrow{{\scriptstyle #1}}_{\raisebox{0pt}{\!$\scriptstyle{#2}$}}}

\newcommand{\store}{\nu}

\newcommand{\statemap}{\varrho}
\newcommand{\config}{(\comps,\interacs,\statemap,\store)}
\newcommand{\aconfig}{\gamma}

\newcommand{\configset}{\Gamma}

\newcommand{\errconfigs}{\top}
\newcommand{\ahavoc}{\mathfrak{h}}
\newcommand{\havocrule}{$\mathsf{Havoc}$}

\newcommand{\comp}{\bullet}
\newcommand{\bigcomp}{\scalebox{1.7}{$\comp$}}
\newcommand{\substreq}{\sqsubseteq}

\newcommand{\isubstreq}{\substreq}

\newcommand{\lift}[2]{{#1}\!\!\uparrow^{\scriptscriptstyle{#2}}}

\renewcommand{\smallsetminus}{\!\setminus\!}

\newcommand{\adl}{\textsf{CL}}
\newcommand{\predname}[1]{\mathsf{#1}}
\newcommand{\apred}{\predname{A}}
\newcommand{\bpred}{\predname{B}}
\newcommand{\emp}{\predname{emp}}

\let\Asterisk\undefined
\newcommand{\Asterisk}{\mathop{\scalebox{1.9}{\raisebox{-0.2ex}{$\ast$}}}\hspace*{1pt}}%
\renewcommand{\vec}[1]{\mathbf #1}
\newcommand{\fv}[1]{\mathrm{fv}({#1})}
\newcommand{\spaceform}{\xi}
\newcommand{\pureform}{\pi}
\newcommand{\compin}[2]{{#1}@{#2}}
\newcommand{\company}[1]{\compin{#1}{\_}}

\newcommand{\seplog}{\textsf{SL}}

\newcommand{\unfold}[1]{\Leftarrow_{#1}}

\newcommand{\unfoldrule}{\leftarrow}
\newcommand{\asid}{\Delta}

\newcommand{\chain}[2]{\predname{chain}_{{#1},{#2}}}

\newcommand{\symconfeq}[1]{\simeq_{#1}}
\newcommand{\symconfneq}[1]{\not\simeq_{#1}}

\newcommand{\stageno}{\mathfrak{n}}
\newcommand{\stagemset}{\mathfrak{m}}

\newcommand{\acomm}{\predname{R}}
\newcommand{\primcomms}{\mathfrak{P}}
\newcommand{\localcomms}{\mathfrak{L}}
\newcommand{\new}{\predname{new}}
\newcommand{\delete}{\predname{delete}}
\newcommand{\connect}{\predname{connect}}
\newcommand{\disconnect}{\predname{disconnect}}

\newcommand{\whencomm}[2]{\predname{when}\;#1\;\predname{do}\;#2}
\newcommand{\withcomm}[3]{\predname{with}\;{#1}:{#2}\;\predname{do}\;{#3}\;\predname{od}}

\newcommand{\predtrue}{\predname{true}}
\newcommand{\predfalse}{\predname{false}}

\newcommand{\skipcomm}{\predname{skip}}
\renewcommand{\mod}{~\mathrm{mod}~}

\newcommand{\succj}[3]{{#1}: {#2} \leadsto {#3}}
\newcommand{\abort}[2]{{#1}: {#2} ~\text{\rotatebox[origin=c]{90}{$\leadsto$}}}
\newcommand{\rbr}{{\bf ]\!]}}
\newcommand{\lbr}{{\bf [\![}}
\newcommand{\sem}[2]{{\lbr #1 \rbr}_{\scriptscriptstyle{#2}}}
\newcommand{\lng}{{\langle \! \langle}}
\newcommand{\rng}{{\rangle \! \rangle}}
\newcommand{\llng}{{\langle \!\! \langle}}
\newcommand{\rrng}{{\rangle \!\! \rangle}}
\newcommand{\semcomm}[1]{\lng #1 \rng}

\newcommand{\hoare}[3]{\{ {#1} \} ~{#2}~ \{ {#3} \}}

\newcommand{\modif}[1]{\predname{modif}({#1})}
\newcommand{\hinv}{(\sharp)}

\newcommand{\open}{\mathfrak{o}}
\newcommand{\looserule}{$\mathsf{Loose}$}
\newcommand{\tightrule}{$\mathsf{Tight}$}

\newcommand{\interof}[1]{\Sigma[{#1}]}
\newcommand{\are}{\predname{L}}

\newcommand{\parcomp}{\bowtie}
\newcommand{\lpar}{\{\!\!\{}
\newcommand{\rpar}{\}\!\!\}}
\newcommand{\aenv}{\eta}
\newcommand{\havoctriple}[4]{{#1}\triangleright\lpar{#2}\rpar~{#3}~\lpar{#4}\rpar}
\newcommand{\semlang}[2]{\llng {#1} \rrng({#2})}
\newcommand{\requiv}[1]{\cong_{\scriptstyle{#1}}}
\newcommand{\wopen}[1]{\open{[{#1}]}}
\newcommand{\proj}[2]{{#1}\!\downarrow_{\scriptscriptstyle{#2}}}
\newcommand{\supp}[1]{\mathrm{supp}({#1})}
\newcommand{\iatoms}[1]{\mathrm{inter}({#1})}
\newcommand{\patoms}[1]{\mathrm{preds}({#1})}
\newcommand{\atoms}[1]{\mathrm{atoms}({#1})}
\newcommand{\disabled}[2]{{#1} \dagger {#2}}
\newcommand{\excluded}[2]{{#1} \ddagger {#2}}
\newcommand{\struceq}{\simeq}
\newcommand{\strucgeq}{\succeq}

\newcommand{\inter}{$\Sigma$}
\newcommand{\ui}{$\cup$}
\newcommand{\ue}{$\subset$}
\newcommand{\congr}{$\requiv{}$}
\newcommand{\lu}{$\mathsf{LU}$}
\newcommand{\conseq}{$\mathsf{C}$}
\newcommand{\disr}{$\dagger$}
\newcommand{\supr}{\ii}
\newcommand{\ii}{$\ainterac-$}
\newcommand{\iidisr}{$\ainterac$\disr}
\newcommand{\ie}{$\ainterac+$}
\newcommand{\parr}{$\parcomp$}
\newcommand{\botr}{$\bot$}
\newcommand{\epsilonr}{$\epsilon$}
\newcommand{\front}[2]{\mathcal{F}({#1},{#2})}

\newcommand{\havocline}[4]{
  \UnaryInfC{$\havoctriple{#1}{#2}{#3}{#4}$} }
\newcommand{\binhavocline}[4]{
  \BinaryInfC{$\havoctriple{#1}{#2}{#3}{#4}$}
}
\newcommand{\trihavocline}[4]{
  \TrinaryInfC{$\havoctriple{#1}{#2}{#3}{#4}$}
}

\newcommand{\havoctwolines}[4]{
  \UnaryInfC{${#1}\triangleright~\lpar{#2}\rpar$}
  \noLine
  \UnaryInfC{${#3}~\lpar{#4}\rpar$}
}
\newcommand{\binhavoctwolines}[4]{
  \BinaryInfC{${#1}\triangleright~\lpar{#2}\rpar$}
  \noLine
  \UnaryInfC{${#3}~\lpar{#4}\rpar$}
}
\newcommand{\trihavoctwolines}[4]{
  \TrinaryInfC{${#1}\triangleright~\lpar{#2}\rpar$}
  \noLine
  \UnaryInfC{${#3}~\lpar{#4}\rpar$}
}
\newcommand{\quahavoctwolines}[4]{
  \QuaternaryInfC{${#1}\triangleright~\lpar{#2}\rpar$}
  \noLine
  \UnaryInfC{${#3}~\lpar{#4}\rpar$}
}
\newcommand{\havocfourlines}[4]{
  \UnaryInfC{${#1}$}
  \noLine
  \UnaryInfC{$\triangleright~\lpar{#2}\rpar$}
  \noLine
  \UnaryInfC{${#3}$}
  \noLine
  \UnaryInfC{$\lpar{#4}\rpar$}
}

\newcommand{\interofyxch}{{\Sigma_{y,x}^1}}

\newcommand{\interofzy}{{\Sigma_{z,y}}}

\newcommand{\interoftree}[1]{\interof{ \tree({#1}) }}

\newcommand{\toktoken}{\mathsf{T}}
\newcommand{\toknotok}{\mathsf{H}}

\newcommand{\tokin}{\textit{in}}
\newcommand{\tokout}{\textit{out}}

\newcommand{\stateidle}{\mathit{idle}}
\newcommand{\stateleft}{\mathit{left}}
\newcommand{\stateright}{\mathit{right}}

\newcommand{\stateleafidle}{\mathit{leaf\!\_idle}}
\newcommand{\stateleafbusy}{\mathit{leaf\!\_busy}}
\newcommand{\lrecv}{\mathit{r_\ell}}
\newcommand{\rrecv}{\mathit{r_r}}
\newcommand{\send}{\mathit{s}}
\newcommand{\treeseg}{\predname{tseg}}

\newcommand{\tree}{\predname{tree}}
\newcommand{\treeidle}{\tree_{\stateidle}}

\newcommand{\treenotidle}{\tree_{\neg\stateidle}}
\newcommand{\treestar}{\tree_{\star}}

\newcommand{\codesize}{\footnotesize}
\newcommand{\boxaround}[1]{\ovalbox{${#1}$}}

\lstdefinelanguage{JavaScript}{
  keywords={typeof, true, false, catch, function, return, null, catch, switch, var, if, in, while, do, od, else, case, break, when, with, assume},
  ndkeywords={class, export, boolean, throw, implements, import, this},
  sensitive=false,
  comment=[l]{//},
  morecomment=[s]{/*}{*/},
  morecomment=[s]{$}{$},
  morestring=[b]',
  morestring=[b]"
}
\lstnewenvironment{code}[1][]{%
    \lstset{%
	    basicstyle=\codesize\ttfamily, 
	    breaklines=false,
	    commentstyle=\color[HTML]{073642}\ttfamily\itshape,
	    identifierstyle=\color{black},
	    inputencoding=latin1,
	    keywordstyle=\color[HTML]{228B22}\bfseries,
	    language=JavaScript,
	    ndkeywordstyle=\color[HTML]{228B22}\bfseries,
	    numbers=left,
	    numbersep=5pt,
	    xleftmargin=30pt,
	    xrightmargin=10pt,
	    frame=Tb,
	    framexleftmargin=15pt,
	    numberstyle=\scriptsize, 
	    prebreak = \raisebox{0ex}[0ex][0ex]{\ensuremath{\hookleftarrow}},
	    showstringspaces=false,
	    stringstyle=\color[rgb]{0.639,0.082,0.082}\ttfamily,
	    upquote=true,
	    mathescape=true,
    	escapebegin=\color[HTML]{073642},
	    tabsize=2,
	    #1
    }
}{}

\lstnewenvironment{examplecode}[1][]{%
    \lstset{%
      basicstyle=\codesize\ttfamily, 
      breaklines=false,
      commentstyle=\color[HTML]{073642}\ttfamily\itshape,
      identifierstyle=\color{black},
      inputencoding=latin1,
      keywordstyle=\color[HTML]{228B22}\bfseries,
      language=JavaScript,
      ndkeywordstyle=\color[HTML]{228B22}\bfseries,
      numbers=left,
      numbersep=5pt,
      xleftmargin=75pt,
      xrightmargin=60pt,
      frame=Tb,
      framexleftmargin=15pt,
      numberstyle=\scriptsize, 
      prebreak = \raisebox{0ex}[0ex][0ex]{\ensuremath{\hookleftarrow}},
      showstringspaces=false,
      stringstyle=\color[rgb]{0.639,0.082,0.082}\ttfamily,
      upquote=true,
      mathescape=true,
      escapebegin=\color[HTML]{073642},
      tabsize=2,
      #1
    }
}{}

\begin{document}

\newtheorem{remark}{Remark}
\newtheorem{example}{Example}
\newtheorem{definition}{Definition}
\newtheorem{theorem}{Theorem}
\newtheorem{lemma}{Lemma}
\newtheorem{proposition}{Proposition}
\newcommand{\proof}[1]{\emph{Proof}. {#1}}
\newcommand{\qed}{$\Box$}

\title{Reasoning about Reconfigurations of Distributed Systems}         

\author{Emma Ahrens\footnote{Emma.Ahrens@rwth-aachen.de, RWTH Aachen
  University, D-52056, Germany} \and Marius
  Bozga\footnote{Marius.Bozga@univ-grenoble-alpes.fr, Univ. Grenoble
  Alpes, CNRS, Grenoble INP, VERIMAG, 38000, France} \and Radu
  Iosif\footnote{Radu.Iosif@univ-grenoble-alpes.fr, Univ. Grenoble
  Alpes, CNRS, Grenoble INP, VERIMAG, 38000, France} \and Joost-Pieter
  Katoen\footnote{katoen@cs.rwth-aachen.de, RWTH Aachen University,
  D-52056, Germany}}

\maketitle

\begin{abstract}
This paper presents a Hoare-style calculus for formal reasoning about
reconfiguration programs of distributed systems. Such programs create
and delete components and/or interactions (connectors) while the
system components change state according to their internal behaviour.
Our proof calculus uses a resource logic, in the spirit of Separation
Logic \cite{Reynolds}, to give local specifications of reconfiguration
actions. Moreover, distributed systems with an unbounded number of
components are described using inductively defined predicates. The
correctness of reconfiguration programs relies on havoc invariants,
that are assertions about the ongoing interactions in a part of the
system that is not affected by the structural change caused by the
reconfiguration. We present a proof system for such invariants in an
assume/rely-guarantee style. We illustrate the feasibility of our
approach by proving the correctness of real-life distributed systems
with reconfigurable (self-adjustable) tree architectures.
\end{abstract}

\section{Introduction}

\paragraph{The relevance of dynamic reconfiguration.}
Dynamic reconfigurable distributed systems are used increasingly as
critical parts of the infrastructure of our digital society,
e.g.\ datacenters, e-banking and social networking. In order to
address maintenance (e.g., replacement of faulty and obsolete network
nodes by new ones) and data traffic issues (e.g., managing the traffic
inside a datacenter \cite{DBLP:journals/comsur/Noormohammadpour18}),
the distributed systems community has recently put massive effort in
designing algorithms for \emph{reconfigurable systems}, whose network
topologies change at runtime \cite{DBLP:journals/sigact/FoersterS19}.
This development provides new impulses to distributed algorithm
design~\cite{DBLP:conf/podc/MichailSS20} and has given rise to
self-adjustable network architectures whose topology reconfigurations
are akin to amendments of dynamic data structures such as splay
trees~\cite{DBLP:conf/infocom/PeresSGA019}. However, reconfiguration
is an important source of bugs, that may result in denial of services
or even data corruption\footnote{E.g., Google reports a cloud failure
  caused by reconfiguration:
  \url{https://status.cloud.google.com/incident/appengine/19007}}. \emph{This
  paper introduces a logical framework for reasoning about the safety
  properties of such systems, in order to prove e.g., absence of
  deadlocks or data races.}

\vspace*{-.5\baselineskip}
\paragraph{Modeling distributed systems.}
In this paper we model distributed systems at the level of abstraction
commonly used in component-based design of large heterogenous systems
\cite{magee1996dynamic}. We rely on a clean separation of (a
finite-state abstraction of) the \emph{behaviour} from the
\emph{coordination} of behaviors, described by complex graphs of
components (nodes) and interactions (edges). Mastering the complexity
of a distributed system requires a deep understanding of the
coordination mechanisms. We distinguish between \emph{endogenous}
coordination, that explicitly uses synchronization primitives in the
code describing the behavior of the components (e.g.\ semaphores,
monitors, compare-and-swap, etc.) and \emph{exogenous} coordination,
that defines global rules describing how the components
interact. These two orthogonal paradigms play different roles in the
design of a system: exogenous coordination is used during high-level
model building, whereas endogenous coordination is considered at a
later stage of development, to implement the model using low-level
synchronization primitives.

Here we focus on \emph{exogenous coordination} of distributed systems,
consisting of an unbounded number of interconnected components, with a
flexible topology, i.e.\ not fixed \emph{\`a priori}. We abstract from
low-level coordination mechanisms between processes such as
semaphores, compare-and-swap operations and the like. Components
behave according to a small set of finite-state abstractions of
sequential programs, whose transitions are labeled with events. They
communicate via interactions (handshaking) modeled as sets of events
that occur simultaneously in multiple components. Despite their
apparent simplicity, these models capture key aspects of distributed
computing, such as message delays and transient faults due to packet
loss. Moreover, the explicit graph representation of the network is
essential for the modeling of dynamic reconfiguration actions.

\vspace*{-.5\baselineskip}
\paragraph{Programming reconfiguration}
The study of dynamic reconfiguration has led to the development of a
big variety of formalisms and approaches to specify the changes to the
structure of a system using e.g., graph-based, logical or
process-algebraic formalisms (see \cite{bradbury2004survey} and
\cite{rumpe2017classification} for surveys). With respect to existing
work, we consider a simple yet general imperative reconfiguration
language, encompassing four primitive reconfiguration actions
(creation and deletion of components and interactions) as well as
non-deterministic reconfiguration triggers (constraints) evaluated on
small parts of the structure and the state of the system.  These
features exist, in very similar forms, in the vast majority of
existing graph-based reconfiguration formalisms e.g., using explicit
reconfiguration scripts as in \textsc{CommUnity} \cite{CommUnity},
reconfiguration controllers expressed as production rules in
graph-grammars \cite{LeMetayer}, guarded reconfiguration actions in
\textsc{Dr-Bip} \cite{DR-BIP-STTT} and graph rewriting rules in
\textsc{Reo} \cite{Reo}, to cite only a few. In our model, the
primitive reconfiguration actions are executed sequentially, but
interleave with the firing of interactions i.e., the normal execution
of the system. Sequential reconfiguration is not a major restriction,
as the majority of reconfiguration languages rely on a centralized
management \cite{bradbury2004survey}. Nevertheless, for the sake of
simplicity, most existing reconfiguration languages avoid the
fine-grain interleaving of reconfiguration and execution steps i.e.,
they freeze the system's execution during reconfiguration. Our choice
of allowing this type of interleaving is more realistic and closer to
real-life implementation. Finally, our language supports open
reconfigurations, in which the number of possible configurations is
unbounded \cite{rumpe2017classification}, via non-deterministic choice
and iteration.

\vspace*{-.5\baselineskip}
\paragraph{An illustrative example.}
We illustrate the setting by a token ring example, consisting of a
finite but unbounded number of components, indexed from $1$ to $n$,
connected via an unidirectional ring (Fig. \ref{fig:token-ring}).  A
token may be passed from a component $i$ in state $\toktoken$ (it has
a token) to its neighbour, with index $(i \mod n) + 1$, which must be
in state $\toknotok$ (it has a hole instead of a token). As result of
this interaction, the $i$-th component moves to state $\toknotok$
while the $(i\mod n) + 1$ component moves to state $\toktoken$. Note
that token passing interactions are possible as long as at least two
components are in different states; if all the components are in the
same state at the same time, the ring is in a \emph{deadlock}
configuration.

\begin{figure}[t!]
  \vspace*{-\baselineskip}
  \caption[Token Ring]{Reconfiguration of a Parametric Token Ring System}
  \label{fig:token-ring}
  \begin{subfigure}[b]{.5\linewidth}
    \centering
    \ifArticle
    \scalebox{0.64}{		\begin{tikzpicture}[>=stealth',shorten >=1pt,auto,node distance=2cm]
			\begin{scope}[local bounding box=a, shift={(-3,2)}]
				\node[state] (h)      {$\toknotok$};
				\node[state, fill=lightgray] (t) [right of=h]     {$\toktoken$};

				\path[->] (h) edge [bend left] node (i1) {$\tokin$} (t);
				\path[->] (t) edge [bend left] node (i2) {$\tokout$} (h);

				\node (Ca) [draw=black, fit= (h) (t) (i1) (i2)] {};
				\node [xshift=-10pt,yshift=-22pt] at (Ca.east) {$c_1$};
			\end{scope}

			\begin{scope}[local bounding box=b, shift={(3,2)}]
				\node[state, fill=lightgray] (h)      {$\toknotok$};
				\node[state] (t) [right of=h]     {$\toktoken$};

				\path[->] (h) edge [bend left] node (i1) {$\tokin$} (t);
				\path[->] (t) edge [bend left] node (i2) {$\tokout$} (h);

				\node (Cb) [draw=black, fit= (h) (t) (i1) (i2)] {};
				\node [xshift=-10pt,yshift=-22pt] at (Cb.east) {$c_2$};
			\end{scope}

			\begin{scope}[local bounding box=c, shift={(3,-2)}]
				\node[state, fill=lightgray] (h)      {$\toknotok$};
				\node[state] (t) [right of=h]     {$\toktoken$};

				\path[->] (h) edge [bend left] node (i1) {$\tokin$} (t);
				\path[->] (t) edge [bend left] node (i2) {$\tokout$} (h);

				\node (Cc) [draw=black, fit= (h) (t) (i1) (i2)] {};
				\node [xshift=-10pt,yshift=-22pt] at (Cc.east) {$c_3$};
			\end{scope}

			\begin{scope}[local bounding box=d, shift={(-3,-2)}]
				\node[state, fill=lightgray] (h)      {$\toknotok$};
				\node[state] (t) [right of=h]     {$\toktoken$};

				\path[->] (h) edge [bend left] node (i1) {$\tokin$} (t);
				\path[->] (t) edge [bend left] node (i2) {$\tokout$} (h);

				\node (Cd) [draw=black, fit= (h) (t) (i1) (i2)] {};
				\node [xshift=-10pt,yshift=-22pt] at (Cd.east) {$c_n$};
			\end{scope}

			\begin{scope}[local bounding box=e, shift={(1,-2)}]
				\node[state, draw=white] (Ce) {\ldots};
			\end{scope}

			\path [draw=black,fill=black] (Ca.south) circle (0.5mm);
			\path [draw=black,fill=black] (Ca.east) circle (0.5mm);
			\path [draw=black,fill=black] (Cb.west) circle (0.5mm);
			\path [draw=black,fill=black] (Cb.south) circle (0.5mm);
			\path [draw=black,fill=black] (Cc.north) circle (0.5mm);
			\path [draw=black,fill=black] (Cc.west) circle (0.5mm);
			\path [draw=black,fill=black] (Cd.north) circle (0.5mm);
			\path [draw=black,fill=black] (Cd.east) circle (0.5mm);

			\path [draw=black,fill=black] (Ce.west) circle (0.5mm);
			\path [draw=black,fill=black] (Ce.east) circle (0.5mm);

			\draw[-] (Ca) edge node[pos=0.1] () {$\tokout$} (Cb);
			\draw[-] (Ca) edge node[pos=0.9] () {$\tokin$} (Cb);
			\draw[-] (Ca) edge node[pos=0.5, below] () {$(c_1,\mathit{out},c_2,\mathit{in})$} (Cb);
			\draw[-] (Cb) edge node[pos=0.1] () {$\tokout$} (Cc);
			\draw[-] (Cb) edge node[pos=0.85] () {$\tokin$} (Cc);
			\draw[-] (Cc) edge node[pos=0.5] () {$(c_2,\mathit{out},c_3,\mathit{in})$} (Cb);
			\draw[-] (Cc) edge node[pos=0.35] () {$\tokout$} (Ce);
			\draw[-] (Ce) edge node[pos=0.8] () {$\tokin$} (Cd);
			\draw[-] (Cd) edge node[pos=0.15] () {$\tokout$} (Ca);
			\draw[-] (Cd) edge node[pos=0.9] () {$\tokin$} (Ca);
			\draw[-] (Ca) edge node[pos=0.5] () {$(c_n,\mathit{out},c_1,\mathit{in})$} (Cd);
		\end{tikzpicture}}
    \else
    \scalebox{0.7}{		\begin{tikzpicture}[>=stealth',shorten >=1pt,auto,node distance=2cm]
			\begin{scope}[local bounding box=a, shift={(-3,2)}]
				\node[state] (h)      {$\toknotok$};
				\node[state, fill=lightgray] (t) [right of=h]     {$\toktoken$};

				\path[->] (h) edge [bend left] node (i1) {$\tokin$} (t);
				\path[->] (t) edge [bend left] node (i2) {$\tokout$} (h);

				\node (Ca) [draw=black, fit= (h) (t) (i1) (i2)] {};
				\node [xshift=-10pt,yshift=-22pt] at (Ca.east) {$c_1$};
			\end{scope}

			\begin{scope}[local bounding box=b, shift={(3,2)}]
				\node[state, fill=lightgray] (h)      {$\toknotok$};
				\node[state] (t) [right of=h]     {$\toktoken$};

				\path[->] (h) edge [bend left] node (i1) {$\tokin$} (t);
				\path[->] (t) edge [bend left] node (i2) {$\tokout$} (h);

				\node (Cb) [draw=black, fit= (h) (t) (i1) (i2)] {};
				\node [xshift=-10pt,yshift=-22pt] at (Cb.east) {$c_2$};
			\end{scope}

			\begin{scope}[local bounding box=c, shift={(3,-2)}]
				\node[state, fill=lightgray] (h)      {$\toknotok$};
				\node[state] (t) [right of=h]     {$\toktoken$};

				\path[->] (h) edge [bend left] node (i1) {$\tokin$} (t);
				\path[->] (t) edge [bend left] node (i2) {$\tokout$} (h);

				\node (Cc) [draw=black, fit= (h) (t) (i1) (i2)] {};
				\node [xshift=-10pt,yshift=-22pt] at (Cc.east) {$c_3$};
			\end{scope}

			\begin{scope}[local bounding box=d, shift={(-3,-2)}]
				\node[state, fill=lightgray] (h)      {$\toknotok$};
				\node[state] (t) [right of=h]     {$\toktoken$};

				\path[->] (h) edge [bend left] node (i1) {$\tokin$} (t);
				\path[->] (t) edge [bend left] node (i2) {$\tokout$} (h);

				\node (Cd) [draw=black, fit= (h) (t) (i1) (i2)] {};
				\node [xshift=-10pt,yshift=-22pt] at (Cd.east) {$c_n$};
			\end{scope}

			\begin{scope}[local bounding box=e, shift={(1,-2)}]
				\node[state, draw=white] (Ce) {\ldots};
			\end{scope}

			\path [draw=black,fill=black] (Ca.south) circle (0.5mm);
			\path [draw=black,fill=black] (Ca.east) circle (0.5mm);
			\path [draw=black,fill=black] (Cb.west) circle (0.5mm);
			\path [draw=black,fill=black] (Cb.south) circle (0.5mm);
			\path [draw=black,fill=black] (Cc.north) circle (0.5mm);
			\path [draw=black,fill=black] (Cc.west) circle (0.5mm);
			\path [draw=black,fill=black] (Cd.north) circle (0.5mm);
			\path [draw=black,fill=black] (Cd.east) circle (0.5mm);

			\path [draw=black,fill=black] (Ce.west) circle (0.5mm);
			\path [draw=black,fill=black] (Ce.east) circle (0.5mm);

			\draw[-] (Ca) edge node[pos=0.1] () {$\tokout$} (Cb);
			\draw[-] (Ca) edge node[pos=0.9] () {$\tokin$} (Cb);
			\draw[-] (Ca) edge node[pos=0.5, below] () {$(c_1,\mathit{out},c_2,\mathit{in})$} (Cb);
			\draw[-] (Cb) edge node[pos=0.1] () {$\tokout$} (Cc);
			\draw[-] (Cb) edge node[pos=0.85] () {$\tokin$} (Cc);
			\draw[-] (Cc) edge node[pos=0.5] () {$(c_2,\mathit{out},c_3,\mathit{in})$} (Cb);
			\draw[-] (Cc) edge node[pos=0.35] () {$\tokout$} (Ce);
			\draw[-] (Ce) edge node[pos=0.8] () {$\tokin$} (Cd);
			\draw[-] (Cd) edge node[pos=0.15] () {$\tokout$} (Ca);
			\draw[-] (Cd) edge node[pos=0.9] () {$\tokin$} (Ca);
			\draw[-] (Ca) edge node[pos=0.5] () {$(c_n,\mathit{out},c_1,\mathit{in})$} (Cd);
		\end{tikzpicture}}
    \fi
    \vspace*{2\baselineskip}
  \end{subfigure}%
  \begin{subfigure}[b]{.5\linewidth}
    \centering

{\footnotesize\begin{code}[caption={Delete Component (wrong version)}
        \label{lst:delete-component-wrong}]
with $x,y,z : \interac{x}{out}{y}{in} * \compin{y}{\toknotok} * \interac{y}{out}{z}{in}$ do
  disconnect(y.$\mathit{out}$,z.$\mathit{in}$);
  disconnect(x.$\mathit{out}$,y.$\mathit{in}$);
  delete(y);
  connect(x.$\mathit{out}$,z.$\mathit{in}$); od
\end{code}}    

\vspace*{-.5\baselineskip}

{\footnotesize\begin{code}[caption={Delete Component (correct version)}
    \label{lst:delete-component-notok}]
with $x,y,z : \interac{x}{out}{y}{in} * \compin{y}{\toknotok} * \interac{y}{out}{z}{in}$ do
  disconnect(x.$\mathit{out}$,y.$\mathit{in}$);
  disconnect(y.$\mathit{out}$,z.$\mathit{in}$);
  delete(y);
  connect(x.$\mathit{out}$,z.$\mathit{in}$); od
\end{code}}    
\end{subfigure}
\vspace*{-2\baselineskip}
\end{figure}
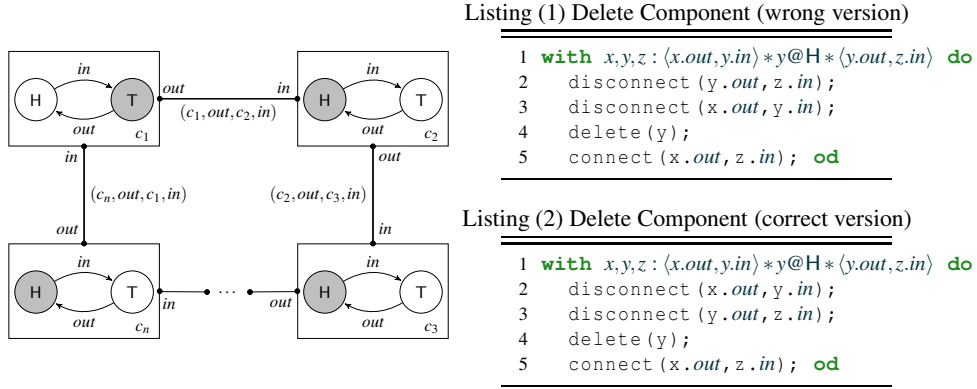

During operation, components can be added to, or removed from the
ring.  On removing the component with index $i$, its incoming (from
$i-1$, if $i>1$, or $n$, if $i=1$) and outgoing (to $(i\mod n)+1$)
connectors are deleted before the component is deleted, and its left
and right neighbours are reconnected in order to re-establish the
ring-shaped topology. Consider the program in Listing
\ref{lst:delete-component-wrong}, where the variables $x$, $y$ and $z$
are assigned indices $i$, $(i \mod n) + 1$ and $(i \mod n) + 2$,
respectively (assuming $n>2$). The program removes first the right
connector between $y$ and $z$ (line 2), then removes the left
connector between $x$ and $y$ (line 3), before removing the component
indexed by $y$ (line 4) and reconnecting the $x$ and $z$ components
(line 5). Note that the order of the disconnect commands is crucial: assume that
component $x$ is the only one in state $\toktoken$ in the entire
system. Then the token may move from $x$ to $y$ and is deleted
together with the component (line 4). In this case, the resulting ring
has no token and the system is in a deadlock configuration. The
reconfiguration program in Listing \ref{lst:delete-component-notok} is
obtained by swapping lines 2 and 3 from Listing
\ref{lst:delete-component-wrong}. In this case, the deleted component
is in state $\toknotok$ before the reconfiguration and its left
connector is removed before its right one, thus ensuring that the
token does not move to the $y$ component (deleted at line 4).

The framework developed in this paper allows to prove that e.g., when
applied to a token ring with at least two components in state
$\toknotok$ and at least one component in state $\toktoken$, the
program in Listing \ref{lst:delete-component-notok} yields a system
with at least two components in different states, \emph{for any $n >
  2$}. Using, e.g.\ invariant synthesis methods similar to those
described in
\cite{AbdullaHendaDelzannoRezine07,ChenHongLinRummer17,DBLP:conf/tacas/BozgaEISW20,DBLP:conf/facs2/BozgaI21},
an initially correct parametric systems can be automatically proved to
be deadlock-free, \emph{after the application of a sequence of
  reconfiguration actions}.

\vspace*{-.5\baselineskip}
\paragraph{The contributions of this paper.}
Whereas various formalisms for modeling distributed systems support
dynamic reconfiguration, the formal verification of system properties
under reconfigurations has received scant attention. We provide a
\emph{configuration logic} that specifies the safe configurations of a
distributed system. This logic is used to build Hoare-style proofs of
correctness, by annotating reconfiguration programs (i.e.\ programs
that delete and create interactions or components) with assertions
that describe both the topology of the system (i.e.\ the components
and connectors that form its coordinating architecture) and the local
states of the components. The annotations of the reconfiguration
program are proved to be valid under so-called \emph{havoc
invariants}, expressing global properties about the states of the
components, that remain, moreover, unchanged under the ongoing
interactions in the system. In order to prove these havoc invariants
for networks of any size, we develop an induction-based proof system,
that uses a parallel composition rule in the style of
assume/rely-guarantee reasoning. In contrast with existing formal
verification techniques, we do not consider the network topology to be
fixed in advance, and allow it to change dynamically, as described by
the reconfiguration program. This paper provides the details of our
proof systems and the semantics of reconfiguration programs. We
illustrate the usability of our approach by proving the correctness of
self-adjustable tree architectures \cite{Schmid16} and conclude with a
list of technical problems relevant for the automation of our
method.

\vspace*{-.5\baselineskip}
\paragraph{Main challenges.}
Formal reasoning about reconfigurable distributed systems faces two
technical challenges.
The first issue is the huge complexity of nowadays distributed
systems, that requires highly scalable proof techniques, which can
only be achieved by \emph{local reasoning}, a key ingredient of other
successful proof techniques, based on Separation Logic
\cite{OHearnReynoldsYang01}. To this end, atomic reconfiguration
commands in our proof system are specified by axioms that only refer
to the components directly involved in the action, while framing out
the rest of the distributed system. This principle sounds appealing,
but is technically challenging, as components from the local
specification interfere with components from the
frame\footnote{Essentially the equivalent of the environment in a
compositional proof system for parallel programs.}. To tackle this
issue, we assume that frames are invariant under the exchange of
messages between components (interactions) and discharge these
invariance conditions using cyclic proofs. The inference rules used to
write such proof rely on a compositional proof rule, in the spirit of
\emph{rely/assume-guarantee} reasoning
\cite{Owicki1978,DBLP:phd/ethos/Jones81}, whose assumptions about the
environment behavior are automatically synthesized from the
formul{\ae} describing the system and the environment. 

The second issue is dealing with the non-trivial interplay between
reconfigurations and interactions. Reconfigurations change the system
by adding/removing components/interactions \emph{while the system is
  running}, i.e.\ while state changes occur within components by
firing interactions. Although changes to the structure of the
distributed system seem, at first sight, orthogonal to the state
changes within components, the impact of a reconfiguration can be
immense. For instance, deleting a component holding the token in a
token-ring network yields a deadlocked system, while adding a
component with a token could lead to a data race, in which two
components access a shared resource simultaneously. Technically, this
means that a frame rule cannot be directly applied to sequentially
composed reconfigurations, as e.g.\ an arbitrary number of
interactions may fire between two atomic reconfiguration
actions. Instead, we must prove havoc invariance of the intermediate
assertions in a sequential composition of reconfiguration actions. As
an optimization of the proof technique, such costly checks do not have
to be applied along sequential compositions of reconfiguration actions
that only decrease/increase the size of the architecture; in such
monotonic reconfiguration sequences, invariance of a set of
configurations under interaction firing needs only to be checked in
the beginning (for decreasing sequences) or in the end (for increasing
sequences).

\section{A Model of Distributed Systems}

For a function $f : A \rightarrow B$, we denote by $\dom{f}$ its
domain and by $f[a \leftarrow b]$ the function that maps $a$ into $b$
and behaves like $f$ for all other elements from the domain of $f$.
By $\pow{A}$ we denote the powerset of a set $A$. For a relation $R
\subseteq A \times A$, we denote by $R^*$ its reflexive and transitive
closure. Given sets $A$ and $B$, we write $A \finsubseteq B$ if $A$ is
a finite subset of $B$ and define $A \uplus B \isdef A \cup B$ if $A
\cap B = \emptyset$ and $A \uplus B$ is undefined, if $A \cap B \neq
\emptyset$.

We model a distributed system by a finite set $\comps \finsubseteq
\universe$, where $\universe$ is a countably infinite universe of
\emph{components}. The components in $\comps$ are said to be
\emph{present} (in the system) and those from
$\universe\setminus\comps$ are \emph{absent} (from the system). 

The present components can be thought of as the nodes of a network,
each executing a copy of the same program, called \emph{behavior} in
the following. The behavior is described by a finite-state machine
$\beh = (\ports,\states,\arrow{}{})$, where $\ports$ is a finite set
of \emph{ports} i.e., the event alphabet of the machine, $\states$ is
a finite set of \emph{states}, and $\arrow{}{} \subseteq \states
\times \ports \times \states$ is a transition relation. We denote
transitions as $q \arrow{p}{} q'$ instead of $(q,p,q')$, the states
$q$ and $q'$ being referred to as the pre- and post-state of the
transition.

The network of the distributed system is described by a finite set
$\interacs \finsubseteq \universe \times \ports \times \universe
\times \ports$ of \emph{interactions}. Intuitively, an interaction
$(c_1, p_1, c_2, p_2)$ connects the port $p_1$ of component $c_1$ with
the port $p_2$ of component $c_2$, provided that $c_1$ and $c_2$ are
distinct components. Intuitively, an interaction $(c_1, p_1, c_2,
p_2)$ can be thought of as a joint execution of transitions labeled
with the ports $p_1$ and $p_2$ from the components $c_1$ and $c_2$,
respectively.


\begin{definition}\label{def:configuration}
A \emph{configuration} is a quadruple $\aconfig =
(\comps,\interacs,\statemap,\store)$, where $\comps$ and $\interacs$
describe the present components and the interactions of the system,
$\statemap: \comps \rightarrow \states$ is a \emph{state map}
associating each present component a state of the common behavior
$\beh = (\ports,\states,\arrow{}{})$ and $\store : \vars \rightarrow
\universe$ is a \emph{store} that maps variables, taken from a
countably infinite set $\vars$, to components (not necessarily
present). We denote by $\configset$ the set of configurations.
\end{definition}

\begin{example}\label{ex:configuration}
  For instance, the configuration $\config$ of the token ring system,
  depicted in Fig. \ref{fig:token-ring} (left) has present components
  $\comps = \set{c_1, \ldots, c_n}$, interactions $\interacs =
  \set{(c_i,\mathit{out},c_{(i \mod n) + 1},\mathit{in}) \mid i \in
    \interv{1}{n}}$ and state map given by $\statemap(c_1) =
  \toktoken$ and $\statemap(c_i) = \toknotok$, for $i \in
  \interv{2}{n}$. The store $\store$ is
  arbitrary. \hfill$\blacksquare$
\end{example}

Given a configuration $\config$, an interaction $(c_1, p_1, c_2, p_2)
\in \interacs$ is \emph{loose} if and only if $c_i \not\in \comps$,
for some $i = 1,2$. A configuration is \emph{loose} if and only if it
contains a loose interaction. Interactions (resp. configurations) that
are not loose are said to be \emph{tight}. In particular, loose
configurations are useful for the definition of a composition
operation, as the union of disjoint sets of components and
interactions, respectively:

\begin{definition}\label{def:composition}
  The composition of two configurations $\aconfig_i = (\comps_i,
  \interacs_i, \statemap_i, \store)$, for $i = 1,2$, is defined as
  $\aconfig_1 \comp \aconfig_2 \isdef (\comps_1 \uplus \comps_2,
  \interacs_1 \uplus \interacs_2, \statemap_1 \cup \statemap_2,
  \store)$. The composition $\aconfig_1 \comp \aconfig_2$ is undefined
  if either $\comps_1 \uplus \comps_2$ or $\interacs_1 \uplus
  \interacs_2$ is undefined\footnote{Since $\dom{\statemap_i}
    \subseteq \comps_i$, for $i=1,2$ and $\comps_1 \cap \comps_2 =
    \emptyset$, the disjointness condition is not necessary for state
    maps.}. A composition $\aconfig_1 \comp \aconfig_2$ is
  \emph{trivial} if $\comps_i = \interacs_i = \statemap_i =
  \emptyset$, for some $i = 1,2$. A configuration $\aconfig_2$ is a
  \emph{subconfiguration} of $\aconfig_1$, denoted $\aconfig_1
  \substreq \aconfig_2$, if and only if there exists a configuration
  $\aconfig_3 \in \configset$, such that $\aconfig_1 = \aconfig_2
  \comp \aconfig_3$.
\end{definition}
Note that a tight configuration may be the result of composing two
loose configurations, whereas the composition of tight configurations
is always tight. The example below shows that, in most cases, a
non-trivial decomposition of a tight configuration necessarily
involves loose configurations.

\begin{example}\label{ex:composition}
  Let $\aconfig_i = (\comps_i, \interacs_i, \statemap_i, \store)$,
  where $\comps_i = \set{c_i}$, $\interacs_i = \set{(c_i,
    \mathit{out}, c_{(i \mod 3) + 1}, \mathit{in})}$, for all $i \in
  \interv{1}{3}$, $\statemap_1(c_1)=\statemap_2(c_2)=\toknotok$ and
  $\statemap_3(c_3)=\toktoken$. Then $\aconfig \isdef \aconfig_1 \comp
  \aconfig_2 \comp \aconfig_3$ is the configuration from the top-left
  corner of Fig. \ref{fig:reconfiguration}, where the store $\store$
  is arbitrary. Note that $\aconfig_1, \aconfig_2$, and $\aconfig_3$
  are loose, respectively, but $\aconfig$ is tight. Moreover, the only
  way of decomposing $\aconfig$ into two tight subconfigurations
  $\aconfig'_1$ and $\aconfig'_2$ is taking $\aconfig'_1 \isdef
  \aconfig$ and $\aconfig'_2 \isdef (\emptyset, \emptyset, \emptyset,
  \store)$, or viceversa. \hfill$\blacksquare$
\end{example}

\begin{figure}
  \vspace*{-\baselineskip}
  \caption{Havoc and Reconfigurations of a Token Ring}
  \label{fig:reconfiguration}
  \begin{center}
    \scalebox{.6}[.5]{\input{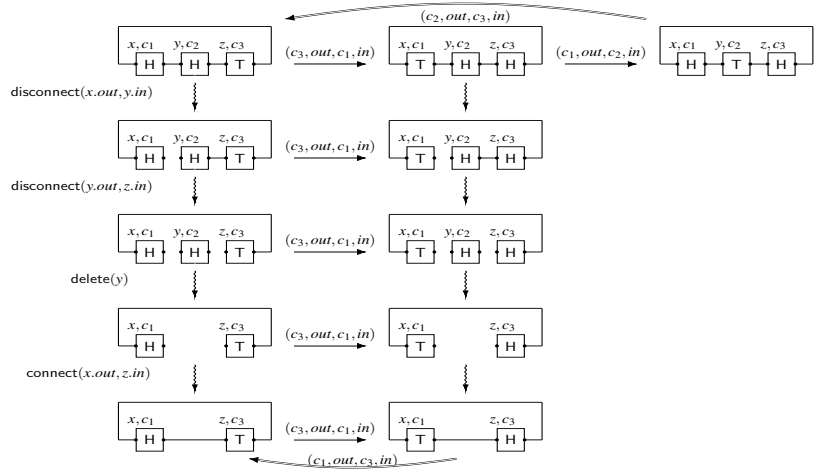}}
  \end{center}
  \vspace*{-\baselineskip}
\end{figure}

A configuration is changed by two types of \emph{actions}: (a)
\emph{havoc} actions change the local states of the components by
executing interactions (that trigger simultaneous transitions in
different components), without changing the structure or the store,
and (b) \emph{reconfiguration} actions that change the structure,
store and possibly the state map of a configuration. We refer to
Fig. \ref{fig:reconfiguration} for a depiction of havoc and
reconfiguration actions. Each havoc action is the result of executing
a sequence of interactions (horizontally depicted using straight
double arrows), whereas each reconfiguration action (vertically
depicted using snake-shaped arrows) corresponds to a statement in a
reconfiguration program. The two types of actions may interleave,
yielding a transition graph with a finite but unbounded (parametric)
or even infinite (obtained by iteratively adding new components) set
of vertices (configurations).

Formally, an \emph{action} is a function $f : \configset \rightarrow
\pow{\configset}^\errconfigs$, where $\pow{\configset}^\errconfigs
\isdef \pow{\configset} \cup \set{\errconfigs}$.  The complete lattice
$(\pow{\configset},\subseteq,\cup,\cap)$ is extended with a greatest
element $\errconfigs$, with the conventions $S \cup \errconfigs \isdef
\errconfigs$ and $S \cap \errconfigs \isdef S$, for each $S \in
\pow{\configset}$. We consider that an action $f$ is \emph{disabled}
in a configuration $\aconfig$ if and only if $f(\aconfig) = \emptyset$
and that it \emph{faults} in $\aconfig$ if and only if $f(\aconfig) =
\errconfigs$. Actions are naturally lifted to sets of configurations
as $f(S) \isdef \bigcup_{\aconfig\in S} f(\aconfig)$, for each $S
\subseteq \configset$.

\begin{definition}\label{def:havoc}
  The \emph{havoc} action $\ahavoc : \configset \rightarrow
  \pow{\configset}$ is defined as $\ahavoc(\aconfig) \isdef
  \set{\aconfig' \mid \aconfig \Arrow{}{}^* \aconfig'}$, where
  $\Arrow{}{}^*$ is the reflexive and transitive closure of the
  relation $\Arrow{}{} ~\subseteq \configset \times \configset$,
  defined by the following rule:
  \begin{prooftree}
  \AxiomC{$\begin{array}{ccccc}
      (c_1, p_1, c_2, p_2) \in \interacs & 
      c_1, c_2 \in \comps & \statemap(c_i) = q_i & q_i \arrow{p_i}{} q'_i, & \text{for all } i = 1,2 
    \end{array}$}
  \LeftLabel{(\havocrule)}
  \UnaryInfC{$(\comps,\interacs,\statemap,\store) \Arrow{}{}
    (\comps,\interacs,\statemap[c_1 \leftarrow q'_1][c_2 \leftarrow q'_2],\store)$}
  \end{prooftree}
\end{definition}
\noindent
Note that the havoc action is the result of executing any sequence of
tight interactions, whereas loose interactions are simply ignored.
The above definition can be generalized to multi-party interactions
$(c_1, p_1, \ldots, c_n, p_n)$ with $n\geq1$ pairwise distinct
participant components $c_1, \ldots, c_n$, that fire simultaneously
transitions of the behavior labeled with the ports $p_1, \ldots, p_n$,
respectively. In particular, the interactions of arity $n=1$
correspond to the local (silent) actions performed independently by a
single component. To keep the presentation simple, we refrain from
considering such generalizations, for the time being.

\begin{example}\label{ex:havoc}
  Let $\aconfig_i = (\set{c_1, c_2, c_3}, \set{(c_i, \mathit{out},
    c_{i \mod 3 + 1}, \mathit{in}) \mid i \in \interv{1}{3}},
  \statemap_i, \store)$, for $i \in \interv{1}{3}$ be the top-most
  configurations from Fig. \ref{fig:reconfiguration}, where
  $\statemap_1(c_1) = \statemap_1(c_2) = \toknotok$, $\statemap_1(c_3)
  = \toktoken$, $\statemap_2(c_1) = \toktoken$, $\statemap_2(c_2) =
  \statemap_2(c_3) = \toknotok$, $\statemap_3(c_1) = \statemap_3(c_3)
  = \toknotok$, $\statemap_3(c_2) = \toktoken$ and $\store(x) = c_1$,
  $\store(y) = c_2$, $\store(z) = c_3$. Then $\ahavoc(\aconfig_i) =
  \set{\aconfig_1, \aconfig_2, \aconfig_3}$, for all $i \in
  \interv{1}{3}$. \hfill$\blacksquare$
\end{example}

\subsection{The Expressiveness of the Model}

Before moving on with the definition of a logic describing sets of
configurations (\S\ref{sec:adl}), a reconfiguration language and a
proof system for reconfiguration programs (\S\ref{sec:rpl}), we
discuss the expressive power of the
\emph{components-behavior-interactions} model of distributed systems
introduced so far, namely \emph{what kinds of distributed algorithms
  can be described in our model}?

On one hand, this model can describe message-passing algorithms on
networks with unrestricted topologies (pipelines, rings, stars, trees,
grids, cliques, etc.), such as flooding/notification of a crowd,
token-based mutual exclusion, deadlock problems (dining
philosophers/cryptographers), etc. Moreover, the model captures
asynchronous communication, via bounded message channels modeled using
additional components\footnote{The number of messages in transit
  depends on the number of states in the behavior; unbounded message
  queues would require an extension of the model to infinite-state
  behaviors.}. Furthermore, \emph{transient faults} (process delays,
message losses, etc.) can be modeled as well, by nondeterministic
transitions e.g., a channel component might chose to
nondeterministically lose a message. In particular, having a single
finite-state machine that describes the behavior of all components is
not a limitation, because finitely many behaviors $\beh_1, \ldots,
\beh_m$ can be represented by state machines with disjoint transition
graphs, the state map distinguishing between different behavior types
-- if $\statemap(c) = q$ and $q$ is a state of $\beh_i$, the value of
$\statemap(c)$ can never change to a state of a different behavior
$\beh_j$, as the result of a havoc action.

On the other hand, the current model cannot describe complex
distributed algorithms, such as leader election
\cite{10.1145/359104.359108,DOLEV1982245}, spanning tree
\cite{6773228,10.2307/2033241}, topological linearization
\cite{145da081119e454586debadaade17bf0}, Byzantine consensus
\cite{10.1145/357172.357176} or Paxos parliament
\cite{10.1145/279227.279229}, due to the following
limitations: \begin{itemize}
\item Finite-state behavior is oblivious of the identity of the
  components (processes) in distributed systems of arbitrary
  sizes. For instance, there is no distributed algorithm over rings
  that can elect a leader under the assumption of anonymous processes
  \cite{10.1145/359104.359108,DOLEV1982245}.
\item Interactions between a bounded number of participants cannot
  describe broadcast between arbitrarily many components, as in most
  common consensus algorithms
  \cite{10.1145/357172.357176,10.1145/279227.279229}.
\end{itemize}
We proceed in the rest of the paper under these simplifying
assumptions (i.e., finite-state behavior and bounded-arity
interactions), as our focus is modeling the reconfiguration aspect of
a distributed system, and consider the following extensions for future
work: \begin{itemize}
\item \emph{Identifiers in registers}: the behavior is described by a
  finite-state machine equipped with finitely many registers $r$
  holding component identifiers, that can be used to send ($p!r$) and
  receive ($p?r$) identifiers ($p$ stands for a port name), perform
  equality ($r=r'$) and strict inequality ($r < r'$) checks, with no
  other relation or function on the domain of identifiers. For
  instance, identifier-aware behaviors are considered in
  \cite{DBLP:conf/concur/AiswaryaBG15} in the context of bounded model
  checking i.e., verification of temporal properties (safety and
  liveness) under the assumption that the system has a ring topology
  and proceeds in a bounded number of rounds (a round is completed
  when every component has executed exactly one
  transition). Algorithms running on networks of arbitrary topologies
  (described by graphs) are modeled using \emph{distributed register
  automata} \cite{10.1007/978-3-030-17127-8_7}, that offer a promising
  lead for verifying properties of distributed systems with mutable
  networks.
\item \emph{Broadcast interactions}: interactions involving an
  unbounded number of component-port pairs e.g., the $p_0$ ports of
  all components except for a bounded set $c_1, \ldots, c_k$, that
  interact with ports $p_1, \ldots, p_k$, for a given integer constant
  $k\geq0$. Broadcast interactions are described using universal
  quantifiers in \cite{DBLP:conf/tacas/BozgaEISW20}, where network
  topologies are specified using first-order logic. To accomodate
  broadcast interactions in our model, one has to redefine
  composition, by considering e.g., glueing of interactions, in
  addition to the disjoint union of configurations
  (Def. \ref{def:composition}). Changing this definition would have a
  non-trivial impact on the configuration logic used to write
  assertions in reconfiguration proofs (\S\ref{sec:adl}).
\end{itemize}
Considering a richer model of behavior (e.g., register automata, timed
automata, or even Markov decision processes) would impact mainly the
part of the framework that deals with checking the properties (i.e.,
safety, liveness or havoc invariance) of a set of configurations
described by a formula of the configuration logic (defined in
\S\ref{sec:adl}) but should not, in principle, impact the
configuration logic itself, the programming language or the proof
system for reconfiguration programs (defined in
\S\ref{sec:rpl}). However, accomodating broadcast communication
requires changes at the level of the logic and, consequently, the
reconfiguration programs proof system.

\section{A Logic of Configurations}
\label{sec:adl}

We define a \emph{Configuration Logic} (\adl) that is, an assertion
language describing sets of configurations. Let $\preds$ be a
countably infinite set of predicate symbols, where $\#(\apred) \geq 1$
denotes the arity of a predicate symbol $\apred \in
\preds$. The \adl\ formul{\ae} are inductively described by the following
syntax:
\[\begin{array}{rcl}
\phi & ::= & \predtrue \mid \emp \mid x = y \mid \compin{x}{q} \mid
\interac{x_1}{p_1}{x_2}{p_2} \mid \apred(x_1, \ldots, x_{\#(\apred)}) \mid
\phi * \phi \mid \phi \wedge \phi \mid \neg \phi \mid \exists x ~.~ \phi
\end{array}\]
where $q \in \states$, $\apred\in\preds$ are predicate symbols and
$x,y,x_1,x_2, \ldots\in\vars$ are variables. The atomic formul{\ae}
$\compin{x}{q}$, $\interac{x_1}{p_1}{x_2}{p_2}$, and $\apred(x_1,
\ldots, x_{\#(\apred)})$ are called \emph{component},
\emph{interaction} and \emph{predicate atoms}, respectively. A formula
is said to be \emph{predicate-free} if it has no occurrences of
predicate atoms.  By $\fv{\phi}$ we denote the set of free variables
in $\phi$, that do not occur within the scope of an existential
quantifier. A formula is \emph{quantifier-free} if it has no
occurrence of existential quantifiers. A \emph{substitution} is a
partial mapping $\sigma : \vars \rightarrow \vars$ and the formula
$\phi\sigma$ is the result of replacing each free variable $x \in
\fv{\phi} \cap \dom{\sigma}$ by $\sigma(x)$ in $\phi$. We denote by
$[x_1/y_1, \ldots, x_k/y_k]$ the substitution that replaces $x_i$ with
$y_i$, for all $i \in \interv{1}{k}$. We use the shorthands
$\predfalse \isdef \neg\predtrue$, $x \neq y \isdef \neg x = y$,
$\phi_1 \vee \phi_2 \isdef \neg(\neg\phi_1 \wedge \neg\phi_2)$,
$\forall x ~.~ \phi_1 \isdef \neg(\exists x ~.~ \neg\phi_1)$ and
$\company{x} \isdef \bigvee_{q \in \states} \compin{x}{q}$.

We distinguish the boolean ($\wedge$) from the separating ($*$)
conjunction: $\phi_1 \wedge \phi_2$ means that $\phi_1$ and $\phi_2$
hold for the same configuration, whereas $\phi_1 * \phi_2$ means that
$\phi_1$ and $\phi_2$ hold separately, on two disjoint parts of the
same configuration. Intuitively, a formula $\emp$ describes empty
configurations, with no components and interactions, $\compin{x}{q}$
describes a configuration with a single component, given by the store
value of $x$, in state $q$, and $\interac{x_1}{p_1}{x_2}{p_2}$
describes a single interaction between ports $p_1$ and $p_2$ of the
components given by the store values of $x_1$ and $x_2$,
respectively. The formula $\compin{x_1}{q_1} * \ldots *
\compin{x_n}{q_n} * \interac{x_1}{p_1}{x_2}{p_2} * \ldots *
\interac{x_{n-1}}{p_{n-1}}{x_n}{p_n}$ describes a structure consisting
of $n$ \emph{pairwise distinct} components, in states $q_1, \ldots,
q_n$, respectively, joined by interactions between ports $p_i$ and
$p_{i+1}$, respectively, for all $i \in \interv{1}{n-1}$.

The \adl\ logic is used to describe configurations of distributed
systems of unbounded size, by means of predicate symbols, defined
inductively by a given set of rules. For reasons related to the
existence of (least) fixed points, the definitions of predicates are
given in a restricted fragment of the logic. The \emph{symbolic
configurations} are formul{\ae} of the form $\spaceform \wedge
\pureform$, where $\spaceform$ and $\pureform$ are defined by the
following syntax:
\[\begin{array}{lcl}
\spaceform ::= \emp \mid \compin{x}{q} \mid \interac{x_1}{p_1}{x_2}{p_2} 
\mid \apred(x_1, \ldots, x_{\#(\apred)}) \mid \spaceform * \spaceform && 
\pureform ::= x = y \mid x \neq y \mid \pureform \wedge \pureform
\end{array}\]

The interpretation of \adl\ formul{\ae} is given by a semantic
relation $\models_\asid$, parameterized by a finite \emph{set of
inductive definitions} (SID) $\asid$, consisting of rules $\apred(x_1,
\ldots, x_{\#(\apred)}) \unfoldrule \exists y_1 \ldots y_k ~.~ \phi$,
where $\phi$ is a symbolic configuration, such that $\fv{\phi}
\subseteq \set{x_1, \ldots, x_{\#(\apred)}} \cup \set{y_1, \ldots,
  y_k}$. The relation $\models_\asid$ is defined inductively on the
structure of formul{\ae}, as follows:
\[\begin{array}{rclcl}
\config & \models_\asid & \predtrue & \iff & \text{true}
\\
\config & \models_\asid & \emp & \iff & \comps = \emptyset \text{ and } \interacs = \emptyset 
\\
\config & \models_\asid & x = y & \iff & \store(x) = \store(y)
\\
\config & \models_\asid & \compin{x}{q} & \iff & \comps = \set{\store(x)},~ \interacs = \emptyset \text{ and } \statemap(\store(x)) = q
\\
\config & \models_\asid & \interac{x_1}{p_1}{x_2}{p_2} & \iff & \comps = \emptyset,~ \interacs = \set{(\store(x_1), p_1, \store(x_2), p_2)} 
\\
\config & \models_\asid & \apred(y_1, \ldots, y_{\#(\apred)}) & \iff &
\config \models_\asid \varphi[x_1/y_1, \ldots, x_{\#(\apred)}/y_{\#(\apred)}] \text{, for some} \\
&&&& \text{rule } \apred(x_1, \ldots, x_{\#(\apred)}) \unfoldrule \varphi \text{ from $\asid$}
\\
\config & \models_\asid & \phi_1 * \phi_2 & \iff & \text{there exist configurations } 
\aconfig_1 \text{ and } \aconfig_2, \text{ such that } \\
&&&& \config = \aconfig_1 \comp \aconfig_2 \text{ and } \aconfig_i \models_\asid \phi_i \text{, for both } i = 1,2
\\
\config & \models_\asid & \phi_1 \wedge \phi_2 & \iff & \config \models_\asid \phi_i \text{, for both } i = 1,2
\\
\config & \models_\asid & \neg\phi_1 & \iff & \text{not } \config \models_\asid \phi_1
\\
\config & \models_\asid & \exists x ~.~ \phi_1 & \iff & (\comps,\interacs,\statemap,\store[x\leftarrow c],\statemap) \models_\asid \phi_1
\text{, for some $c \in \universe$}
\end{array}\]
From now on, we consider the SID $\asid$ to be clear from the context
and write $\aconfig \models \phi$ instead of $\aconfig \models_\asid
\phi$. If $\aconfig \models \phi$, we say that $\aconfig$ is a
\emph{model} of $\phi$ and define the set of models of $\phi$ as
$\sem{\phi}{} \isdef \set{\aconfig \mid \aconfig \models \phi}$. A
formula $\phi$ is \emph{satisfiable} if and only if $\sem{\phi}{} \neq
\emptyset$. Given formul{\ae} $\phi$ and $\psi$, we say that
\emph{$\phi$ entails $\psi$} if and only if $\sem{\phi}{} \subseteq
\sem{\psi}{}$, written $\phi \models \psi$.

\begin{example}\label{ex:chain}
  The SID below defines chains of components and interactions, with at
  least $h, t\in\N$ components in state $\toknotok$ and $\toktoken$,
  respectively:
  \[\begin{array}{ll}
  \chain{0}{1}(x,x) \unfoldrule \compin{x}{\toktoken} &
  \chain{h}{t}(x,y) \unfoldrule \exists z.~\compin{x}{\toktoken}
  * \interac{x}{out}{z}{in} * \chain{h}{t\dot{-}1}(z,y) \\  
  \chain{1}{0}(x,x) \unfoldrule \compin{x}{\toknotok} & 
  \chain{h}{t}(x,y) \unfoldrule \exists z.~\compin{x}{\toknotok} * \interac{x}{out}{z}{in}
  * \chain{h\dot{-}1}{t}(z,y) \\
  \chain{0}{0}(x,x) \unfoldrule \company{x} 
  \end{array}\]
  where $k\dot{-}1 \isdef \max(k-1,0)$, for all $k\in\N$. The
  configurations $(\{c_1, \ldots, c_n\}, \{(c_i,\mathit{out},c_{(i
    \mod n) + 1},\mathit{in}) \mid i \in \interv{1}{n}\}, \statemap,
  \store)$ from Example~\ref{ex:configuration} are models of the
  formula $\exists x \exists y ~.~ \chain{0}{0}(x,y) *
  \interac{y}{out}{x}{in}$, for all $n \in \N$. This is because any
  such configuration can be decomposed into a model of
  $\interac{y}{out}{x}{in}$ and a model of $\chain{0}{0}(x,y)$.  The
  latter is either a model of $\company{x}$, matching the body of the
  rule $\chain{0}{0}(x,x) \unfoldrule \company{x}$ if $x=y$, or a
  model of $\exists z.~\company{x} * \interac{x}{out}{z}{in} *
  \chain{0}{0}(z,y)$, matching the body of the rule $\chain{0}{0}
  \unfoldrule \exists z.~\company{x} * \interac{x}{out}{z}{in} *
  \chain{0}{0}(z,y)$ etc. \hfill$\blacksquare$
\end{example}

\subsection{The Expressiveness of \adl}

The \adl\ logic is quite expressive, due to the interplay between
first-order quantifiers and inductively defined predicates. For
instance, the class of \emph{cliques}, in which there is an
interaction between the $\mathit{out}$ and $\mathit{in}$ ports of any
two present components are defined by the formula: \[\forall x \forall
y ~.~ \company{x} * \company{y} * \predtrue \rightarrow
\interac{x}{out}{y}{in} * \predtrue\] Describing clique-structured
networks is important for modeling consensus protocols, such as
Byzantine \cite{10.1145/357172.357176} or Paxos
\cite{10.1145/279227.279229}.

\emph{Connected} networks, used in e.g., linearization algorithms
\cite{tocs13}, are such that there exists a path of interactions
between each two present components in the system: \[\forall x \forall
y ~.~ \company{x} * \company{y} * \predtrue \rightarrow
\mathit{reach}(x,y)\] where the predicate $\mathit{reach}(x,y)$ is
defined by the following rules:
\[\begin{array}{ll}
\mathit{reach}(x,y) \leftarrow x = y, & 
\mathit{reach}(x,y) \leftarrow \exists z ~.~ \company{x} *
\interac{x}{out}{z}{in} * \mathit{reach}(z,y) * \predtrue
\end{array}\]

A \emph{grid} is a connected network that, moreover, satisfies the
following formula:
\[\forall x \forall y \forall z ~.~ \interac{x}{out}{y}{in} * \interac{x}{out}{z}{in} * \predtrue
\rightarrow \exists u ~.~ \interac{y}{out}{u}{in} *
\interac{z}{out}{u}{in} * \mathit{present}(x,y,z,u) \] where
$\mathit{present}(x_1, \ldots, x_n) \isdef \Asterisk_{1 \leq i \leq n}
~\company{x_i}$ states that the components $x_1, \ldots, x_n$ are
present and pairwise distinct. Grids are important in modeling
distributed scientific computing \cite{Foster02}.

As suggested by work on Separation Logic
\cite{DBLP:conf/fossacs/DemriLM18}, the price to pay for this
expressivity is the inherent impossibility of having decision
procedures for a fragment of \adl, that combines first-order
quantifiers with inductively defined predicates. A non-trivial
fragment of \adl\ that has decision procedures for satisfiability and
entailment is the class of symbolic configurations
\cite{BozgaBueriIosif2022Arxiv}. However, we do not expect to describe
systems with clique or grid network topologies using symbolic
configurations. Furthermore, we conjecture that these classes are
beyond the expressiveness of the symbolic configuration fragment. We
discuss these issues in more detail in \S\ref{sec:open-problems}.

\section{A Language for Programming Reconfigurations}
\label{sec:rpl}

This section defines \emph{reconfiguration} actions that change the
structure of a configuration. We distinguish between reconfigurations
and havoc actions (Def. \ref{def:havoc}), that change configurations
in orthogonal ways (see Fig. \ref{fig:reconfiguration} for an
illustration of the interplay between the two types of actions). The
reconfiguration actions are the result of executing a given
reconfiguration program on the distributed system at hand. This
section presents the syntax and operational semantics of the
reconfiguration language. Later on, we introduce a Hoare-style proof
system to reason about the correctness of reconfiguration programs.

\subsection{Syntax and Operational Semantics}

Reconfiguration programs, ranged over by $\acomm$, are inductively
defined by the following syntax:
\[\begin{array}{rcl}
\acomm & ::= & \new(q,x) \mid \delete(x) \mid \connect(x_1.p_1,x_2.p_2) \mid
\disconnect(x_1.p_1,x_2.p_2)
\\
&& \mid \withcomm{x_1, \ldots, x_k}{\theta}{\acomm_1} \mid
\acomm_1; \acomm_2 \mid \acomm_1+\acomm_2 \mid \acomm_1^* 
\end{array}\]
where $q\in\states$ is a state, $x,x_1,x_2, \ldots \in \vars$ are
program variables and $\theta$ is a \emph{predicate-free}
\emph{quantifier-free} formula of the \adl\ logic, called a
\emph{trigger}.

The \emph{primitive commands} are $\new(q,x)$ and $\delete(x)$, that
create and delete a component (the newly created component is set to
execute from state $q$) given by the store value of $x$,
$\connect(x_1.p_1, x_2.p_2)$ and $\disconnect(x_1.p_1, x_2.p_2)$, that
create and delete an interaction, between the ports $p_1$ and $p_2$ of
the components given by the store values of $x_1$ and $x_2$,
respectively. We denote by $\primcomms$ the set of primitive commands.

A \emph{conditional} is a program of the form $(\withcomm{x_1, \ldots,
  x_k}{\theta}{\acomm})$ that performs the following steps, \emph{with
no havoc action (Def. \ref{def:havoc}) in between the first and second
steps} below: \begin{enumerate}
\item maps the variables $x_1, \ldots, x_k$ to some components $c_1,
  \ldots, c_k \in \universe$ such that the configuration after the
  assignment contains a model of the trigger $\theta$; the conditional
  is disabled if the current configuration is not a model of $\exists
  x_1 \ldots \exists x_k ~.~ \theta * \predtrue$,
\item launches the first command of the program $\acomm$ on this
  configuration, and
\item continues with the remainder of $\acomm$, in interleaving with
  havoc actions;
\item upon completion of $\acomm$, the values of $x_1,
  \ldots, x_k$ are forgotten.
\end{enumerate}
To avoid technical complications, we assume that nested conditionals
use pairwise disjoint tuples of variables; every program can be
statically changed to meet this condition, by renaming variables. Note
that the trigger $\theta$ of a conditional $(\withcomm{x_1, \ldots,
  x_k}{\theta}{\acomm})$ has no quantifiers nor predicate atoms, which
means that the overall number of components and interactions in a
model of $\theta$ is polynomially bounded by the size of (number of
symbols needed to represent) $\theta$. Intuitively, this means that
the part of the system (matched by $\theta$) to which the
reconfiguration is applied is relatively small, thus the procedure
that evaluates the trigger can be easily implemented in a distributed
environment, as e.g., consensus between a small number of neighbouring
components.

The sequential composition $\acomm_1; \acomm_2$ executes $\acomm_1$
followed by $\acomm_2$, with zero or more interactions firing in
between. This is because, even though being sequential, a
reconfiguration program runs in parallel with the state changes that
occur as a result of firing the interactions. Last, $\acomm_1 +
\acomm_2$ executes either $\acomm_1$ or $\acomm_2$, and $\acomm^*$
executes $\acomm$ zero or more times in sequence,
nondeterministically.


It is worth pointing out that the reconfiguration language does not
have explicit assignments between variables. As a matter of fact, the
conditionals are the only constructs that nondeterministically bind
variables to indices that satisfy a given logical constraint. This
design choice sustains the view of a distributed system as a cloud of
components and interactions in which reconfigurations can occur
anywhere a local condition is met. In other words, we do not need
variable assignments to traverse the architecture --- the program
works rather by identifying a part of the system that matches a small
pattern, and applying the reconfiguration locally to that
subsystem. For instance, a typical pattern for writing reconfiguration
programs is $(\withcomm{\vec{x}_1}{\theta_1}{\acomm_1} + \ldots +
\withcomm{\vec{x}_k}{\theta_k}{\acomm_k})^*$, where $\acomm_1, \ldots,
\acomm_k$ are loop-free sequential compositions of primitive
commands. This program continuously choses a reconfiguration sequence
$\acomm_i$ nondeterministically and either applies it on a small part
of the configuration that satisfies $\theta_i$, or does nothing, if no
such subconfiguration exists within the current configuration.

\begin{figure}[t!]
\vspace*{-\baselineskip}
\caption{Operational Semantics of the Reconfiguration Language}
\label{fig:os}
{\footnotesize\begin{center}
  \begin{prooftree}
    \AxiomC{$c\in\universe\smallsetminus\comps$}    
    \UnaryInfC{$\succj{\new(q,x)}{\config}{(\comps \cup \set{c}, \interacs, \statemap[c \leftarrow q], \store[x \leftarrow c])}$}
  \end{prooftree}

  \begin{minipage}{.55\textwidth}
    \begin{prooftree}
      \AxiomC{$\nu(x) \in \comps$} 
      \UnaryInfC{$\succj{\delete(x)}{\config}{
          (\comps\smallsetminus\set{\nu(x)}, \interacs, \statemap, \store)}$}
    \end{prooftree}
  \end{minipage}
  \begin{minipage}{.4\textwidth}
    \begin{prooftree}
      \AxiomC{$\nu(x) \not\in \comps$}
      \UnaryInfC{$\abort{\delete(x)}{\config}$}
    \end{prooftree}
  \end{minipage}

  \begin{prooftree}
    \AxiomC{}
    \UnaryInfC{$\succj{\connect(x_1.p_1, x_2.p_2)}{\config}{(\comps,\interacs\cup\set{(\nu(x_1), p_1, \nu(x_2), p_2)},\statemap,\store)}$}
  \end{prooftree}

  \begin{prooftree}
    \AxiomC{$(\nu(x_1), p_1, \nu(x_2), p_2) \in \interacs$} 
    \UnaryInfC{$\succj{\disconnect(x_1.p_1, x_2.p_2)}{\config}{
        (\comps,\interacs \smallsetminus \set{(\nu(x_1), p_1, \nu(x_2), p_2)},\statemap,\store)}$}
  \end{prooftree}

  \begin{prooftree}
    \AxiomC{$(\nu(x_1), p_1, \nu(x_2), p_2) \not\in \interacs$}
    \UnaryInfC{$\abort{\disconnect(x_1.p_1, x_2.p_2)}{\config}$}
  \end{prooftree}

  \begin{prooftree}
    \AxiomC{$\begin{array}{l}
        c_1,c'_1,\ldots,c_k,c'_k\in\universe \hspace*{4mm}
        (\comps,\interacs,\statemap,\store[x_1\leftarrow c_1,\ldots,x_k\leftarrow c_k]) \models \varphi * \predtrue \\
        \succj{\acomm}{(\comps,\interacs,\statemap,\store[x_1\leftarrow c_1,\ldots,x_k\leftarrow c_k])}{(\comps',\interacs',\statemap',\store')}
      \end{array}$}
    \UnaryInfC{$\succj{\withcomm{x_1, \ldots, x_k}{\varphi}{\acomm}}{\config}{(\comps',\interacs',\statemap',\store'[x_1\leftarrow c'_1,\ldots,x_k\leftarrow c'_k])}$}
  \end{prooftree}

  \begin{minipage}{.43\textwidth}
  \begin{prooftree}
    \AxiomC{$\begin{array}{c}\succj{\acomm_1}{\aconfig}{\aconfig_0}$ \quad $\aconfig_1 \in \ahavoc(\aconfig_0) \end{array}$
    \quad $\begin{array}{c}\succj{\acomm_2}{\aconfig_1}{\aconfig'} \end{array}$}
    \UnaryInfC{$\succj{\acomm_1;\acomm_2}{\aconfig}{\aconfig'}$}
  \end{prooftree}
  \end{minipage}
  \begin{minipage}{.2\textwidth}
  \begin{prooftree}
    \AxiomC{$\begin{array}{c}\succj{\acomm_1}{\aconfig}{\aconfig'} \end{array}$}
    \UnaryInfC{$\succj{\acomm_1+\acomm_2}{\aconfig}{\aconfig'}$}
  \end{prooftree}
  \end{minipage}
  \begin{minipage}{.36\textwidth} 
    \begin{prooftree}
      \AxiomC{$\succj{\acomm^n}{\aconfig}{\aconfig'}$}
      \RightLabel{, $\acomm^n = \left\{\begin{array}{ll}
        \acomm^{n-1}; \acomm & \text{if } n \geq 1 \\
        \skipcomm & \text{if } n = 0
        \end{array}\right.$}
      \UnaryInfC{$\succj{\acomm^*}{\aconfig}{\aconfig'}$}
    \end{prooftree}
  \end{minipage}
\end{center}}
\vspace*{-\baselineskip}
\end{figure}

The operational semantics of reconfiguration programs is given by the
structural rules in Fig. \ref{fig:os}, that define the judgements
$\succj{\acomm}{\aconfig}{\aconfig'}$ and $\abort{\acomm}{\aconfig}$,
where $\aconfig$ and $\aconfig'$ are configurations and $\acomm$ is a
program. Intuitively, $\succj{\acomm}{\aconfig}{\aconfig'}$ means that
$\aconfig'$ is a successor of $\aconfig$ following the execution of
$\acomm$ and $\abort{\acomm}{\aconfig}$ means that $\acomm$ faults in
$\aconfig$. The semantics of a program $\acomm$ is the action
$\semcomm{\acomm} : \configset \rightarrow
\pow{\configset}^\errconfigs$, defined as \ifLongVersion:
\[\semcomm{\acomm}(\aconfig) \isdef \left\{\begin{array}{ll}
\top & \text{ if } \abort{\acomm}{\aconfig} \\
\set{\aconfig' \mid \succj{\acomm}{\aconfig}{\aconfig'}} & \text{ otherwise}
\end{array}\right.\]
\else $\semcomm{\acomm}(\aconfig) \isdef \top$, if
$\abort{\acomm}{\aconfig}$ and $\semcomm{\acomm}(\aconfig) \isdef
\set{\aconfig' \mid \succj{\acomm}{\aconfig}{\aconfig'}}$, otherwise.
\fi The only primitive commands that may fault are $\delete(x)$ and
$\disconnect(x_1.p_1, x_2.p_2)$; for both, the premisses of the faulty
rules are disjoint from the ones for normal termination, thus the
action $\semcomm{\acomm}$ is properly defined for all programs
$\acomm$. Notice that the rule for sequential composition uses the
havoc action $\ahavoc$ in the premiss, thus capturing the interleaving
of havoc state changes and reconfiguration actions.


\subsection{Reconfiguration Proof System}

To reason about the correctness properties of reconfiguration
programs, we introduce a Hoare-style proof system consisting of a set
of axioms that formalize the primitive commands
(Fig. \ref{fig:hoare}a), a set of inference rules for the composite
programs (Fig. \ref{fig:hoare}b) and a set of structural rules
(Fig. \ref{fig:hoare}c). The judgements are Hoare triples
$\hoare{\phi}{\acomm}{\psi}$, where $\phi$ and $\psi$ (called pre- and
postcondition, respectively) are \adl\ formul{\ae}. The triple
$\hoare{\phi}{\acomm}{\psi}$ is \emph{valid}, written $\models
\hoare{\phi}{\acomm}{\psi}$, if and only if
$\semcomm{\acomm}(\sem{\phi}{}) \subseteq \sem{\psi}{}$. Note that a
triple is valid only if the program does not fault on any model of the
precondition. In other words, an invalid Hoare triple
$\hoare{\phi}{\acomm}{\psi}$ cannot distinguish between
$\semcomm{\acomm}(\sem{\phi}{}) \not\subseteq \sem{\psi}{}$
(non-faulting incorrectness) and $\semcomm{\acomm}(\sem{\phi}{}) =
\errconfigs$ (faulting).

The axioms (Fig. \ref{fig:hoare}a) give the \emph{local
specifications} of the primitive commands in the language by Hoare
triples whose preconditions describe only those resources (components
and interactions) necessary to avoid faulting. In particular,
$\delete(x)$ and $\disconnect(x_1.p_1, x_2.p_2)$ require a single
component $\company{x}$ and an interaction
$\interac{x_1}{p_1}{x_2}{p_2}$ to avoid faulting, respectively. The
rules for sequential composition and iteration (Fig \ref{fig:hoare}b)
use the following semantic side condition, based on the havoc action
(Def. \ref{def:havoc}):
\begin{definition}\label{def:havoc-invariant}
  A formula $\phi$ is \emph{havoc invariant} if and only if
  $\ahavoc(\sem{\phi}{}) \subseteq \sem{\phi}{}$.
\end{definition}
Note that the dual inclusion $\sem{\phi}{} \subseteq
\ahavoc(\sem{\phi}{})$ always holds, because $\ahavoc$ is the
reflexive and transitive closure of the $~\Arrow{}{}~$ relation
(Def. \ref{def:havoc}). Since havoc invariance is required to prove
the validity of Hoare triples involving sequential composition, it is
important to have a way of checking havoc invariance. We describe a
proof system for such havoc queries in \S\ref{sec:havoc}.  Moreover,
the side condition of the \emph{consequence rule}
(Fig. \ref{fig:hoare}c left) consists of two entailments, that are
discharged by an external decision procedure (discussed in
\S\ref{sec:open-problems}).

\begin{figure}[t!]
  \vspace*{-\baselineskip}
  \caption{Proof System for the Reconfiguration Language}
  \label{fig:hoare}
          {\footnotesize\begin{center}

            \begin{minipage}{.45\textwidth}
            \begin{prooftree}
              \AxiomC{}
              \UnaryInfC{$\hoare{\emp}{\new(q,x)}{\compin{x}{q}}$}
            \end{prooftree}
            \end{minipage}
            \begin{minipage}{.45\textwidth}
              \begin{prooftree}
                \AxiomC{}
                \UnaryInfC{$\hoare{\company{x}}{\delete(x)}{\emp}$}
              \end{prooftree}
            \end{minipage}
            \vspace*{.5\baselineskip}
            
            \begin{minipage}{.44\textwidth}
              \begin{prooftree}
                \AxiomC{}
                \UnaryInfC{$\hoare{\emp}{\connect(x_1.p_1, x_2.p_2)}{\interac{x_1}{p_1}{x_2}{p_2}}$}
              \end{prooftree}
            \end{minipage}
            \begin{minipage}{.55\textwidth}
              \begin{prooftree}
                \AxiomC{}
                \UnaryInfC{$\hoare{\interac{x_1}{p_1}{x_2}{p_2}}{\disconnect(x_1.p_1, x_2.p_2)}{\emp}$}
              \end{prooftree}
            \end{minipage} \\
            \vspace*{.5\baselineskip}
            \centerline{a. Axioms for Primitive Commands} 
            \vspace*{.5\baselineskip}
            
            \begin{minipage}{.5\textwidth}
            \begin{prooftree}
              \AxiomC{$\hoare{\phi \wedge (\theta * \predtrue)}{\acomm}{\psi}$}
              \RightLabel{$\fv{\phi} \cap \set{x_1, \ldots, x_k} = \emptyset$}
              \UnaryInfC{$\hoare{\phi}{
                  \withcomm{x_1, \ldots, x_k}{\theta}{\acomm}
                }{\exists x_1 \ldots \exists x_k ~.~ \psi}$}
            \end{prooftree} 
            \end{minipage}
            
            \begin{minipage}{.49\textwidth}
            \begin{prooftree}
              \AxiomC{$\hoare{\phi}{\acomm_1}{\varphi}$}
              \AxiomC{$\hoare{\varphi}{\acomm_2}{\psi}$}
              \RightLabel{$\ahavoc(\sem{\varphi}{}) \subseteq \sem{\varphi}{}$}
              \BinaryInfC{$\hoare{\phi}{\acomm_1;\acomm_2}{\psi}$}
            \end{prooftree}
            \end{minipage}
            
            \begin{minipage}{.4\textwidth}
              \begin{prooftree}
                \AxiomC{$\hoare{\phi}{\acomm_1}{\psi}$}
                \AxiomC{$\hoare{\phi}{\acomm_2}{\psi}$}
                \BinaryInfC{$\hoare{\phi}{\acomm_1 + \acomm_2}{\psi}$}
              \end{prooftree}
            \end{minipage}
            \begin{minipage}{.4\textwidth}
              \begin{prooftree}
                \AxiomC{$\hoare{\phi}{\acomm}{\phi}$}
                \RightLabel{$\ahavoc(\sem{\phi}{}) \subseteq \sem{\phi}{}$}
                \UnaryInfC{$\hoare{\phi}{\acomm^*}{\phi}$}
              \end{prooftree}
            \end{minipage} \\            
            \vspace*{.5\baselineskip}
            
            \centerline{b. Inference Rules for Programs} 
            \vspace*{.5\baselineskip}
            \begin{minipage}{.4\textwidth}
              \begin{prooftree}
                \AxiomC{$\hoare{\phi_i}{\acomm}{\psi_i} \mid i \in \interv{1}{k}$}
                \UnaryInfC{$\hoare{\bigvee_{i=1}^k \phi_i}{\acomm}{\bigvee_{i=1}^k \psi_i}$}
              \end{prooftree}
            \end{minipage}
            \begin{minipage}{.4\textwidth}
              \begin{prooftree}
                \AxiomC{$\hoare{\phi_i}{\acomm}{\psi_i} \mid i \in \interv{1}{k}$}
                \UnaryInfC{$\hoare{\bigwedge_{i=1}^k \phi_i}{\acomm}{\bigwedge_{i=1}^k \psi_i}$}
              \end{prooftree}
            \end{minipage} \\
            \vspace*{.5\baselineskip} 
            \begin{minipage}{.4\textwidth}
              \begin{prooftree}
                \AxiomC{$\hoare{\phi'}{\acomm}{\psi'}$}
                \RightLabel{$\begin{array}{l} \phi \models \phi' \\
                    \psi' \models \psi \end{array}$}
                \UnaryInfC{$\hoare{\phi}{\acomm}{\psi}$}
              \end{prooftree}
            \end{minipage}
            \begin{minipage}{.4\textwidth}
              \begin{prooftree}
                \AxiomC{$\hoare{\phi}{\acomm}{\psi}$}
                \RightLabel{$\begin{array}{l}
                    \acomm \in \localcomms \\
                    \modif{\acomm}\cap\fv{F}=\emptyset
                    \end{array}$}
                \UnaryInfC{$\hoare{\phi * F}{\acomm}{\psi * F}$}
              \end{prooftree}
            \end{minipage} \\            
            
            \vspace*{.5\baselineskip} 
            \centerline{c. Structural Inference Rules}
          \end{center}}
\vspace*{-\baselineskip}
\end{figure}

The \emph{frame rule} (Fig. \ref{fig:hoare}c bottom-right) allows to
apply the specification of a \emph{local program}, defined below, to a
set of configurations that may contain more resources (components and
interactions) than the ones asserted by the precondition. Intuitively,
a local program requires a bounded amount of components and
interactions to avoid faulting and, moreover, it only changes the
configuration of the local subsystem, not affecting the entire
system's configuration. Formally, the set $\localcomms$ of local
programs is the least set that contains the primitive commands
$\primcomms$ and is closed under the application of the following
rules:
\[
\acomm \in \localcomms \Rightarrow \withcomm{\vec{x}}{\pureform}{\acomm} \in
\localcomms \text{, if $\pureform$ is a conjunction of (dis-)equalities}
\hspace*{4mm} \acomm_1, \acomm_2 \in \localcomms \Rightarrow \acomm_1
+ \acomm_2 \in \localcomms
\]
The extra resources, not required to execute a local program, are
specified by a frame $F$, whose free variables are not modified by the
local program $\acomm$. Formally, the set of variables modified by a
local program $\acomm \in \localcomms$ is defined inductively on its
structure:
\[\begin{array}{l}
\hspace*{-2mm}\modif{\new(q,x)} \isdef \set{x} \hspace*{2.5cm}
\modif{\acomm} \isdef \emptyset \text{, for all }
\acomm\in\primcomms\smallsetminus\set{\new(q,x)\mid q\in \states, x \in \vars} \\
\hspace*{-2mm}\modif{\withcomm{\vec{x}}{\theta}{\acomm}} \isdef \vec{x} \cup \modif{\acomm} 
\hspace*{1.8cm}
\modif{\acomm_1 + \acomm_2} \isdef \modif{\acomm_1} \cup \modif{\acomm_2}
\end{array}\]
The frame rule is sound only for programs whose semantics are
\emph{local actions}, defined below:

\begin{definition}[Locality]\label{def:locality}
  Given a set of variables $X \subseteq \vars$, an action $f :
  \configset \rightarrow \pow{\configset}^\errconfigs$ is \emph{local
    for $X$} if and only if $f(\aconfig_1 \comp \aconfig_2) \subseteq
  f(\aconfig_1) \comp \lift{\set{\aconfig_2}}{X}$ for all $\aconfig_1,
  \aconfig_2 \in \configset$, where, for any set $S$ of
  configurations:
  \[\lift{S}{X} \isdef
  \Set{(\comps, \interacs, \statemap', \store') \mid \config \in S,
    \forall x \in \vars \smallsetminus X ~.~ \store'(x)=\store(x),~
    \forall c \in \universe \smallsetminus \store(X) ~.~ \statemap'(c)
    = \statemap(c) }.\] An action $f$ is \emph{local} if and only if
  it is local for the empty set of variables.
\end{definition}
An action that is local for a set of variables $X$ allows for the
change of the store values of the variables in $X$ and the states of
the components indexed by those values, only. Essentially, $\new(q,x)$
is local for $\set{x}$, because the fresh index associated to $x$ is
nondeterministically chosen and the state is $q$, whereas the other
primitive commands are local, in general. The semantics of every local
program $\acomm\in\localcomms$ is a local action, as shown below:

\begin{restatable}{lemma}{LemmaLocality}\label{lemma:locality}
  For every program $\acomm \in \localcomms$, the action
  $\semcomm{\acomm}$ is local for $\modif{\acomm}$.
\end{restatable}

Moreover, $\localcomms$ is precisely the set of programs with local
semantics, as conditionals and sequential compositions (hence also
iterations) are not local, in general:

 \begin{example}\label{ex:nonlocal}
   To understand why $\localcomms$ is precisely the set of local
   commands, consider the programs: \begin{itemize}
   \item $(\withcomm{x}{\compin{x}{q}}{\delete(x)})$ is
     not local because, letting $\aconfig_1$ be a configuration with
     zero components and $\aconfig_2$ be a configuration with one
     component in state $q$, we have:
     \[\begin{array}{l}
     \semcomm{\withcomm{x}{\compin{x}{q}}{\delete(x)}}(\aconfig_1 \comp \aconfig_2) = 
     \semcomm{\withcomm{x}{\compin{x}{q}}{\delete(x)}}(\aconfig_2) = \set{\aconfig_1} \\
     \text{whereas } \semcomm{\withcomm{x}{\compin{x}{q}}{\delete(x)}}(\aconfig_1) \comp
     \set{\aconfig_2} = \emptyset \comp \set{\aconfig_2} = \emptyset.
     \end{array}\]
   \item $(\skipcomm; \skipcomm)$ is not local because, considering
     the system from Fig. \ref{fig:token-ring}, if we take $\aconfig_1$ and
     $\aconfig_2$, such that $\aconfig_1 \models
     \compin{x}{\toktoken}$ and $\aconfig_2 \models
     \interac{x}{out}{y}{in} * \compin{y}{\toknotok}$, we have:
     \[\begin{array}{l}
     \semcomm{\skipcomm; \skipcomm}(\aconfig_1 \comp \aconfig_2) = 
     \sem{\compin{x}{\toktoken} * \interac{x}{out}{y}{in} * \compin{y}{\toknotok}}{} ~\cup~
     \sem{\compin{x}{\toknotok} * \interac{x}{out}{y}{in} * \compin{y}{\toktoken}}{} \\
     \text{whereas } \semcomm{\skipcomm; \skipcomm}(\aconfig_1) \comp \set{\aconfig_2} =
     \sem{\compin{x}{\toktoken} * \interac{x}{out}{y}{in} * \compin{y}{\toknotok}}{} \hfill\blacksquare
     \end{array}\]
   \end{itemize}
\end{example}

We write $\vdash \hoare{\phi}{\acomm}{\psi}$ if and only if
$\hoare{\phi}{\acomm}{\psi}$ can be derived from the axioms using the
inference rules from Fig. \ref{fig:hoare} and show the soundness of
the proof system in the following. The next lemma gives sufficient
conditions for the soundness of the axioms (Fig. \ref{fig:hoare}a):

\begin{restatable}{lemma}{LemmaHoareAxioms}\label{lemma:hoare-axioms}
  For each axiom $\hoare{\phi}{\acomm}{\psi}$, where $\acomm \in
  \primcomms$ is primitive, we have $\semcomm{\acomm}(\sem{\phi}{}) =
  \sem{\psi}{}$.
\end{restatable}
The soundness of the proof system in Fig. \ref{fig:hoare} follows from
the soundness of each inference rule:

\begin{restatable}{theorem}{ThmSoundness}\label{thm:soundness}
  For any Hoare triple $\hoare{\phi}{\acomm}{\psi}$, if\; $\vdash
  \hoare{\phi}{\acomm}{\psi}$ then $\models
  \hoare{\phi}{\acomm}{\psi}$.
\end{restatable}

As an optimization, reconfiguration proofs can often be simplified, by
safely skipping the check of one or more havoc invariance side
conditions of sequential compositions, as explained below.

\begin{definition}\label{def:single-reversal}
  A program of the form \(\disconnect(x_1.p_1, x'_1.p'_1); ~\ldots~
  \disconnect(x_k.p_k, x'_k.p'_k);\) \(\connect(x_{k+1}.p_{k+1},
  x'_{k+1}.p'_{k+1}); ~\ldots~ \connect(x_\ell.p_\ell,
  x'_\ell.p'_\ell)\) is said to be a \emph{single reversal program}.
\end{definition}
Single reversal programs first disconnect components and then
reconnect them in a different way. For such programs, only the first
and last application of the sequential composition rule require
checking havoc invariance:

\begin{restatable}{proposition}{PropSingleReversal}\label{prop:single-reversal}
  Let \(\acomm = \disconnect(x_1.p_1, x'_1.p'_1); ~\ldots~
  \disconnect(x_k.p_k, x'_k.p'_k);\) \\
  \(\connect(x_{k+1}.p_{k+1}, x'_{k+1}.p'_{k+1}); ~\ldots~
  \connect(x_\ell.p_\ell,x'_\ell.p'_\ell)\) be a single reversal
  program.  If $\phi_0, \ldots, \phi_\ell$ are \adl\ formul{\ae}, such
  that: \begin{itemize}
  \item $\models \hoare{\phi_{i-1}}{\disconnect(x_i.p_i, x'_i.p'_i)}{\phi_i}$, for all $i \in \interv{1}{k}$,
  \item $\models \hoare{\phi_{j-1}}{\connect(x_j.p_j, x'_j.p'_j)}{\phi_j}$, for
    all $j \in \interv{k+1}{\ell}$, and
  \item $\phi_1$ and $\phi_{\ell-1}$ are havoc invariant,
  \end{itemize}
  then we have \(\models \hoare{\phi_0}{\acomm}{\phi_\ell}\).
\end{restatable}

\subsection{Examples of Reconfiguration Proofs}
\label{sec:reconfiguration-proofs}

We prove that the outcome of the reconfiguration program from
Fig. \ref{fig:token-ring} (Listing~\ref{lst:delete-component-notok}),
started in a token ring configuration with at least two components in
state $\toknotok$ and at least one in state $\toktoken$, is a token
ring with at least one component in each state. The pre- and
postcondition are $\exists x \exists y ~.~ \chain{2}{1}(x,y) *
\interac{y}{out}{x}{in}$ and $\exists x \exists y ~.~
\chain{1}{1}(x,y) * \interac{y}{out}{x}{in}$, respectively, with the
definitions of $\chain{h}{t}(x,y)$ given in Example \ref{ex:chain},
for all constants $h,t\in\N$.

\begin{center}
    \begin{minipage}{.6\textwidth}
    {\small\begin{lstlisting}
$\textcolor{violet}{\set{\exists x \exists y ~.~ \chain{2}{1}(x,y) * \interac{y}{out}{x}{in}}}$
with $x,y,z : \interac{x}{out}{y}{in} \ast \compin{y}{\toknotok} \ast \interac{y}{out}{z}{in}$ do
$\textcolor{violet}{\set{(\exists x \exists y ~.~ \chain{2}{1}(x,y) * \interac{y}{out}{x}{in})
 \wedge \big(\interac{x}{out}{y}{in} * \compin{y}{\toknotok} * \interac{y}{out}{z}{in} * \predtrue\big)}}$ $(\star)$
$\textcolor{violet}{\set{\interac{x}{out}{y}{in} * \boxaround{\compin{y}{\toknotok} * \interac{y}{out}{z}{in} * \chain{1}{1}(z,x)}}}$ 
disconnect(x.$\mathit{out}$, y.$\mathit{in}$);
$\textcolor{violet}{\set{\boxaround{\compin{y}{\toknotok}} * \interac{y}{out}{z}{in} * \boxaround{\chain{1}{1}(z,x)}}}$ $\hinv$
disconnect(y.$\mathit{out}$, z.$\mathit{in}$); 
$\textcolor{violet}{\set{\compin{y}{\toknotok} * \boxaround{\chain{1}{1}(z,x)}}}$ $\hinv$
  \end{lstlisting}}
  \end{minipage}
  \begin{minipage}{.39\textwidth}
    {\small\begin{lstlisting}



      
delete(y); 
$\textcolor{violet}{\set{\boxaround{\chain{1}{1}(z,x)}}}$ $\hinv$
connect(x.$\mathit{out}$, z.$\mathit{in}$)
$\textcolor{violet}{\set{\chain{1}{1}(z,x) * \interac{x}{out}{z}{in}}}$
od
$\textcolor{violet}{\set{\exists x \exists y ~.~\chain{1}{1}(x,y) * \interac{y}{out}{x}{in}}}$
    \end{lstlisting}}
    \end{minipage}
\end{center}
  
The inference rule for conditional programs sets up the precondition
$(\star)$ for the body of the conditional. This formula is equivalent
to $\interac{x}{out}{y}{in} * \compin{y}{\toknotok} *
\interac{y}{out}{z}{in} * \chain{1}{1}(z,x)$. To understand this
point, we derive from $(\star)$ the equivalences:
\[\begin{array}{l}
(\exists x \exists y ~.~ \chain{2}{1}(x,y) * \interac{y}{out}{x}{in}) \wedge
(\interac{x}{out}{y}{in} * \compin{y}{\toknotok} * \interac{y}{out}{z}{in} * \predtrue) \equiv \\
\exists \overline{x} \exists \overline{y} \exists \overline{z} ~.~ \interac{\overline{x}}{out}{\overline{y}}{in} *
\compin{\overline{y}}{\toknotok} * \interac{\overline{y}}{out}{\overline{z}}{in} *
\chain{1}{1}(\overline{z},\overline{x}) \wedge (\interac{x}{out}{y}{in} * \compin{y}{\toknotok} *
\interac{y}{out}{z}{in} * \predtrue) \\
\equiv \interac{x}{out}{y}{in} * \compin{y}{\toknotok} * \interac{y}{out}{z}{in} * \chain{1}{1}(z,x)
\end{array}\]
where the first step can be proven by entailment checking (discussed
in \S\ref{sec:open-problems}). The following four annotations above
are obtained by applications of the axioms and the frame rule (the
frame formul{\ae} are displayed within boxes). The sequential
composition rule is applied by proving first that the annotations
marked as $\hinv$ are havoc invariant (the proof is given later in
\S\ref{sec:havoc-proof}).

We have considered the reconfiguration program from
Fig. \ref{fig:token-ring} (Listing~\ref{lst:delete-component-notok})
which deletes a component from a token ring. The dual operation is the
addition of a new component. Here the precondition states that the
system is a valid token ring, with at least one component in state
$\toknotok$ and at least another one in state $\toktoken$. We prove
that the execution of the dual program yields a token ring with at least
two components in state $\toknotok$, as the new component is added
without a token.

\begin{center}
  \begin{minipage}{.64\textwidth}
  {\small\begin{lstlisting}
 $\textcolor{violet}{\set{\exists x \exists y ~.~ \chain{1}{1}(x,y) * \interac{y}{out}{x}{in}}}$
with $x,z : \interac{x}{out}{z}{in}$ do
$\textcolor{violet}{\set{(\exists x \exists y ~.~ \chain{1}{1}(x,y) * \interac{y}{out}{x}{in})
  \wedge \big(\interac{x}{out}{z}{in} * \predtrue\big)}}$ $(\star)$
$\textcolor{violet}{\set{\interac{x}{out}{z}{in} * \boxaround{\chain{1}{1}(z,x)}}}$
disconnect(x.$\mathit{out}$, z.$\mathit{in}$);
$\textcolor{violet}{\set{\chain{1}{1}(z,x)}}$ $\hinv$
new($\toknotok$,y);
$\textcolor{violet}{\set{ \compin{y}{\toknotok} * \boxaround{\chain{1}{1}(z,x)}}}$ $\hinv$
  \end{lstlisting}}
  \end{minipage}
  \begin{minipage}{.35\textwidth}
    {\small\begin{lstlisting}
      
connect(y.$\mathit{out}$,z.$\mathit{in}$);
$\textcolor{violet}{\set{ \boxaround{\scriptstyle \compin{y}{\toknotok}} *
  \interac{y}{out}{z}{in} * \boxaround{\chain{1}{1}(z,x)}}}$
$\textcolor{violet}{\set{\chain{2}{1}(y,x)}}$ $\hinv$
connect(x.$\mathit{out}$,y.$\mathit{in}$)
$\textcolor{violet}{\set{\chain{2}{1}(y,x) * \interac{x}{out}{y}{in}}}$
od
$\textcolor{violet}{\set{\exists x \exists y ~.~\chain{2}{1}(x,y) * \interac{y}{out}{x}{in}}}$
    \end{lstlisting}}
    \end{minipage}
\end{center}

The annotation $(\star)$ is given by the inference rule for
conditional programs. Then we can derive the equivalence \((\exists x
\exists y ~.~ \chain{1}{1}(x,y) * \interac{y}{out}{x}{in}) \wedge
\big(\interac{x}{out}{z}{in} * \predtrue\big) \equiv
\interac{x}{out}{z}{in} * \chain{1}{1}(z,x)\). In the subsequent
lines, some axioms and the frame rule are applied (the frame is
displayed in the postconditions of the commands within the boxes). The
annotations marked as $\hinv$ must be shown to be havoc invariant and
then the sequential composition rule is applied to complete the proof.

\section{The Havoc Proof System}
\label{sec:havoc}

This section describes a set of axioms and inference rules for proving
the validity of havoc invariance queries of the form
$\ahavoc(\sem{\phi}{}) \subseteq \sem{\phi}{}$, where $\phi$ is a
\adl\ formula interpreted over a given SID and $\ahavoc$ is the havoc
action (Def. \ref{def:havoc}). Such a query is valid if and only if
the result of applying any sequence of interactions on a model of
$\phi$ is again a model of $\phi$ (Def. \ref{def:havoc-invariant}).
Havoc invariance queries occur as side conditions in the rules for
sequential composition and iteration (Fig. \ref{fig:hoare}b) of
reconfiguration programs. Thus, having a proof system for havoc
invariance is crucial for the applicability of the rules in
Fig. \ref{fig:hoare} to obtain proofs of reconfiguration programs.

The havoc proof system uses a compositional rule, able to split a
query of the form $\ahavoc(\sem{\phi_1 * \phi_2}{}) \subseteq
\sem{\psi_1 * \psi_2}{}$ into two queries $\ahavoc(\sem{\phi_i *
  \mathcal{F}_i}{}) \subseteq \sem{\psi_i * \mathcal{F}_i}{}$, where
each \emph{frontier formula} $\mathcal{F}_i$ defines a set of
interactions that over-approximate the effect of executing the system
described by $\phi_{3-i}$ (resp. $\psi_{3-i}$) over the one described
by $\phi_i$ (resp. $\psi_i$), for $i = 1,2$. In principle, the
frontier formul{\ae} ($\mathcal{F}_1$ and $\mathcal{F}_2$) can be
understood as describing the interference between parallel actions in
an assume/rely guarantee-style parallel composition rule
\cite{Owicki1978,DBLP:phd/ethos/Jones81}. In particular, since the
frontier formul{\ae} only describe interactions and carry no state
information whatsoever, such assumptions about events triggered by the
environment are reminiscent of compositional reasoning about
input/output automata \cite{DBLP:journals/scp/ChiltonJK14}.

Compositional reasoning about havoc actions requires the following
relaxation of the definition of havoc state changes
(Def. \ref{def:havoc}), by allowing the firing of loose, in addition
to tight interactions:

\begin{definition}\label{def:open}
  The following rules define a relation $ \Open{(c_1, p_1, c_2,
    p_2)}{} ~\subseteq \configset \times \configset$, parameterized by
  a given interaction $(c_1, p_1, c_2, p_2)$:
  \begin{prooftree}
  \AxiomC{$\begin{array}{ccccc}
      (c_1, p_1, c_2, p_2) \in \interacs & c_i \in \comps, c_{3-i} \not\in \comps & \statemap(c_i) = q_i & q_i \arrow{p_i}{} q'_i
    \end{array}$}
  \RightLabel{$i = 1,2$}
  \LeftLabel{(\looserule)}
  \UnaryInfC{$(\comps,\interacs,\statemap,\store) \Open{(c_1, p_1, c_2, p_2)}{} (\comps,\interacs,\statemap[c_i \leftarrow q'_i],\store)$}
  \end{prooftree}
  \begin{prooftree}
    \AxiomC{$\begin{array}{ccccc}
        (c_1, p_1, c_2, p_2) \in \interacs & c_1 \neq c_2 \in \comps & \statemap(c_i) = q_i & q_i \arrow{p_i}{} q'_i,~ i=1,2
      \end{array}$}
    \LeftLabel{(\tightrule)}
    \UnaryInfC{$(\comps,\interacs,\statemap,\store) \Open{(c_1, p_1, c_2, p_2)}{} (\comps,\interacs,\statemap[c_1 \leftarrow q'_1][c_2 \leftarrow q'_2],\store)$}
  \end{prooftree}
  For a sequence $w = i_1 \ldots i_n$ of interactions, we define
  $\Open{w}{}$ to be the composition of $\Open{i_1}{}, \ldots,
  \Open{i_n}{}$, assumed to be the identity relation, if $w$ is empty.
\end{definition}
The difference with Def. \ref{def:havoc} is that only the states of
the components from the configuration are changed according to the
transitions in the behavior. This more relaxed definition matches the
intuition of partial systems in which certain interactions may be
controlled by an external environment; those interactions are
conservatively assumed to fire anytime they are enabled by the
components of the current structure, independently of the environment.

\begin{example}(contd. from Example \ref{ex:havoc})\label{ex:open}
   Let $\widetilde{\aconfig}_i = (\set{c_2, c_3}, \set{(c_i, \mathit{out}, c_{i
       \mod 3 + 1}, \mathit{in}) \mid i \in \interv{1}{3}},
   \statemap_i, \store)$, for $i \in \interv{1}{3}$ be the top-most
   configurations from Fig. \ref{fig:reconfiguration} without the
   $c_1$ component, where $\statemap_1(c_2) = \toknotok$,
   $\statemap_1(c_3) = \toktoken$, $\statemap_2(c_2) =
   \statemap_2(c_3) = \toknotok$, $\statemap_3(c_2) = \toktoken$,
   $\statemap_3(c_3) = \toknotok$. Then, by executing the loose
   interactions $(c_3,\mathit{out},c_1,\mathit{in})$ and
   $(c_1,\mathit{out},c_2,\mathit{in})$ from $\aconfig_1$, we obtain:
   \[\widetilde{\aconfig}_1 \Open{(c_3,\mathit{out},c_1,\mathit{in})}{} \widetilde{\aconfig}_2
   \Open{(c_1,\mathit{out},c_2,\mathit{in})}{} \widetilde{\aconfig}_3\]
   Executing the tight interaction $(c_2, \mathit{out}, c_3,
   \mathit{in})$ from $\widetilde{\aconfig}_3$ leads back to $\widetilde{\aconfig}_1$ i.e.,
   \(\widetilde{\aconfig}_3 \Open{(c_2,\mathit{out},c_3,\mathit{in})}{}
   \widetilde{\aconfig}_1\).\hfill$\blacksquare$
\end{example}

\subsection{Regular Expressions}

Proving the validity of a havoc query $\ahavoc(\sem{\phi}{}) \subseteq
\sem{\psi}{}$ involves reasoning about the sequences of interactions
that define the outcome of the havoc action. We specify languages of
such sequences using extended regular expressions, defined inductively
by the following syntax:
\[\are ::= \epsilon \mid \interof{\alpha} \mid \are \cdot \are
\mid \are \cup \are \mid \are^* \mid \are
\parcomp_{\aenv,\aenv} \are\] where $\epsilon$ denotes the empty
string, $\interof{\alpha}$ is an \emph{alphabet symbol} associated
with either an interaction atom or a predicate atom $\alpha$ and
$\cdot$, $\cup$ and ${~}^*$ are the usual concatenation, union and
Kleene star. By $\are_1 \parcomp_{\aenv_1,\aenv_2} \are_2$ we denote
the interleaving (zip) product of the languages described by $\are_1$
and $\are_2$ with respect to the sets $\aenv_1$ and $\aenv_2$ of
alphabet symbols of the form $\interof{\alpha}$, respectively.

The \emph{language} of a regular expression $\are$ in a configuration
$\aconfig = \config$ is defined below:
\[\begin{array}{ll}
\semlang{\epsilon}{\aconfig} \isdef \set{\epsilon} 
&
\semlang{\interof{\alpha}}{\aconfig} \isdef \bigcup\set{\interacs \mid (\comps,\interacs,\statemap,\store) \substreq \aconfig,~
  (\comps,\interacs,\statemap,\store) \models \alpha}
\\
\semlang{\are_1 \cdot \are_2}{\aconfig} \isdef \set{w_1w_2 \mid w_i \in \semlang{\are_i}{\aconfig},~ i = 1,2}
&
\semlang{\are_1 \cup \are_2}{\aconfig} \isdef \semlang{\are_1}{\aconfig} \cup \semlang{\are_2}{\aconfig} \\
\semlang{\are^*}{\aconfig} \isdef \bigcup_{i\geq0} \semlang{\are^i}{\aconfig} 
&
\semlang{\are_1 \parcomp_{\aenv_1,\aenv_2} \are_2}{\aconfig} \isdef \set{w \mid \proj{w}{\semlang{\aenv_i}{\aconfig}} \in \semlang{\are_i}{\aconfig}, ~ i = 1,2}
\end{array}\]
where $\semlang{\aenv}{\aconfig} \isdef \bigcup_{\interof{\alpha} \in
  \aenv}\semlang{\interof{\alpha}}{\aconfig}$ and
$\proj{w}{\semlang{\aenv}{\aconfig}}$ is the word obtained from $w$ by
deleting each symbol not in $\semlang{\aenv}{\aconfig}$ from it. The
$i$-th composition of $\are$ with itself is defined, as usual, by
$\are^0 \isdef \epsilon$ and $\are^{i+1} = \are^i \cdot \are$, for $i
\geq 0$. We denote by $\supp{\are}$ the \emph{support} of $\are$ i.e.,
set of alphabet symbols $\interof{\alpha}$ from the regular expression
$\are$.

\begin{example}\label{ex:interaction-language}
  Let $\aconfig = (\set{c_1, c_2, c_3, c_4}, \set{(c_1, \mathit{out},
    c_2, \mathit{in}), (c_2, \mathit{out}, c_3, \mathit{in}), (c_3,
    \mathit{out}, c_4, \mathit{in})}, \statemap, \store)$ be a
  configuration, such that $\store(x) = c_1$, $\store(y) = c_2$ and
  $\store(z) = c_3$. Then, we have
  $\semlang{\interof{\interac{x}{out}{y}{in}}}{\aconfig} = \set{(c_1,
    \mathit{out}, c_2, \mathit{in})}$,
  $\semlang{\interof{\interac{y}{out}{z}{in}}}{\aconfig} = \set{(c_2,
    \mathit{out}, c_3, \mathit{in})}$ and
  $\semlang{\interof{\chain{0}{0}(x,z)}}{\aconfig} = \{(c_1,
  \mathit{out}, c_2, \mathit{in})$, $(c_2, \mathit{out}, c_3,
  \mathit{in})\}$. \hfill$\blacksquare$
\end{example}

Given a configuration $\aconfig$ and a predicate atom $\alpha$, there
can be, in principle, more than one subconfiguration $\aconfig'
\substreq \aconfig$, such that $\aconfig' \models \alpha$. This is
problematic, because then $\semlang{\interof{\alpha}}{\aconfig}$ may
contain interactions from different subconfigurations of $\aconfig$,
that are models of $\alpha$, thus cluttering the definition of the
language $\semlang{\interof{\alpha}}{\aconfig}$. We fix this issue by
adapting the notion of \emph{precision}, originally introduced for
\seplog\ \cite{CalcagnoOHearnYan07,OHearnYangReynolds09}, to our
configuration logic:

\begin{definition}[Precision]\label{def:precise}
  A formula $\phi$ is \emph{precise on a set $S$ of configurations} if
  and only if, for every configuration $\aconfig \in S$, there exists
  at most one configuration $\aconfig'$, such that $\aconfig'
  \substreq \aconfig$ and $\aconfig' \models \phi$. A set of
  formul{\ae} $\Phi$ is \emph{precisely closed} if $\psi$ is precise
  on $\sem{\phi}{}$, for any two formul{\ae} $\phi,\psi \in \Phi$.
\end{definition}
Symbolic configurations using predicate atoms are not precise for
$\configset$, in general\footnote{Unlike the predicates that define
  acyclic data structures (lists, trees) in \seplog, which are
  typically precise.}. To understand this point, consider a
configuration consisting of two overlapping models of
$\chain{h}{t}(x,y)$, starting and ending in $x$ and $y$, respectively,
with a component that branches on two interactions after $x$ and
another component that joins the two branches before $y$. Then
$\chain{h}{t}(x,y)$ is not precise on such configurations (that are
not models of $\chain{h}{t}(x,y)$ whatsoever). On the positive side,
we can state the following:

\begin{restatable}{proposition}{PropChainPrecise}\label{prop:chain-precise}
  The set of symbolic configurations built using predicate atoms
  $\chain{h}{t}(x,y)$, for $h,t\geq0$ (Example \ref{ex:chain}) is
  precisely closed.
\end{restatable}

Two regular expressions are \emph{congruent} if they denote the same
language, whenever interpreted in the same configuration. Lifted to
models of a symbolic configuration, we define:
\begin{definition}\label{def:congruence}
  Given a symbolic configuration $\phi$, the regular expressions
  $\are_1$ and $\are_2$ are \emph{congruent for $\phi$}, denoted
  $\are_1 \requiv{\phi} \are_2$, if and only if
  $\semlang{\are_1}{\aconfig} = \semlang{\are_2}{\aconfig}$, for all
  configurations $\aconfig \in \sem{\phi}{}$.
\end{definition}
Despite the universal condition that ranges over a possibly infinite
set of configurations, congruence of regular expressions with alphabet
symbols of the form $\interof{\alpha}$, where $\alpha$ is an
interaction or a predicate atom, is effectively decidable by an
argument similar to the one used to prove equivalence of symbolic
automata \cite{DAntoniV21}.

\subsection{Inference Rules for Havoc Triples}

We use judgements of the form $\havoctriple{\aenv}{\phi}{\are}{\psi}$,
called \emph{havoc triples}, where $\phi$ and $\psi$ are
\adl\ formul{\ae}, $\are$ is a regular expression, and $\aenv$ is an
\emph{environment} (a set of alphabet symbols), whose role will be
made clear below (Def. \ref{def:distinctive} and Lemma
\ref{lemma:distinctive}). A havoc triple states that each finite
sequence of (possibly loose) interactions described by a word in
$\are$, when executed in a model of the precondition $\phi$, yields a
model of the postcondition $\psi$.

\begin{definition}\label{def:havoc-valid}
  A havoc triple $\havoctriple{\aenv}{\phi}{\are}{\psi}$ is
  \emph{valid}, written $\models
  \havoctriple{\aenv}{\phi}{\are}{\psi}$ if and only if, for each
  configuration $\aconfig \in \sem{\phi}{}$, each sequence of
  interactions $w\in\semlang{\are}{\aconfig}$ and each configuration
  $\aconfig'$, such that $\aconfig \Open{w}{} \aconfig'$, we have
  $\aconfig' \in \sem{\psi}{}$.
\end{definition}

For a symbolic configuration $\phi$, we denote by $\iatoms{\phi}$ and
$\patoms{\phi}$ the sets of interaction and predicate atoms from
$\phi$, respectively and define the set of atoms $\atoms{\phi} \isdef
\iatoms{\phi} \cup \patoms{\phi}$ and the regular expression
$\interof{\phi} \isdef \bigcup_{\alpha\in\atoms{\phi}}
\interof{\alpha}$. We show that the validity of a havoc triple is a
sufficient argument for the validity of a havoc query; because havoc
triples are evaluated via open state changes
(Def. \ref{def:havoc-valid}), the dual implication is not true, in
general.

\begin{restatable}{proposition}{PropHavocValid}\label{prop:havoc-valid}
  If $\models
  \havoctriple{\aenv}{\phi}{\interof{\phi}^*}{\psi}$ then
  $\ahavoc(\sem{\phi}{}) \subseteq \sem{\psi}{}$.
\end{restatable}

We describe next a set of axioms and inference rules used to prove the
validity of havoc triples. For a symbolic configuration $\phi$, we
write $x \symconfeq{\phi} y$ ($x \symconfneq{\phi} y$) if and only if
the equality (disequality) between $x$ and $y$ is asserted by the
symbolic configuration $\phi$, e.g. $x \symconfeq{\emp * x = z * z =
  y} y$ and $x \symconfneq{\company{x} * \company{y}} y$; note that $x
\symconfneq{\phi} y$ is not necessarily the negation of $x
\symconfeq{\phi} y$.

\begin{definition}\label{def:disabled-excluded}
  For a symbolic configuration $\phi$ and an interaction atom
  $\interac{x_1}{p_1}{x_2}{p_2}$, we write: \begin{itemize}
  \item $\disabled{\phi}{\interac{x_1}{p_1}{x_2}{p_2}}$ if and only if
    $\phi$ contains a subformula $\compin{y}{q}$, such that $y
    \symconfeq{\phi} x_i$ and $q$ is not the pre-state of some
    behavior transition with label $p_i$, for some $i = 1,2$;
    intuitively, any interaction defined by the formula
    $\interac{x_1}{p_1}{x_2}{p_2}$ is disabled in any model of $\phi$,
  \item $\excluded{\phi}{\interac{x_1}{p_1}{x_2}{p_2}}$ if and only
    if, for each interaction atom $\interac{y_1}{p'_1}{y_2}{p'_2} \in
    \iatoms{\phi}$, there exists $i \in \interv{1}{2}$, such that $x_i
    \symconfneq{\phi} y_i$; intuitively, the interaction defined by
    the formula $\interac{x_1}{p_1}{x_2}{p_2}$ is not already present
    in a model of $\phi$ i.e., $\interac{x_1}{p_1}{x_2}{p_2} * \phi$
    is satisfiable.
  \end{itemize}
\end{definition}

The axioms (Fig. \ref{fig:havoc-rules}a) discharge valid havoc triples
for the empty sequence (\epsilonr), that changes nothing and the
sequence consisting of a single interaction atom, that can be either
disabled in every model (\disr), or enabled in some model (\inter) of
the precondition, respectively; in particular, the (\inter) axiom
describes the open state change produced by an interaction
(Def. \ref{def:open}), firing on a (possibly empty) set of components,
whose states match the pre-states of transitions for the associated
behaviors. The (\botr) axiom discharges trivially valid triples with
unsatisfiable (false) preconditions.

The redundancy rule (\ii) in Fig. \ref{fig:havoc-rules}b removes an
interaction atom from the precondition of a havoc triple, provided
that the atom is never interpreted as an interaction from the language
denoted by the regular expression from the triple. Conversely, the
rule (\ie) adds an interaction to the precondition, provided that the
precondition (with that interaction atom) is consistent. Note that,
without the $\excluded{\phi}{\alpha}$ side condition, we would obtain
a trivial proof for any triple, by adding an interaction atom twice to
the precondition, i.e. using the rule (\ie), followed by (\botr).

\begin{figure}[t!]
  \vspace*{-\baselineskip}
  \caption{Proof System for Havoc Triples}
  \label{fig:havoc-rules}
        {\small\begin{center}
          \begin{minipage}{.29\textwidth}
            \begin{prooftree}
              \AxiomC{}
              \LeftLabel{(\epsilonr)}
              \UnaryInfC{$\havoctriple{\aenv}{\phi}{\epsilon}{\phi}$}
            \end{prooftree}
          \end{minipage}
          \begin{minipage}{.4\textwidth}
            \begin{prooftree}
              \AxiomC{}
              \LeftLabel{(\disr)}
              \RightLabel{$\begin{array}{ll}
                  {\alpha = \interac{x_1}{p_1}{x_2}{p_2}} \\[-.5mm]
                  {\disabled{\phi}{\alpha}}
                \end{array}$}
              \UnaryInfC{$\havoctriple{\aenv}{\phi}{\interof{\alpha}}{\predfalse}$}
            \end{prooftree}
          \end{minipage}
          \ifArticle
          \begin{minipage}{.4\textwidth}
          \else
          \begin{minipage}{.29\textwidth}
            \fi
            \begin{prooftree}
              \AxiomC{}
              \LeftLabel{(\botr)}
              \UnaryInfC{$\havoctriple{\aenv}{\predfalse}{\are}{\psi}$}
            \end{prooftree}
          \end{minipage}

          \vspace*{.5\baselineskip}
          \begin{prooftree}
            \AxiomC{}
            \LeftLabel{(\inter)}
            \RightLabel{$\begin{array}{ll}
                {\alpha = \interac{x_1}{p_1}{x_2}{p_2}} \\[-.5mm]
                {J \subseteq \interv{1}{2}} \\[-.5mm]
              \end{array}$}
            \UnaryInfC{$\havoctriple{\aenv}{
                \alpha * \Asterisk_{\!\!\!\!j\in J} ~\compin{x_j}{q_j}
              }{
                \interof{\alpha}
              }{
                \alpha * 
                \Asterisk_{\!\!\!\!j \in J}\bigvee_{q_j \arrow{p_j}{} q'_j}
                \compin{x_j}{q'_j}
              }$}
          \end{prooftree}

          \centerline{a. Axioms} 
          \vspace*{.5\baselineskip}

          \ifArticle
          \begin{minipage}{.6\textwidth}
          \else
          \begin{minipage}{.5\textwidth}
          \fi
            \begin{prooftree}
              \AxiomC{$\havoctriple{\aenv \setminus \set{\interof{\alpha}}}{\phi}{\are}{\psi}$}
              \LeftLabel{(\ii)}
              \RightLabel{$\begin{array}{l}
                  {\alpha = \interac{x_1}{p_1}{x_2}{p_2}} \\[-.5mm]
                  {\interof{\alpha} \in \aenv \setminus \supp{\are}}
                \end{array}$}
              \UnaryInfC{$\havoctriple{\aenv}{\phi * \alpha}{\are}{\psi * \alpha}$}
            \end{prooftree}
          \end{minipage}
          \begin{minipage}{.49\textwidth}
            \begin{prooftree}
              \AxiomC{$\havoctriple{\aenv \cup \set{\interof{\alpha}}}{\phi * \alpha}{\are}{\psi * \alpha}$}
              \LeftLabel{(\ie)}
              \RightLabel{$\begin{array}{l}
                  {\alpha = \interac{x_1}{p_1}{x_2}{p_2}} \\[-.5mm]
                  {\excluded{\phi}{\alpha}}
                \end{array}$}
              \UnaryInfC{$\havoctriple{\aenv}{\phi}{\are}{\psi}$}
            \end{prooftree}
          \end{minipage}


          \vspace*{.5\baselineskip}
          \centerline{b. Redundancy Rules}
          \vspace*{.5\baselineskip}
          
          \begin{prooftree}
            \AxiomC{$\havoctriple{\aenv_i}{\phi_i * \front{\phi_i}{\phi_{3-i}}}{\are_i}{\psi_i * \front{\phi_i}{\phi_{3-i}}} \mid i=1,2$}
            \LeftLabel{(\parr)}
            \RightLabel{$\begin{array}{l}
                {\aenv_i = \interof{\phi_i * \front{\phi_i}{\phi_{3-i}}}} \\[-.5mm]
                {i=1,2}
              \end{array}$}
            \UnaryInfC{$\havoctriple{\aenv_1 \cup \aenv_2}{\phi_1 * \phi_2}{\are_1 \parcomp_{\aenv_1,\aenv_2} \are_2}{\psi_1 * \psi_2}$}
          \end{prooftree}
          
          \vspace*{.5\baselineskip}
          \centerline{c. Composition Rule}
          \vspace*{.5\baselineskip}

          \begin{minipage}{.75\textwidth}
            \begin{prooftree}
              \AxiomC{$\havoctriple{\aenv}{\phi}{\are_1}{\varphi}$}
              \AxiomC{$\havoctriple{\aenv}{\varphi}{\are_2}{\psi}$}
              \LeftLabel{($\cdot$)}
              \RightLabel{$\phi \strucgeq \varphi$}
              \BinaryInfC{$\havoctriple{\aenv}{\phi}{\are_1 \cdot \are_2}{\psi}$}
            \end{prooftree}
          \end{minipage}
          \begin{minipage}{.24\textwidth}
            \begin{prooftree}
              \AxiomC{$\havoctriple{\aenv}{\phi}{\are}{\phi}$}
              \LeftLabel{($*$)}
              \UnaryInfC{$\havoctriple{\aenv}{\phi}{\are^*}{\phi}$}
            \end{prooftree}
          \end{minipage}
          \vspace*{.5\baselineskip}

          \ifArticle
          \begin{minipage}{.5\textwidth}
          \else          
          \begin{minipage}{.4\textwidth}
          \fi
          \begin{prooftree}
              \AxiomC{$\havoctriple{\aenv}{\phi}{\are_1}{\psi}$}
              \AxiomC{$\havoctriple{\aenv}{\phi}{\are_2}{\psi}$}
              \LeftLabel{(\ui)}
              \BinaryInfC{$\havoctriple{\aenv}{\phi}{\are_1 \cup \are_2}{\psi}$}
            \end{prooftree}
          \end{minipage}
          \ifArticle
          \begin{minipage}{.4\textwidth}
          \else
          \begin{minipage}{.28\textwidth}
            \fi
            \begin{prooftree}
              \AxiomC{$\havoctriple{\aenv}{\phi}{\are_1 \cup \are_2}{\psi}$}
              \LeftLabel{(\ue)}
              \UnaryInfC{$\havoctriple{\aenv}{\phi}{\are_1}{\psi}$}
            \end{prooftree}
          \end{minipage}
          \ifArticle
          \begin{minipage}{.4\textwidth}
          \else
          \begin{minipage}{.26\textwidth}
            \fi
            \begin{prooftree}
              \AxiomC{$\havoctriple{\aenv}{\phi}{\are_1}{\psi}$}
              \RightLabel{$\are_1 \requiv{\phi} \are_2$}
              \LeftLabel{(\congr)}
              \UnaryInfC{$\havoctriple{\aenv}{\phi}{\are_2}{\psi}$}
            \end{prooftree}
          \end{minipage}

          \vspace*{.5\baselineskip}
          \centerline{d. Regular Expression Rules}          
          \begin{minipage}{.37\textwidth}
          \begin{prooftree}
            \AxiomC{$\begin{array}{c} \\\\\\
                \havoctriple{\aenv}{\phi}{\are}{\psi'}
              \end{array}$}
            \LeftLabel{(\conseq)}
            \RightLabel{${\psi' \models \psi}$}
            \UnaryInfC{$\havoctriple{\aenv}{\phi}{\are}{\psi}$}
          \end{prooftree}
          \end{minipage}
          \begin{minipage}{.6\textwidth}
          \begin{prooftree}
            \AxiomC{$\havoctriple{\aenv'}{\phi * \varphi'}{\are'}{\psi}
              ~\left| \begin{array}{l} 
                {\apred(x_1,\ldots,x_{\#(\apred)}) \unfoldrule \exists \vec{z} ~.~ \varphi ~\in~ \asid} \\
                {(\exists \vec{z} ~.~ \varphi)[x_1/y_1,\ldots,x_{\#(\apred)}/y_{\#(\apred)}] = \exists \vec{z} ~.~ \varphi'} \\
                {\aenv' = \big(\aenv \setminus \set{\interof{\apred(y_1,\ldots,y_{\#(\apred)})}}\big) \cup \interof{\varphi'}} \\
                {\are' = \are\big[\interof{\apred(x_1,\ldots,x_{\#(\apred)})}~/~\interof{\varphi'}\big]}
              \end{array}\right.$}
            \LeftLabel{(\lu)}              
            \UnaryInfC{$\havoctriple{\aenv}{\phi * \apred(y_1,\ldots,y_{\#(\apred)})}{\are}{\psi}$}
          \end{prooftree}
          \end{minipage}
          \vspace*{.5\baselineskip}
          \begin{minipage}{.54\textwidth}
            \begin{prooftree}
              \AxiomC{$\havoctriple{\aenv}{\phi_i}{\are}{\psi_i} \mid i \in \interv{1}{k}$}
              \LeftLabel{($\vee$)}
              \RightLabel{$\begin{array}{l}
                  {\phi_i \struceq \phi_j} \\[-.5mm]
                  {i\neq j \in \interv{1}{k}} 
                \end{array}$
              }
              \UnaryInfC{$\havoctriple{\aenv}{\bigvee_{i=1}^k \phi}{\are}{\bigvee_{i=1}^k\psi_i}$}
            \end{prooftree}
          \end{minipage}
          \begin{minipage}{.45\textwidth}
            \begin{prooftree}
              \AxiomC{$\havoctriple{\aenv}{\phi_i}{\are}{\psi_i} \mid i \in \interv{1}{k}$}
              \LeftLabel{($\wedge$)}
              \RightLabel{$\begin{array}{l}
                  {\phi_i \struceq \phi_j} \\[-.5mm]
                  {i\neq j \in \interv{1}{k}}
                \end{array}$
              }
              \UnaryInfC{$\havoctriple{\aenv}{\bigwedge_{i=1}^k \phi_i}{\are}{\bigwedge_{i=1}^k\psi_i}$}
            \end{prooftree}
          \end{minipage}

          \vspace*{.5\baselineskip}
          \centerline{e. Structural Rules} 
        \end{center}}
        \vspace*{-.5\baselineskip}
\end{figure}


The composition rule (\parr) splits a proof obligation into two
simpler havoc triples (Fig. \ref{fig:havoc-rules}c). The pre- and
postconditions of the premisses are subformul{\ae} of the pre- and
postcondition of the conclusion, joined by separating conjunction and
extended by so-called \emph{frontier} formul{\ae}, describing those
sets of interaction atoms that may cross the boundary between the two
separated conjuncts. The frontier formul{\ae} play the role of
environment assumptions in a rely/assume-guarantee style of reasoning
\cite{Owicki1978,DBLP:phd/ethos/Jones81}. They are required for
soundness, under the semantics of open state changes
(Def. \ref{def:open}), which considers that the interactions can fire
anytime, unless they are explicitly disabled by some component from
$\phi_i$, for $i=1,2$.

Nevertheless, defining the frontier syntactically faces the following
problem: interactions introduced by a predicate atom in $\phi_i$, can
impact the state of a component defined by $\phi_{3-i}$. We tackle
this problem by forbidding predicate atoms that describe
configurations with loose ports, that belong to components lying
outside of the current configuration. We recall that a configuration
$\config$ is tight if and only if, for each interaction $(c_1, p_1,
c_2, p_2) \in \interacs$, we have $c_1, c_2 \in \comps$. Moreover, we
say that a formula $\varphi$ is \emph{tight} if and only if every
model of $\varphi$ is tight. For instance, a predicate atom
$\chain{h}{t}(x,y)$, for given $h,t \geq 0$ (Example \ref{ex:chain})
is tight, because, in each model, the interactions involve only the
$\mathit{out}$ and $\mathit{in}$ ports of adjacent components from the
configuration.

\begin{definition}[Frontier]\label{def:frontier}
  Given symbolic configurations $\phi_1$ and $\phi_2$, the
  \emph{frontier of $\phi_i$ and $\phi_{3-i}$} is the formula
  \(\front{\phi_i}{\phi_{3-i}} \isdef \Asterisk_{\!\!\alpha \in
    \iatoms{\phi_{3-i}} \setminus (\iatoms{\overline{\phi}_{3-i}} \cup
    \iatoms{\phi_i})} ~\alpha\), where $\overline{\phi}_i$ is the
  largest tight subformula of $\phi_i$, for $i = 1,2$.
\end{definition}

\begin{example}\label{ex:frontier}
Let $\phi_1 = \chain{h}{t}(x,y) * \interac{y}{out}{z}{in}$ and $\phi_2
= \chain{h}{t}(y,z) * \interac{x}{out}{y}{in}$. We have
$\front{\phi_1}{\phi_2} = \interac{x}{out}{y}{in}$ and
$\front{\phi_2}{\phi_1} = \interac{y}{out}{z}{in}$, because the
tightness of $\chain{h}{t}(x,y)$ and $\chain{h}{t}(y,z)$ means that
the only interactions crossing the boundary of $\phi_1$ and $\phi_2$
are the ones described by $\interac{y}{out}{z}{in}$ and
$\interac{x}{out}{y}{in}$. \hfill$\blacksquare$
\end{example}

Finally, the regular expression of the conclusion of
the (\parr) rule is the interleaving of the regular expressions from
the premisses, taken with respect to the sets of alphabet symbols
$\aenv_i = \interof{\phi_i * \front{\phi_i}{\phi_{3-i}}}$, for
$i=1,2$. 

The rules in Fig. \ref{fig:havoc-rules}d introduce regular expressions
built using concatenation, Kleene star and union. In particular, for
reasons related to the soundness of the proof system, the
concatenation rule ($\cdot$) applies to havoc triples whose
preconditions are finite disjunctions of symbolic configurations,
sharing the same structure of component, interaction and predicate
atoms, whereas the cut formul{\ae} (postcondition of the left and
precondition of the right premisse) share the same structure as the
precondition. We formalize below the fact that two formul{\ae} share
the same structure:

\begin{definition}\label{def:same-struc}
  Two formul{\ae} $\phi$ and $\psi$ \emph{share the same structure},
  denoted $\phi \struceq \psi$ if and only if they become equivalent
  when every component atom $\compin{x}{q}$ is replaced by the formula
  $\company{x}$, in both $\phi$ and $\psi$. We write $\phi \strucgeq
  \psi$ if and only if $\phi$ is satisfiable and $\psi$ is not, or
  else $\phi \struceq \psi$.
\end{definition}

The (\ue) rule is the dual of (\ui), that restricts the language from
the conclusion to a subset of the one from the premisse. As a remark,
by applying the (\ie) and (\ue) rules in any order, one can derive the
havoc invariance of the intermediate assertions in a single-reversal
reconfiguration sequence (see Def. \ref{def:single-reversal} and
Prop. \ref{prop:single-reversal}). The rule (\congr) substitutes a
regular expression with a congruent one, with respect to the
precondition.

Last, the rules in Fig. \ref{fig:havoc-rules}e modify the structure of
the pre- and postconditions. In particular, the left unfolding rule
(\lu) has a premisse for each step of unfolding of a predicate atom
from the conclusion's precondition, with respect to a rule from the
SID. The environment and the regular expression in each premisse are
obtained by replacing the alphabet symbol of the unfolded predicate
symbol by the set of alphabet symbols from the unfolding step, where
$\are[\interof{\alpha}/\are']$ denotes the regular expression obtained
by replacing each occurrence of the alphabet symbol $\interof{\alpha}$
in $\are$ with the regular expression $\are'$.

\subsection{Havoc Proofs}

A \emph{proof tree} is a finite tree $T$ whose nodes are labeled by
havoc triples and, for each node $n$ not on the frontier of $T$, the
children of $n$ are the premisses of the application of a rule from
Fig. \ref{fig:havoc-rules}, whose conclusion is the label of $n$. For
the purposes of this paper, we consider only proof trees that meet the
following condition:

\begin{assumption}\label{ass:root}
The root of the proof tree is labeled by a havoc triple
$\havoctriple{\aenv}{\phi}{\are}{\psi}$, such that $\phi$ is a
symbolic configuration and $\aenv = \set{\interof{\alpha} \mid \alpha
  \in \atoms{\phi}}$.
\end{assumption}
It is easy to check that the above condition on the shape of the
precondition and the relation between the precondition and the
environment holds recursively, for the labels of all nodes in a proof
tree that meets assumption \ref{ass:root}. Before tackling the
soundness of the havoc proof system (Fig. \ref{fig:havoc-rules}), we
state an invariance property of the environments of havoc triples that
occur in a proof tree:

\begin{definition}\label{def:distinctive}
  A havoc triple $\havoctriple{\aenv}{\phi}{\are}{\psi}$ is
  \emph{distinctive} if and only if
  $\semlang{\interof{\alpha_1}}{\aconfig} \cap
  \semlang{\interof{\alpha_2}}{\aconfig} = \emptyset$, for all
  $\interof{\alpha_1}, \interof{\alpha_2} \in \aenv$ and all $\aconfig
  \in \sem{\phi}{}$.
\end{definition}

The next lemma is proved inductively on the structure of the proof
tree, using Assumption \ref{ass:root}.

\begin{restatable}{lemma}{LemmaDistinctive}\label{lemma:distinctive}
  Given a proof tree $T$, each node in $T$ is labeled with a
  distinctive havoc triple.
\end{restatable}

In order to deal with inductively defined predicates that occur within
the pre- and postconditions of the havoc triples, we use cyclic proofs
\cite{BrotherstonS11}. A \emph{cyclic proof tree} $T$ is a proof tree
such that every node on the frontier is either the conclusion of an
axiom in Fig. \ref{fig:havoc-rules}a, or there is another node $m$
whose label matches the label of $n$ via a substitution of variables;
we say that $n$ is a \emph{bud} and $m$ is its \emph{companion}. A
cyclic proof tree is a \emph{cyclic proof} if and only if every
infinite path through the proof tree extended with bud-companion
edges, goes through the conclusion of a (\lu) rule infinitely
often\footnote{This condition can be effectively decided by checking
  the emptiness of a B\"uchi automaton \cite{BrotherstonS11}.}. We
denote by $\Vdash \havoctriple{\aenv}{\phi}{\are}{\psi}$ the fact that
$\havoctriple{\aenv}{\phi}{\are}{\psi}$ labels the root of a cyclic
proof and state the following soundness theorem:

\begin{restatable}{theorem}{ThmHavocSoundness}\label{thm:havoc-soundness}
  If $\Vdash \havoctriple{\aenv}{\phi}{\are}{\psi}$ then
  $\models \havoctriple{\aenv}{\phi}{\are}{\psi}$.
\end{restatable}
The proof is by induction on the structure of the proof tree, using
Lemma \ref{lemma:distinctive}.

\subsection{A Havoc Proof Example}
\label{sec:havoc-proof}

We demonstrate the use of the proof system in
Fig. \ref{fig:havoc-rules} on the havoc invariance side conditions
required by the reconfiguration proofs from
\S\ref{sec:reconfiguration-proofs}. In fact, we prove a more general
statement, namely that $\chain{h}{t}(x,y)$ is havoc invariant, for all
$h,t\geq0$. An immediate consequence is that $\chain{1}{1}(z,x)$ is
havoc invariant. In particular, the havoc invariance proof for
$\compin{y}{\toknotok} * \interac{y}{out}{z}{in} * \chain{1}{1}(z,x)$
is an instance of the subgoal (\textbf{A}) below, whereas the proof
for $\compin{y}{\toknotok} * \chain{1}{1}(z,x)$ can be obtained by
applying rules (\ie) and (\ue) to (\textbf{A}), for $h=t=1$.
  
For space reasons, we introduce backlinks from buds to companions
whose labels differ by a renaming of free variables and of the $h$ and
$t$ indices in $\chain{h}{t}$, such that each pair $(h',t')$ in the
label of a companion is lexicographically smaller or equal to a pair
$(h,t)$ in the bud. This is a compact (folded) representation of a
cyclic proof tree, obtained by repeatedly appending the subtree rooted
at the companion to the bud, until all buds are labeled with triples
that differ from their companion's only by a renaming of free
variables\footnote{This is bound to happen, because a pair $(h,t)$ of
positive integers cannot be decreased indefinitely. }. Note that such
folding is only possible because the definitions of
$\chain{h}{t}(x,y)$ and $\chain{h'}{t'}(x,y)$, for $h,t,h',t' \geq 1$
are the same, up to the indices of the predicate symbols (Example
\ref{ex:chain}).

{\footnotesize
  \begin{prooftree}
    \AxiomC{}
    \LeftLabel{(\epsilonr)}
    \havocline{\emptyset}{\company{x}}{\epsilon}{\company{x}}
    \AxiomC{}
    \LeftLabel{(\epsilonr)}
    \havocline{\emptyset}{\compin{x}{\toknotok}}{\epsilon}{\compin{x}{\toknotok}}
    \AxiomC{}
    \LeftLabel{(\epsilonr)}
    \havocline{\emptyset}{\compin{x}{\toktoken}}{\epsilon}{\compin{x}{\toktoken}}
    \AxiomC{\textbf{(A)\qquad (B)}}
    \LeftLabel{(\lu)}
    \QuaternaryInfC{$
      \havoctriple{\set{\interof{\chain{h}{t}(z,x)}}}
      {\chain{h}{t}(z,x)}
      {\interof{\chain{h}{t}(z,x)}^*}
      {\chain{h}{t}(z,x)}$ \textbf{ (1)}}
  \end{prooftree}
}

\noindent In the proof of the subgoal \textbf{(A)} below, alphabet
symbols are abbreviated as $\interofzy \isdef
\interof{\interac{z}{out}{y}{in}}$ and $\interofyxch \isdef
\interof{\chain{h\dot{-}1}{t}(y,x)}$. We use the following congruence (Def. \ref{def:congruence}): 
\[(\interofzy \cup
\interofyxch)^* \requiv{\compin{z}{\toknotok} *
  \interac{z}{out}{y}{in} * \chain{h\dot{-}1}{t}(y,x)} {\interofyxch^*
  \cup [\interofyxch^* \cdot \interofzy \cdot (\interofzy \cup
    \interofyxch)^*]}\] The rule (\conseq) strenghtens the
postcondition $\chain{h}{t}(z,x)$ to an unfolding $\chain{h}{t}(z,x)
\unfold{} \exists y ~.~ \compin{z}{\toknotok} *
\interac{z}{out}{y}{in} * \chain{h\dot{-}1}{t}(y,x)$, whose
existentially quantified variable is, moreover, bound to the free
variable $y$ from the precondition. The frontier formul{\ae} in the
application of rule (\parr) are
$\front{\compin{z}{\toknotok}}{\chain{h\dot{-}1}{t}(y,x)} =
\front{\chain{h\dot{-}1}{t}(y,x)}{\compin{z}{\toknotok}} = \emp$.

{\footnotesize
  \begin{prooftree}
    \AxiomC{}
    \LeftLabel{(\epsilonr)}
    \havocline{\emptyset}{\compin{z}{\toknotok}}{\epsilon}{\compin{z}{\toknotok}}

    \AxiomC{backlink to \textbf{(1)}}
    \havocline{\set{\interofyxch}}
      {\chain{h\dot{-}1}{t}(y,x)}
      {\interofyxch^*}
      {\chain{h\dot{-}1}{t}(y,x)}

    \LeftLabel{(\parr)}
    \binhavocline{\set{\interofyxch}}
      {\compin{z}{\toknotok} \ast \chain{h\dot{-}1}{t}(y,x)}
      {\interofyxch^*}
      {\compin{z}{\toknotok} \ast \chain{h\dot{-}1}{t}(y,x)}
    \LeftLabel{(\ii)}
    \UnaryInfC{\textbf{(A1) }
      $\havoctriple{\set{\interofzy, \interofyxch}}
      {\compin{z}{\toknotok} \ast \interac{z}{out}{y}{in} \ast \chain{h\dot{-}1}{t}(y,x)}
      {\interofyxch^*}
      {\compin{z}{\toknotok} \ast \interac{z}{out}{y}{in} \ast \chain{h\dot{-}1}{t}(y,x)}$
    }
    \LeftLabel{(\conseq)}
    \havocline{\set{\interofzy, \interofyxch}}
      {\compin{z}{\toknotok} \ast \interac{z}{out}{y}{in} \ast \chain{h\dot{-}1}{t}(y,x)}
      {\interofyxch^*}
      {\chain{h}{t}(z,x)}

    \AxiomC{\textbf{(A2)}}
    \LeftLabel{(\ui)}
    \binhavocline{\set{\interofzy, \interofyxch}}
      {\compin{z}{\toknotok} \ast \interac{z}{out}{y}{in} \ast \chain{h\dot{-}1}{t}(y,x)}
      {\interofyxch^* \cup [\interofyxch^*
    \cdot \interofzy \cdot (\interofzy \cup \interofyxch)^*]}
      {\chain{h}{t}(z,x)}
    \LeftLabel{(\congr)}
    \havocline{\textbf{(A)}~ \set{\interofzy, \interofyxch}}
      {\compin{z}{\toknotok} \ast \interac{z}{out}{y}{in} \ast \chain{h\dot{-}1}{t}(y,x)}
      {(\interofzy \cup \interofyxch)^*}
      {\chain{h}{t}(z,x)}
  \end{prooftree}

  \vspace*{.5\baselineskip}
  \begin{prooftree}
    \AxiomC{backlink to \textbf{(A1)}}
    \def\defaultHypSeparation{\hskip .1in}
    \def\ScoreOverhang{2pt}
    \havoctwolines{\set{\interofzy, \interofyxch}}
      {\compin{z}{\toknotok} \ast \interac{z}{out}{y}{in} \ast \chain{h\dot{-}1}{t}(y,x)}
      {\interofyxch^*}
      {\compin{z}{\toknotok} \ast \interac{z}{out}{y}{in} \ast \chain{h\dot{-}1}{t}(y,x)}

    \AxiomC{\textbf{(A3)}}

    \AxiomC{\textbf{(A4)}}
    \LeftLabel{($\cdot$)}
    \trihavocline{\set{\interofzy, \interofyxch}}
      {\compin{z}{\toknotok} \ast \interac{z}{out}{y}{in} \ast \chain{h\dot{-}1}{t}(y,x)}
      {\interofyxch^* \cdot \interofzy
        \cdot (\interofzy \cup \interofyxch)^*}
      {\predfalse}
    \LeftLabel{(\conseq)}  
    \havocline{\textbf{(A2)}~ \set{\interofzy, \interofyxch}}
      {\compin{z}{\toknotok} \ast \interac{z}{out}{y}{in} \ast \chain{h\dot{-}1}{t}(y,x)}
      {\interofyxch^* \cdot \interofzy
        \cdot (\interofzy \cup \interofyxch)^*}
      {\chain{h}{t}(z,x)}
  \end{prooftree}

  \vspace*{.5\baselineskip}
  \begin{prooftree}
          \AxiomC{}
          \LeftLabel{(\disr)}
          \havocline{\textbf{(A3)}~ \set{\interofzy, \interofyxch}}
                        {\compin{z}{\toknotok} \ast \interac{z}{out}{y}{in} \ast \chain{h\dot{-}1}{t}(y,x)}
                        {\interofzy}
                        {\predfalse}
  \end{prooftree}

  \vspace*{.5\baselineskip}
  \begin{prooftree}
    \AxiomC{}
    \LeftLabel{(\botr)}
    \havocline{\textbf{(A4)}~ \set{\interofzy, \interofyxch}}
              {\predfalse}
              {(\interofzy \cup \interofyxch)^*}
              {\predfalse}
              \LeftLabel{($\cdot$)}
  \end{prooftree}
}
\noindent For space reasons, the proof of the subgoal \textbf{(B)} is
provided as supplementary material.

\section{A Worked-out Example: Reconfigurable Tree Architectures}
\label{sec:trees}

In addition to token rings (Fig. \ref{fig:token-ring}), we apply our
method to reconfiguration scenarios of distributed systems with
tree-shaped architectures. Such (virtual) architectures are e.g. used
in flooding and leader election algorithms. They are applicable, for
instance, when every component in the system must notify a designated
controller, placed in the root of the tree, about an event that
involves each component from the frontier of the tree. Conversely, the
root component may need to notify the rest of the components. The tree
architecture guarantees that the notification phase takes time
$\bigO(\log n)$ in the number $n$ of components in the tree, when the
tree is balanced, i.e.\ the lengths of the longest and shortest paths
between the root and the frontier differ by at most a constant
factor. A reconfiguration of a tree places a designated component
(whose priority has increased dynamically) closer to the frontier
(dually, closer to the root) in order to receive the notification
faster. In balanced trees, reconfigurations involve
structure-preserving rotations. For instance, \emph{self-adjustable
  splay-tree networks} \cite{Schmid16} use the \emph{zig} (left
rotation), \emph{zig-zig} (left-left rotation) and \emph{zig-zag}
(left-right rotation) operations \cite{SleatorTarjan85} to move nodes
in the tree, while keeping the balance between the shortest and
longest paths.

Fig. \ref{fig:trees} shows a model of reconfigurable tree
architectures, in which each leaf component starts in state
$\stateleafbusy$ and sends a notification to its parent before
entering the $\stateleafidle$ state. An inner component starts in
state $\stateidle$ and waits for notifications from both its left
($\lrecv$) and right ($\rrecv$) children before sending a notification
to its parent ($\send$), unless this component is the root
(Fig. \ref{fig:trees}a). We model notifications by interactions of the form
 $\interac{\_}{\send}{\_}{\lrecv}$ and $\interac{\_}{\send}{\_}{\rrecv}$.
The notification phase is completed when the
root is in state $\stateright$, every inner component is in the
$\stateidle$ state and every leaf is in the $\stateleafidle$ state.

\begin{figure}[t!]
  \vspace*{-\baselineskip}
  \caption[Tree]{Reconfiguration of a Tree Architecture}
  \label{fig:trees}
  \begin{center}    
    \begin{minipage}{.2\textwidth}
      \centering
      \scalebox{0.6}{		\begin{tikzpicture}[>=stealth',shorten >=1pt,auto,node distance=2.5cm,
			every state/.style={minimum size=1.5cm}]

				\node[state] (i)      {$\stateidle$};
				\node[state, below of=i] (l)      {$\stateleft$};
				\node[state, right of=i] (r)      {$\stateright$};

				\node[state, above of=i, yshift=-15pt] (li)	{$\stateleafidle$};
				\node[state, right of=li] (lb) {$\stateleafbusy$};

				\path[->] (i) edge node[above, pos=0.5, rotate=-90] (i1) {$\lrecv$} (l);
				\path[->] (l) edge node[below, pos=0.3, rotate=45] (i2) {$\qquad\rrecv$} (r);
				\path[->] (r) edge node (i3) {$\send$} (i);

				\path[->] (lb) edge node[above] (i4) {$\send$} (li);

				\node (Cc) [draw=black, fit= (i) (l) (r) (li) (lb) (i1) (i2) (i3) (i4)] {};

		\end{tikzpicture}}

      \centerline{\footnotesize(a)}
    \end{minipage}
    \begin{minipage}{.79\textwidth}
      \vspace*{.5\baselineskip}
      \centering
      \scalebox{0.8}[0.7]{		\begin{tikzpicture}[main node/.style={
			minimum width = 3em, minimum height = 3em}
 		]
			\begin{scope}[local bounding box=box1, shift={(0,0)}]
				\node[draw, minimum size=8mm] at (0.75, 1.2) (x) {$x$};
				\node[draw, minimum size=8mm] at (1.75, 2.8) (y) {$y$};

				\node[circle, minimum size=1mm, inner sep=0pt, outer sep=0pt, fill=black]
					at (0, 0) (a) {};
				\node[circle, minimum size=1mm, inner sep=0pt, outer sep=0pt, fill=black]
					at (1.5, 0) (b) {};
				\node[circle, minimum size=1mm, inner sep=0pt, outer sep=0pt, fill=black]
					at (2.75, 1.6) (c) {};
				\node[circle, minimum size=1mm, inner sep=0pt, outer sep=0pt, fill=black]
					at (1.75, 4) (z) {};

				\node[circle, minimum size=1mm, inner sep=0pt, outer sep=0pt, fill=black]
					at (x.south west) (xa) {};
				\node[circle, minimum size=1mm, inner sep=0pt, outer sep=0pt, fill=black]
					at (x.south east) (xb) {};
				\node[circle, minimum size=1mm, inner sep=0pt, outer sep=0pt, fill=black]
					at (x.north) (xy) {};
				\node[circle, minimum size=1mm, inner sep=0pt, outer sep=0pt, fill=black]
					at (y.south west) (yx) {};
				\node[circle, minimum size=1mm, inner sep=0pt, outer sep=0pt, fill=black]
					at (y.south east) (yc) {};
				\node[circle, minimum size=1mm, inner sep=0pt, outer sep=0pt, fill=black]
					at (y.north) (yz) {};

				\node[below of=a, yshift=7mm] () {$a$};
				\node[below of=b, yshift=7mm] () {$b$};
				\node[below of=c, yshift=7mm] () {$c$};
				\node[above of=z, yshift=-7mm] () {$z$};

				\node[right of=z, xshift=2mm] () {$\lrecv$ / $\rrecv$};
				\node[left of=c, xshift=5mm] () {$\send$};
				\node[right of=b, xshift=-5mm] () {$\send$};
				\node[left of=a, xshift=5mm] () {$\send$};

				\node[right of=yz, yshift=2mm, xshift=-6mm] () {$\send$};
				\node[right of=yc, yshift=0mm, xshift=-4mm] () {$\rrecv$};
				\node[left of=yx, yshift=0mm, xshift=4mm] () {$\lrecv$};
				\node[left of=xy, yshift=2mm, xshift=6mm] () {$\send$};
				\node[right of=xb, yshift=0mm, xshift=-4mm] () {$\rrecv$};
				\node[left of=xa, yshift=0mm, xshift=4mm] () {$\lrecv$};

				\draw[-] (a) -- (xa);
				\draw[-] (b) -- (xb);
				\draw[-] (xy) -- (yx);
				\draw[-] (yc) -- (c);
				\draw[-] (yz) -- (z);
			\end{scope}

			\begin{scope}[shift={(5,0)}]
				\node at (0, 2) {$\xRightarrow[\text{Rotation}]{\text{Right}}$};
			\end{scope}

			\begin{scope}[local bounding box=box2, shift={(7,0)}]
				\node[draw, minimum size=8mm] at (1.25, 2.8) (x) {$x$};
				\node[draw, minimum size=8mm] at (2.25, 1.2) (y) {$y$};

				\node[circle, minimum size=1mm, inner sep=0pt, outer sep=0pt, fill=black]
					at (0.25, 1.6) (a) {};
				\node[circle, minimum size=1mm, inner sep=0pt, outer sep=0pt, fill=black]
					at (1.5, 0) (b) {};
				\node[circle, minimum size=1mm, inner sep=0pt, outer sep=0pt, fill=black]
					at (3, 0) (c) {};
				\node[circle, minimum size=1mm, inner sep=0pt, outer sep=0pt, fill=black]
					at (1.25, 4) (z) {};

				\node[circle, minimum size=1mm, inner sep=0pt, outer sep=0pt, fill=black]
					at (x.south west) (xa) {};
				\node[circle, minimum size=1mm, inner sep=0pt, outer sep=0pt, fill=black]
					at (x.south east) (xy) {};
				\node[circle, minimum size=1mm, inner sep=0pt, outer sep=0pt, fill=black]
					at (x.north) (xz) {};
				\node[circle, minimum size=1mm, inner sep=0pt, outer sep=0pt, fill=black]
					at (y.south west) (yb) {};
				\node[circle, minimum size=1mm, inner sep=0pt, outer sep=0pt, fill=black]
					at (y.south east) (yc) {};
				\node[circle, minimum size=1mm, inner sep=0pt, outer sep=0pt, fill=black]
					at (y.north) (yx) {};

				\node[below of=a, yshift=7mm] () {$a$};
				\node[below of=b, yshift=7mm] () {$b$};
				\node[below of=c, yshift=7mm] () {$c$};
				\node[above of=z, yshift=-7mm] () {$z$};

				\node[right of=z, xshift=2mm] () {$\lrecv$ / $\rrecv$};
				\node[right of=c, xshift=-5mm] () {$\send$};
				\node[left of=b, xshift=5mm] () {$\send$};
				\node[right of=a, xshift=-5mm] () {$\send$};

				\node[right of=xz, yshift=2mm, xshift=-6mm] () {$\send$};
				\node[right of=yc, yshift=0mm, xshift=-4mm] () {$\rrecv$};
				\node[right of=yx, yshift=2mm, xshift=-5mm] () {$\send$};
				\node[right of=xy, yshift=0mm, xshift=-5mm] () {$\rrecv$};
				\node[left of=yb, yshift=0mm, xshift=4mm] () {$\lrecv$};
				\node[left of=xa, yshift=0mm, xshift=4mm] () {$\lrecv$};

				\draw[-] (a) -- (xa);
				\draw[-] (xy) -- (yx);
				\draw[-] (b) -- (yb);
				\draw[-] (c) -- (yc);
				\draw[-] (xz) -- (z);
			\end{scope}
		\end{tikzpicture}}

      \vspace*{.3\baselineskip}
      \centerline{\footnotesize(b)}
    \end{minipage}

    \begin{minipage}{.52\textwidth}
  {\small\[\begin{array}{l}
    \treeidle(x) \leftarrow \compin{x}{\stateleafidle} \\
    \treeidle(x) \leftarrow \exists y \exists z ~.~ \compin{x}{\stateidle} * \interac{y}{\send}{x}{\lrecv} ~* \\ 
    \hspace*{1.8cm} \interac{z}{\send}{x}{\rrecv} * \treeidle(y) * \treeidle(z) \\
    \\
    \treenotidle(x) \leftarrow \compin{x}{\stateleafbusy} \\
    \treenotidle(x) \leftarrow \exists y \exists z ~.~ \compin{x}{\stateleft} * \interac{y}{\send}{x}{\lrecv} ~* \\ 
    \hspace*{1.8cm} \interac{z}{\send}{x}{\rrecv} * \treeidle(y) * \treenotidle(z) \\
    \treenotidle(x) \leftarrow \exists y \exists z ~.~ \compin{x}{\stateright} * \interac{y}{\send}{x}{\lrecv} ~* \\ 
    \hspace*{1.8cm} \interac{z}{\send}{x}{\rrecv} * \treeidle(y) * \treeidle(z) \\
    \treenotidle(x) \leftarrow \exists y \exists z ~.~ \compin{x}{\stateidle} * \interac{y}{\send}{x}{\lrecv} ~* \\
    \hspace*{1.8cm} \interac{z}{\send}{x}{\rrecv} * \treenotidle(y) * \treenotidle(z) 
    \end{array}\]}
\end{minipage}
\begin{minipage}{.47\textwidth}
{\small\[\begin{array}{l}
  \tree(x) \leftarrow \compin{x}{\stateleafidle} \\
  \tree(x) \leftarrow \compin{x}{\stateleafbusy} \\
  \tree(x) \leftarrow \exists y \exists z ~.~ \company{x} * \interac{y}{\send}{x}{\lrecv}  ~* \\
  \hspace*{1.4cm} \interac{z}{\send}{x}{\rrecv} * \tree(y) * \tree(z) \\
  \\
  \treeseg(x,x) \leftarrow \company{x} \\
  \treeseg(x,u) \leftarrow \exists y \exists z ~.~ \company{x} * \interac{y}{\send}{x}{\lrecv} ~* \\  
  \hspace*{1.6cm} \interac{z}{\send}{x}{\rrecv} * \treeseg(y,u) * \tree(z) \\
  \treeseg(x,u) \leftarrow \exists y \exists z ~.~ \company{x} * \interac{y}{\send}{x}{\lrecv} ~* \\
  \hspace*{1.6cm} \interac{z}{\send}{x}{\rrecv} * \tree(y) * \treeseg(z,u)  \\\\
  \end{array}\]}

\vspace*{-1.5\baselineskip}
\end{minipage}
\end{center}
\vspace*{-.5\baselineskip}
\centerline{\footnotesize(c)}
\vspace*{-\baselineskip}
\end{figure}
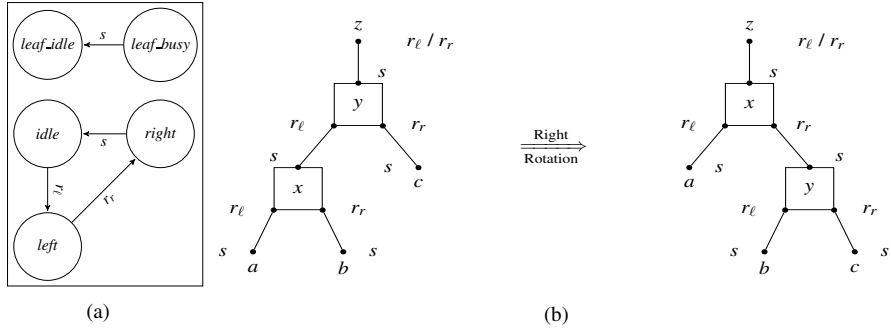

\begin{figure}[t!]
  \vspace*{-\baselineskip}
  \caption{Proof of a Tree Rotation}
  \label{fig:tree-proof}
  \begin{center}
    {\footnotesize
      \begin{lstlisting}
$\textcolor{violet}{\Set{\begin{array}{l}
  \exists r \exists x \exists y \exists z \exists a \exists b \exists c ~.~ \treeseg(r,z) * \interac{a}{\send}{x}{\lrecv} * \interac{c}{\send}{y}{\rrecv} * \interac{y}{\send}{z}{\lrecv} * \interac{x}{\send}{y}{\lrecv} * \interac{b}{\send}{x}{\rrecv} * \\
  \hspace{2.7cm} \compin{x}{\stateidle} * \compin{y}{\stateidle} * \treenotidle(a) * \treenotidle(b) * \treenotidle(c)
\end{array}}}$
with $x,y,z,a,b,c : \interac{a}{\send}{x}{\lrecv} * \interac{c}{\send}{y}{\rrecv} * \interac{y}{\send}{z}{\lrecv} * \interac{x}{\send}{y}{\lrecv} *  \interac{b}{\send}{x}{\rrecv} * \compin{x}{\stateidle} * \compin{y}{\stateidle} $ do
$\textcolor{violet}{\Set{\begin{array}{l}
    \treeseg(r,z) * \interac{a}{\send}{x}{\lrecv} * \interac{c}{\send}{y}{\rrecv} * \interac{y}{\send}{z}{\lrecv} * \interac{x}{\send}{y}{\lrecv} * \interac{b}{\send}{x}{\rrecv} ~* \\
    \compin{x}{\stateidle} * \compin{y}{\stateidle} * \treenotidle(a) * \treenotidle(b) * \treenotidle(c)
\end{array}}}$
    disconnect($b.\send$,$x.\rrecv$);
$\textcolor{violet}{\Set{\begin{array}{l}
    \treeseg(r,z) * \interac{a}{\send}{x}{\lrecv} * \interac{c}{\send}{y}{\rrecv} * \interac{y}{\send}{z}{\lrecv} * \interac{x}{\send}{y}{\lrecv} ~* \\
    (\compin{x}{\stateidle} * \treenotidle(a) \vee \compin{x}{\stateleft} * \treeidle(a)) * \compin{y}{\stateidle} *  \treenotidle(b) * \treenotidle(c)         
\end{array}}}$ $\hinv$
    disconnect($x.\send$,$y.\lrecv$);
$\textcolor{violet}{\Set{\begin{array}{l}
     \treeseg(r,z) * \interac{a}{\send}{x}{\lrecv} * \interac{c}{\send}{y}{\rrecv} * \interac{y}{\send}{z}{\lrecv}  * \\
     (\compin{x}{\stateidle} * \treenotidle(a) \vee \compin{x}{\stateleft} * \treeidle(a)) * \compin{y}{\stateidle} *  \treenotidle(b) * \treenotidle(c)         
\end{array}}}$ 
    disconnect($y.\send$,$z.\lrecv$);
$\textcolor{violet}{\Set{\begin{array}{l}
     \treeseg(r,z) * \interac{a}{\send}{x}{\lrecv} * \interac{c}{\send}{y}{\rrecv} * \\
     (\compin{x}{\stateidle} * \treenotidle(a) \vee \compin{x}{\stateleft} * \treeidle(a)) * \compin{y}{\stateidle} *  \treenotidle(b) * \treenotidle(c)         
\end{array}}}$ 
    connect($b.\send$,$y.\lrecv$); 
$\textcolor{violet}{\Set{\begin{array}{l}
     \treeseg(r,z) * \interac{a}{\send}{x}{\lrecv} * \interac{c}{\send}{y}{\rrecv} * \interac{b}{\send}{y}{\lrecv} * \\
     (\compin{x}{\stateidle} * \treenotidle(a) \vee \compin{x}{\stateleft} * \treeidle(a)) * \\
     (\compin{y}{\stateidle} * \treenotidle(b) * \treenotidle(c) \vee \compin{y}{\stateleft} * \treeidle(b) * \treenotidle(c) \vee \compin{y}{\stateright} * \treeidle(b) * \treeidle(c))
\end{array}}}$ 
    connect($y.\send$,$x.\rrecv$);
$\textcolor{violet}{\Set{\begin{array}{l}
     \treeseg(r,z) * \interac{a}{\send}{x}{\lrecv} * \interac{c}{\send}{y}{\rrecv} * \interac{b}{\send}{y}{\lrecv} * \interac{y}{\send}{x}{\rrecv} * \\
     \big((\compin{x}{\stateidle} * \treenotidle(a) \vee \compin{x}{\stateleft} * \treeidle(a)) * \\
     (\compin{y}{\stateidle} * \treenotidle(b) * \treenotidle(c) \vee \compin{y}{\stateleft} * \treeidle(b) * \treenotidle(c) \vee \compin{y}{\stateright} * \treeidle(b) * \treeidle(c)) \big) \vee \\
     \compin{x}{\stateright} * \compin{y}{\stateidle} * \treeidle(a) * \treeidle(b) * \treeidle(c) 
\end{array}}}$ $\hinv$
    connect($x.\send$,$z.\lrecv$)
$\textcolor{violet}{\Set{\begin{array}{l}
     \treeseg(r,z) * \interac{a}{\send}{x}{\lrecv} * \interac{c}{\send}{y}{\rrecv} * \interac{b}{\send}{y}{\lrecv} * \interac{y}{\send}{x}{\rrecv} * \interac{x}{\send}{z}{\lrecv} * \\
     \big((\compin{x}{\stateidle} * \treenotidle(a) \vee \compin{x}{\stateleft} * \treeidle(a)) * \\
     (\compin{y}{\stateidle} * \treenotidle(b) * \treenotidle(c) \vee \compin{y}{\stateleft} * \treeidle(b) * \treenotidle(c) \vee \compin{y}{\stateright} * \treeidle(b) * \treeidle(c)) \big) \vee \\
     \compin{x}{\stateright} * \compin{y}{\stateidle} * \treeidle(a) * \treeidle(b) * \treeidle(c) 
         
\end{array}}}$ 
od
$\textcolor{violet}{\Set{\begin{array}{l}
     \exists r, x, y, z, a, b, c ~.~ \treeseg(r,z) * \interac{a}{\send}{x}{\lrecv} * \interac{c}{\send}{y}{\rrecv} * \interac{b}{\send}{y}{\lrecv} * \interac{y}{\send}{x}{\rrecv} * \interac{x}{\send}{z}{\lrecv} * \\
     \big((\compin{x}{\stateidle} * \treenotidle(a) \vee \compin{x}{\stateleft} * \treeidle(a)) * \\
     (\compin{y}{\stateidle} * \treenotidle(b) * \treenotidle(c) \vee \compin{y}{\stateleft} * \treeidle(b) * \treenotidle(c) \vee \compin{y}{\stateright} * \treeidle(b) * \treeidle(c)) \big) \vee \\
     \compin{x}{\stateright} * \compin{y}{\stateidle} * \treeidle(a) * \treeidle(b) * \treeidle(c) 
\end{array}}}$   
$\textcolor{violet}{\Set{\begin{array}{l}
    \exists r, x, y, z, a, b, c ~.~ \treeseg(r,z) * \interac{x}{\send}{z}{\lrecv} * \treenotidle(x) ~\wedge \\
    (\interac{a}{\send}{x}{\lrecv} * \interac{c}{\send}{y}{\rrecv} * \interac{b}{\send}{y}{\lrecv} * \interac{y}{\send}{x}{\rrecv} * \predtrue)
\end{array}}}$   
    \end{lstlisting}}
  \end{center}
  \vspace*{-2\baselineskip}
\end{figure}

Fig. \ref{fig:trees}b shows a \emph{right rotation} that reverses the
positions of components with identifiers $x$ and $y$, implemented by
the reconfiguration program from Fig. \ref{fig:tree-proof}. The
rotation applies only to configurations in which both $x$ and $y$ are
in state $\stateidle$, by distinguishing the case when $y$ is a left
or a right child of $z$. For simplicity, Fig. \ref{fig:tree-proof}
shows the program in case $y$ is a left child, the other case being
symmetric. Note that, applying the rotation in a configuration where
the component indexed by $x$ is in state $\stateright$
(both $a$ and $b$ have sent their notifications to $x$) and the one
indexed by $y$ is in state $\stateidle$ ($c$ has not yet sent its
notification to $y$) yields a configuration from which $c$ cannot send
its notification further, because $x$ has now become the root of the
subtree changed by the rotation (a similar scenario is when $y$ is in
state $\stateright$, $x$ is in state $\stateidle$ and $a$, $b$ and $c$
have sent their notifications to their parents).

We prove that, \emph{whenever a right rotation is applied to a tree,
  such that the subtrees rooted at $a$, $b$ and $c$ have not sent
  their notifications yet, the result is another tree in which the
  subtrees rooted at $a$, $b$ and $c$ are still waiting to submit
  their notifications}. This guarantees that the notification phase
will terminate properly with every inner component (except for the
root) in state $\stateidle$ and every leaf component in state
$\stateleafidle$, even if one or more reconfigurations take place in
between. In particular, this proves the correctness of more complex
reconfigurations of splay tree architectures, using e.g.\ the zig-zig
and zig-zag operations \cite{Schmid16}.

The proof in Fig. \ref{fig:tree-proof} uses the inductive definitions
from Fig. \ref{fig:trees}c. The predicates $\treeidle(x)$,
$\treenotidle(x)$ define trees where all components are idle, and
where some notifications are still being propagated, respectively. The
predicate $\tree(x)$ conveys no information about the states of the
components and the predicate $\treeseg(x,u)$ defines a tree segment,
from component $x$ to component $u$. To use the havoc proof system
from Fig. \ref{fig:havoc-rules}, we need the following
statement\footnote{This is similar to
  Prop. \ref{prop:chain-precise}.}:

\begin{restatable}{proposition}{PropTreePrecise}\label{prop:tree-precise}
  The set of symbolic configurations using predicate atoms
  $\treeidle(x)$, $\treenotidle(x)$, $\tree(x)$ and $\treeseg(x,y)$ is
  precisely closed.
\end{restatable}

Moreover, each predicate atom $\treeidle(x)$, $\treenotidle(x)$,
$\tree(x)$ and $\treeseg(x,y)$ is tight, because, in each model of
these atoms, the interactions $\interac{u}{\send}{v}{\lrecv}$ and $\interac{u}{\send}{v}{\rrecv}$
are between the ports $\tuple{\send,\lrecv}$ and
$\tuple{\send,\rrecv}$ of the components
$u$ and $v$, respectively.

The precondition of the reconfiguration program in
Fig. \ref{fig:tree-proof} states that $x$ and $y$ are idle components,
and the $a$, $b$ and $c$ subtrees are not idle, whereas the
postcondition states that the $x$ subtree is not idle. As mentioned,
\emph{this is sufficient to guarantee the correct termination of the
notification phase after the right rotation}. As in the proofs from
\S\ref{sec:reconfiguration-proofs}, proving the correctness of the
sequential composition of primitive commands requires proving the
havoc invariance of the annotations. However, since in this case, the
reconfiguration sequence is single-reversal (Def.
\ref{def:single-reversal}), we are left with proving havoc invariance
only for the annotations marked with $\hinv$ in
Fig. \ref{fig:tree-proof} (Prop. \ref{prop:single-reversal}). For
space reasons, the havoc invariance proofs of these annotations are
provided as supplementary material.

\section{Towards Automated Proof Generation}
\label{sec:open-problems}

Proof generation can be automated, by tackling the following technical
problems, briefly described in this section.

\paragraph{The entailment problem} Given a SID $\asid$ and two \adl\ formul{\ae}
$\phi$ and $\psi$, interpreted over $\asid$, is every model of $\phi$
also a model of $\psi$? This problem arises e.g., when applying the
rule of consequence (Fig. \ref{fig:hoare}c bottom-left) in a
Hoare-style proof of a reconfiguration program. Unsurprisingly, the
\adl\ entailment inherits the positive and negative aspects of the
\seplog\ entailment \cite{Reynolds}. For instance, one can reduce the
undecidable problem of universality of context-free languages
\cite{BarHillel61} to \adl\ entailment, with $\phi$ and $\psi$
restricted to predicate atoms. Decidability can be recovered via two
restrictions on the syntax of the rules in the SID and a semantic
restriction on the configurations that occur as models of the
predicate atoms defined by the SID. The syntactic restrictions are
that, each rule is of the form $\apred(x_1, \ldots, x_{\#(\apred)}) \unfoldrule
\exists y_1 \ldots \exists y_m ~.~ \compin{x}{q} * \phi * \Asterisk_{\ell=1}^h
\bpred^\ell(z^\ell_1, \ldots, z^\ell_{\#(\bpred^\ell)})$,
where $\phi$ consists of interaction atoms, such that: 
(1) $x_1$ occurs in each interaction atom from $\phi$, 
(2) $\bigcup_{ell=1}^h \set{z^\ell_1, \ldots,
  z^\ell_{\#(\bpred^\ell)}} = \set{x_2, \ldots, x_{\#(\apred)}} \cup
\set{y_1, \ldots, y_m}$, and
(3) for each $\ell \in \interv{1}{h}$, $z^\ell_1$ occurs in $\phi$.
Furthermore, the semantic restriction is that, in each model of a
predicate atom, a component must occur in a \emph{bounded number of
interactions}, i.e., the structure is a graph of bounded degree. For
instance, star topologies with a central controller and an unbounded
number of workers can be defined in \adl, but do not satisfy this
constraint. With these restrictions, it can be shown that the
\adl\ entailment problem is \twoexptime-complete, thus matching the
complexity of the similar problem for
\seplog\ \cite{EchenimIosifPeltier20,KatelaanZuleger20}. Technical are
given in \cite{BozgaBueriIosif2022Arxiv}.

\paragraph{The frame inference problem} Given two
\adl\ formul{\ae} $\phi$ and $\psi$ find a formula $\varphi$, such
that $\phi \models \psi * \varphi$. This problem occurs e.g., when
applying the frame rule (Fig. \ref{fig:hoare}c bottom-right) with a
premisse $\hoare{\phi}{\acomm}{\psi}$ to an arbitrary precondition
$\xi$ i.e., one must infer a frame $\varphi$ such that $\xi \models
\phi * \varphi$. This problem has been studied for
\seplog\ \cite{CalcagnoDistefanoOHearnYang11,DBLP:conf/sas/GorogiannisKO11},
in cases where the SID defines only data structures of a restricted
form (typically nested lists). Reconsidering the frame inference
problem for \adl\ is of paramount importance for automating the
generation of Hoare-style correctness proofs and is an open problem.

\paragraph{Automating havoc invariance proofs}
Given a precondition $\phi$ and a regular expression $\are$, the
parallel composition rule (\parr) requires the inference of regular
expressions $\are_1$ and $\are_2$, such that $\are_1
\parcomp_{\aenv_1,\aenv_2} \are_2 \requiv{\phi} \are$. We conjecture
that, under the bounded degree restriction above, the languages of the
frontier (cross-boundary) interactions (Def. \ref{def:frontier}) are
regular and can be automatically inferred by classical automata
construction techniques.

\section{Related Work}
\label{sec:related}

The ability of reconfiguring coordinating architectures of software
systems has received much interest in the Software Engineering
community, see the surveys
\cite{bradbury2004survey,rumpe2017classification}. We consider
\emph{programmed reconfiguration}, in which the architecture changes
occur according to a sequential program, executed in parallel with the
system to which reconfiguration applies. The languages used to write
such programs are classified according to the underlying formalism
used to define their operational semantics: \emph{process algebras},
e.g.\ $\pi$-{\sc ADL} \cite{CavalcanteBO15}, {\sc darwin}
\cite{magee1996dynamic}, \emph{hyper-graphs} and \emph{graph
rewriting}
\cite{taentzer1998dynamic,DBLP:journals/scp/WermelingerF02,LeMetayer,Cao2005,Arad13},
\emph{chemical reactions} \cite{wermelinger1998towards}, etc. We
separate architectures (structures) from behaviors, thus relating to
the BIP framework \cite{basu2006modeling} and its extensions for
dynamic reconfigurable systems DR-BIP \cite{DR-BIP-STTT}. In a similar
vein, the REO language \cite{Arbab04} supports reconfiguration by
changing the structure of connectors \cite{Clarke08}.

Checking the correctness of a dynamically reconfigurable system
considers mainly \emph{runtime verification} methods, i.e.\ checking a
given finite trace of observed configurations against a logical
specification. For instance, in \cite{BucchiaroneG08}, configurations
are described by annotated hyper-graphs and configuration invariants
of finite traces, given first-order logic, are checked using {\sc
  Alloy} \cite{Jackson02}. More recently,
\cite{DormoyKL10,LanoixDK11,DBLP:conf/sac/El-HokayemBS21} apply
temporal logic to runtime verification of reconfigurable systems.
Model checking of temporal specifications is also applied to REO
programs, under simplifying assumption that render the system
finite-state \cite{Clarke08}. In contrast, we use induction to deal
with parameterized systems of unbounded sizes.

To the best of our knowledge, our work is the first to tackle the
\emph{verification} of reconfiguration programs, by formally proving
the absence of bugs, using a Hoare-style annotation of a
reconfiguration program with assertions that describe infinite sets of
configurations, with unboundedly many components. Traditionally,
reasoning about the correctness of unbounded networks of parallel
processes uses mostly hard-coded architectures (see
\cite{BloemJacobsKhalimovKonnovRubinVeithWidder15} for a survey),
whereas the more recently developed architecture description logics
\cite{KonnovKWVBS16,MavridouBBS17} do not consider the
reconfigurability aspect of distributed systems.

Specifying parameterized component-based systems by inductive
definitions is not new. \emph{Network grammars}
\cite{ShtadlerGrumberg89} use context-free grammar rules to describe
systems with linear (pipeline, token-ring) structure, obtained by
composition of an unbounded number of processes. More complex
structures are specified recursively using graph grammar rules with
parameter variables \cite{LeMetayer}. To avoid clashes, these
variables must be renamed to unique names and assigned unique indices
at each unfolding step. Our recursive specifications use existential
quantifiers to avoid name clashes and separating conjunction to
guarantee that the components and interactions obtained by the
unfolding of the rules are unique.

The assertion language introduced in this paper is a resource logic
that supports local reasoning \cite{OHearnReynoldsYang01}. Local
reasoning about parallel programs has been traditionally within the
scope of Concurrent Separation Logic (CSL), that introduced a parallel
composition rule \cite{DBLP:journals/tcs/OHearn07}, with a
non-interfering (race-free) semantics of shared-memory parallelism
\cite{Brookes:2016}. Considering interference in CSL requires more
general proof rules, combining ideas of assume- and rely-guarantee
\cite{Owicki1978,DBLP:phd/ethos/Jones81} with local reasoning
\cite{FengFerreiraShao07,Vafeiadis07} and abstract notions of framing
\cite{Dinsdale-Young10,Dinsdale-Young13,Farka21}. These rules
generalize from both standard CSL parallel composition and
rely-guarantee rules, allowing even to reason about properties of
concurrent objects, such as (non-)linearizability
\cite{Sergey16}. However, the body of work on CSL deals almost
entirely with shared-memory multithreading programs, instead of
distributed systems, which is the aim of our work. In contrast, we
develop a resource logic in which the processes do not just share and
own resources, but become mutable resources themselves.

\section{Conclusions and Future Work}
\label{sec:conclusions}

We present a framework for deductive verification of reconfiguration
programs, based on a configuration logic that supports local
reasoning. We prove the absence of design bugs in ideal networks,
without packet loss and communication delays, using a discrete
event-based model of behavior, the usual level of abstraction in
formal verification of parameterized distributed systems. Our
configuration logic relies on inductive predicates to describe systems
with unbounded number of components. It is used to annotate
reconfiguration programs with Hoare triples, whose validity relies on
havoc invariants about the ongoing interactions in the system. These
invariants are tackled with a specific proof system, that uses a
parallel composition rule in the style of assume/rely-guarantee
reasoning.


As future work, we consider push-button techniques for frame inference
and havoc invariant synthesis, allowing broadcast interactions between
all the components, and extensions of the finite-state model of
behavior, using timed and hybrid automata.

\bibliographystyle{abbrv}
\bibliography{literature}

\newpage
\appendix

\section{Proofs from Section \ref{sec:rpl}}

\LemmaLocality* \proof{ By induction on the structure of the local
  program $\acomm$. For the base case $\acomm \in \primcomms$, we
  check the following points, for all $\aconfig_i = (\comps_i,
  \interacs_i, \statemap_i, \store) \in \configset$, for $i=1,2$, such
  that $\aconfig_1 \comp \aconfig_2$ is defined: \begin{itemize}
  \item $\acomm = \new(\smallstate,x)$: we compute $\semcomm{\new(q,x)}(\aconfig_1 \comp \aconfig_2) =$
    \[\begin{array}{l}
    \set{(\comps_1 \cup \comps_2 \cup \set{c}, \interacs_1 \cup \interacs_2, (\statemap_1 \cup \statemap_2)[c \leftarrow q], \store[x \leftarrow c])} = \\
    \set{(\comps_1 \cup \set{c}, \interacs_1, \statemap_1[c \leftarrow q], \store[x \leftarrow c]) \comp (\comps_2, \interacs_2, \statemap_2, \store[x \leftarrow c])} \subseteq \\
    \semcomm{\new(q,x)}(\aconfig_1) \comp \lift{\big\{\aconfig_2\big\}}{\set{x}} = \semcomm{\new(q,x)}(\aconfig_1) \comp \lift{\big\{\aconfig_2\big\}}{\modif{\new(\smallstate,x)}}
    \end{array}\]
  \item $\acomm = \delete(x)$: we distinguish the following cases: \begin{itemize}
  \item if $\store(x) \in \comps_1$, we compute $\semcomm{\delete(x)}(\aconfig_1 \comp \aconfig_2) =$
    \[\begin{array}{l}
    \set{((\comps_1 \cup \comps_2) \setminus \set{\store(x)}, \interacs_1 \cup \interacs_2, \statemap_1 \cup \statemap_2, \store)} = \\
    \set{(\comps_1 \setminus \set{\store(x)}, \interacs_1, \statemap_1, \store)} \comp \set{(\comps_2, \interacs_2, \statemap_2, \store)} \subseteq \\
    \semcomm{\delete(x)}(\aconfig_1) \comp \lift{\big\{\aconfig_2\big\}}{\emptyset} = \semcomm{\delete(x)}(\aconfig_1) \comp \lift{\big\{\aconfig_2\big\}}{\modif{\delete(x)}}
    \end{array}\]    
  \item else $\store(x) \not\in \comps_1$ and
    $\semcomm{\delete(x)}(\aconfig_1) = \errconfigs$, thus
    we obtain:
      \[\semcomm{\delete(x)}(\aconfig_1 \comp
      \aconfig_2) \subseteq \errconfigs = \errconfigs \comp
      \lift{\set{\aconfig_2}}{\emptyset} =
      \semcomm{\delete(x)}(\aconfig_1) \comp \lift{\set{\aconfig_2}}{\emptyset} =
      \semcomm{\delete(x)}(\aconfig_1) \comp \lift{\set{\aconfig_2}}{\modif{\delete(\acomp,x)}}\]
    \end{itemize}
  \item $\acomm = \connect(x_1.p1,x_2.p_2)$: we compute $\semcomm{\connect(x_1.p1,x_2.p_2)}(\aconfig_1 \comp \aconfig_2) =$
    \[\begin{array}{l}
    \set{(\comps_1 \cup \comps_2, \interacs_1 \cup \interacs_2 \cup \set{(\store(x_1), p_1, \store(x_2), p_2), \statemap_1 \cup \statemap_2, \store})} = \\
    \set{(\comps_1, \interacs_1 \cup \set{(\store(x_1), p_1, \store(x_2), p_2)}, \statemap_1, \store)} \comp \set{(\comps_2,\interacs_2, \statemap_2, \store)} \subseteq \\
    \semcomm{\connect(x_1.p1,x_2.p_2)}(\aconfig_1) \comp \lift{\set{\aconfig_2}}{\emptyset} =
    \semcomm{\connect(x_1.p1,x_2.p_2)}(\aconfig_1) \comp \lift{\set{\aconfig_2}}{\modif{\connect(x_1.p1,x_2.p_2)}}
    \end{array}\]    
  \item $\acomm = \disconnect(x_1.p_1, x_2.p_2)$: we distinguish the
    following cases: \begin{itemize}
    \item if $(\store(x_1), p_1, \store(x_2), p_2) \in \interacs_1$, we compute 
      $\semcomm{\disconnect(x_1.p_1, x_2.p_2)}(\aconfig_1 \comp \aconfig_2) =$
      \[\begin{array}{l}
      \set{(\comps_1 \cup \comps_2, (\interacs_1 \cup \interacs_2) \setminus \set{(\store(x_1), p_1, \store(x_2), p_2)}, \statemap_1 \cup \statemap_2, \store)} = \\
      \set{(\comps_1, \interacs_1 \setminus \set{(\store(x_1), p_1, \store(x_2), p_2)}, \statemap_1 \cup \statemap_2, \store)} \comps
      \lift{\set{(\comps_2, \interacs_2, \statemap_2, \store)}}{\emptyset} \subseteq \\
      \semcomm{\disconnect(x_1.p1,x_2.p_2)}(\aconfig_1) \comp \lift{\set{\aconfig_2}}{\emptyset} =
      \semcomm{\disconnect(x_1.p1,x_2.p_2)}(\aconfig_1) \comp \lift{\set{\aconfig_2}}{\modif{\disconnect(x_1.p1,x_2.p_2)}}
      \end{array}\]      
    \item else $(\store(x_1), p_1, \store(x_2), p_2) \not\in \interacs_1$ and
      $\semcomm{\disconnect(x_1.p_1, x_2.p_2)}(\aconfig_1) = \errconfigs$, thus:
      \[\begin{array}{l}
      \semcomm{\disconnect(x_1.p_1, x_2.p_2)}(\aconfig_1 \comp \aconfig_2) \subseteq \errconfigs = \errconfigs \comp \set{\aconfig_2} = \\ 
      \semcomm{\disconnect(x_1.p_1, x_2.p_2)}(\aconfig_1) \comp \lift{\set{\aconfig_2}}{\emptyset} =
      \semcomm{\disconnect(x_1.p_1, x_2.p_2)}(\aconfig_1) \comp \lift{\set{\aconfig_2}}{\modif{\disconnect(x_1.p_1, x_2.p_2)}}
      \end{array}\] 
    \end{itemize}
  \item $\acomm = \skipcomm$: this case is a trivial check.
  \end{itemize}
  For the inductive step, we check the following points: \begin{itemize}
  \item $\acomm = \acomm_1 + \acomm_2$: we compute $\semcomm{\acomm_1 + \acomm_2}(\aconfig_1 \comp \aconfig_2) =$
    \[\begin{array}{l}     
    \semcomm{\acomm_1}(\aconfig_1 \comp \aconfig_2) \cup \semcomm{\acomm_2}(\aconfig_1 \comp \aconfig_2) \subseteq 
    \text{ [by the inductive hypothesis]} \\
    \semcomm{\acomm_1}(\aconfig_1) \comp \lift{\set{\aconfig_2}}{\modif{\acomm_1}} \cup
    \semcomm{\acomm_2}(\aconfig_1) \comp \lift{\set{\aconfig_2}}{\modif{\acomm_2}} \subseteq 
    \text{ [$\modif{\acomm_1 + \acomm_2} = \modif{\acomm_1} \cup \modif{\acomm_2}$]} \\
    \semcomm{\acomm_1}(\aconfig_1) \comp \lift{\set{\aconfig_2}}{\modif{\acomm_1+\acomm_2}} \cup
    \semcomm{\acomm_2}(\aconfig_1) \comp \lift{\set{\aconfig_2}}{\modif{\acomm_1+\acomm_2}} = \\
    \big(\semcomm{\acomm_1}(\aconfig_1) \cup \semcomm{\acomm_2}(\aconfig_1)\big) \comp \lift{\set{\aconfig_2}}{\modif{\acomm_1+\acomm_2}} =
    \semcomm{\acomm_1 + \acomm_2}(\aconfig_1) \comp \lift{\set{\aconfig_2}}{\modif{\acomm_1+\acomm_2}}
    \end{array}\]
  \item $\acomm = (\withcomm{\vec{x}}{\pureform}{\acomm_1})$, where
    $\pureform$ consists of equalities and disequalities; we distinguish the
    cases below: \begin{itemize}
  \item if $\aconfig_1 \comp \aconfig_2 \models \pureform$, we compute 
    \[\begin{array}{rcl}
    \semcomm{\withcomm{\vec{x}}{\pureform}{\acomm_1}}(\aconfig_1 \comp \aconfig_2) & \subseteq & \semcomm{\acomm_1}(\lift{\set{\aconfig_1 \comp \aconfig_2}}{\vec{x}}) \\
    & \subseteq & \semcomm{\acomm_1}(\aconfig_1) \comp \lift{\set{\aconfig_2}}{\vec{x} \cup \modif{\acomm_1}} \\ 
    & = & \semcomm{\acomm_1}(\aconfig_1) \comp \lift{\set{\aconfig_2}}{\modif{\withcomm{\vec{x}}{\pureform}{\acomm_1}}}
    \end{array}\]
  \item else $\aconfig_1 \comp \aconfig_2 \not\models \pureform$ and 
    \[\semcomm{\withcomm{\vec{x}}{\pureform}{\acomm_1}}(\aconfig_1 \comp \aconfig_2) = 
    \emptyset \subseteq \semcomm{\acomm_1}(\aconfig_1) \comp \lift{\set{\aconfig_2}}{\modif{\whencomm{\pureform}{\acomm_1}}}\]
  \end{itemize}
  \end{itemize}
\qed}

\LemmaHoareAxioms* \proof{ The
  proof goes by case split on the type of the primitive command
  $\acomm$, which determines the pre- and post-condition $\phi$ and
  $\psi$ of the axiom, respectively: \begin{itemize}
  \item $\acomm = \new(q,x)$, $\phi = \emp$ and $\psi
    = \compin{x}{q}$:
    \[\semcomm{\new(q,x)}(\sem{\emp}{}) 
    = \set{ (\set{c}, \emptyset, \statemap[c \leftarrow q], \store[x
        \leftarrow c]) \mid c \in \universe} = \sem{\compin{x}{q}}{}\]
    The second step applies the definition
    $\semcomm{\acomm}{(\sem{\phi}{})} = \set{\aconfig' \mid \exists
      \aconfig \in \sem{\phi}{} ~.~
      \succj{\acomm}{\aconfig}{\aconfig'}}$ to the case $\acomm =
    \new(q,x)$, where the judgement
    $\succj{\new(q,x)}{\aconfig}{\aconfig'}$ is defined in
    Fig. \ref{fig:os}. The rest is by the semantics of \adl.
  \item $\acomm = \delete(x)$, $\phi = \company{x}$ and $\psi = \emp$: 
    \[
    \semcomm{\delete(x)}(\sem{\company{x}}{}) = 
    \semcomm{\delete(x)}(\set{(\set{\store(x)}, \emptyset, \statemap, \store) \mid \dom{\statemap}=\set{\store(x)}}) = \sem{\emp}{}
    \]
    The second step applies the definition
    $\semcomm{\acomm}{(\sem{\phi}{})} = \set{\aconfig' \mid \exists
      \aconfig \in \sem{\phi}{} ~.~
      \succj{\acomm}{\aconfig}{\aconfig'}}$ to the case $\acomm =
    \delete(x)$, where the judgement
    $\succj{\delete(x)}{\aconfig}{\aconfig'}$ is defined in
    Fig. \ref{fig:os}. The rest is by the semantics of \adl.
  \item $\acomm = \connect(x_1.p_1, x_2.p_2)$,
    $\phi = \emp$ and $\psi=\interac{x_1}{p_1}{x_2}{p_2}$:
    similar to $\acomm = \new(q,x)$.
  \item $\acomm = \disconnect(x_1.p_1, x_2.p_2)$, $\phi =
    \interac{x_1}{p_1}{x_2}{p_2}$ and $\psi = \emp$: similar to
    $\acomm = \delete(x)$.
  \item $\acomm = \skipcomm$ and $\phi = \emp$: trivial. \qed
  \end{itemize}}

\ThmSoundness* \proof{We prove that the inference rules in
  Fig. \ref{fig:hoare} are sound. For the axioms, soundness follows
  from Lemma \ref{lemma:hoare-axioms}. The rules for the composite
  programs are proved below by a case split on the syntax of the
  program from the conclusion $\hoare{\phi}{\acomm}{\psi}$, assuming
  that $\models \hoare{\phi_i}{\acomm_i}{\psi_i}$, for each premiss
  $\hoare{\phi_i}{\acomm_i}{\psi_i}$ of the rule: \begin{itemize}
    %
    %
  \item $\acomm = \withcomm{x_1, \ldots, x_k}{\varphi}{\acomm}$: let $\config \in
    \sem{\phi}{}$ be a configuration and distinguish the following
    cases: \begin{itemize}
    \item if $(\comps, \interacs, \statemap, \store[x_1 \leftarrow
      c_1, \ldots, x_k \leftarrow c_k]) \models \varphi * \predtrue$,
      for some $c_1, \ldots, c_k \in \universe$, we obtain $(\comps,
      \interacs, \statemap, \store[x_1 \leftarrow c_1, \ldots, x_k
        \leftarrow c_k]) \models \phi \wedge (\varphi * \predtrue)$,
      because $\fv{\phi} \cap \set{x_1,\ldots,x_k}=\emptyset$. Then
      $\semcomm{\withcomm{x_1, \ldots, x_k}{\varphi}{\acomm}}\config =
      \semcomm{\acomm}(\comps, \interacs, \statemap, \store[x_1
        \leftarrow c_1, \ldots, x_k \leftarrow c_k]) \subseteq
      \sem{\psi}{} \subseteq \sem{\exists \vec{x} ~.~ \psi}{}$ follows
      from the premiss of the rule.
    \item othewise, we have $\config \models \forall x_1 \ldots
      \forall x_k ~.~ \neg(\varphi * \predtrue)$ and
      $\semcomm{\withcomm{\vec{x}}{\varphi}{\acomm}}\config =
      \emptyset \subseteq \sem{\exists \vec{x} ~.~ \psi}{}$ follows.
    \end{itemize}
  \item the cases $\acomm = \acomm_1; \acomm_2$, $\acomm = \acomm_1 +
    \acomm_2$ and $\acomm = \acomm_1^*$ are simple checks using the
    operational semantics rules from Fig. \ref{fig:os}.
  \end{itemize}
  Concerning the structural rules, we show only the soundness of the
  frame rule below; the other rules are simple checks, left to the
  reader. Let $\aconfig \in \sem{\phi * \varphi}{}$ be a
  configuration. By the semantics of $*$, there exists $\aconfig_1 \in
  \sem{\phi}{}$ and $\aconfig_2 \in \sem{\varphi}{}$, such
  that $\aconfig = \aconfig_1 \comp \aconfig_2$. Since $\acomm \in
  \localcomms$, by Lemma \ref{lemma:locality}, we obtain
  $\semcomm{\acomm}(\aconfig_1 \comp \aconfig_2) \subseteq
  \semcomm{\acomm}(\aconfig_1) \comp
  \lift{\set{\aconfig_2}}{\modif{\acomm}}$. Since $\aconfig_1 \in
  \sem{\phi}{}$, by the hypothesis on the premiss we obtain
  $\semcomm{\acomm}(\aconfig_1) \subseteq
  \sem{\psi}{}$. Moreover, since $\aconfig_2 \in
  \sem{\varphi}{}$ and $\modif{\acomm} \cap \fv{\varphi}$, we
  obtain $\lift{\set{\aconfig_2}}{\modif{\acomm}} \subseteq
  \sem{\varphi}{}$, leading to $\semcomm{\acomm}(\aconfig)
  \subseteq \sem{\psi * \varphi}{}$, as required. \qed}

\PropSingleReversal* \proof{ In order to apply the sequential
  composition rule for the entire sequence, we need to prove that
  $\phi_1, \ldots, \phi_{\ell-1}$ are havoc invariant: \begin{itemize}
  \item For $i \in \interv{1}{k}$, the proof is by induction on
    $i$. In the base case, $\phi_1$ is havoc invariant, by the
    hypothesis. For the inductive step $i \in \interv{2}{k}$, let
    $\config \models \phi_{i}$ and $(\comps, \interacs, \statemap,
    \store) \Arrow{}{} \ldots \Arrow{}{} (\comps, \interacs,
    \statemap', \store)$ be a sequence of state changes induced by the
    execution of some interactions $(c_{i_1}, p_{i_1}, c'_{i_1},
    p'_{i_1})$, $\ldots$, $(c_{i_n}, p_{i_n}, c'_{i_n}, p'_{i_n}) \in
    \interacs$. Since $\models\hoare{\phi_{i-1}}{\disconnect(x_i.p_i,
      x'_i.p'_i)}{\phi_i}$, by the hypothesis, there exists a model
    $(\comps,\interacs',\statemap,\store)$ of $\phi_{i-1}$, such that
    $(c_{i_1}, p_{i_1}, c'_{i_1}, p'_{i_1})$, $\ldots$, $(c_{i_n},
    p_{i_n}, c'_{i_n}, p'_{i_n}) \in \interacs'$. Since $\phi_{i-1}$
    is havoc invariant, by the inductive hypothesis, we have
    $(\comps,\interacs',\statemap',\store) \models \phi_{i-1}$ and,
    since $\interacs = \interacs' \setminus \set{(\store(x_i), p_i,
      \store(x'_i), p'_i)}$, we have $(\store(x_i), p_i, \store(x'_i),
    p'_i) \not\in \{(c_{i_1}, p_{i_1}, c'_{i_1}, p'_{i_1}), \ldots,
    (c_{i_n}, p_{i_n}, c'_{i_n}, p'_{i_n})\}$, thus
    $(\comps,\interacs,\statemap',\store) \models \phi_i$. Since the
    choices of $\config$ and $(c_{i_1}, p_{i_1}, c'_{i_1}, p'_{i_1}),
    \ldots, (c_{i_n}, p_{i_n}, c'_{i_n}, p'_{i_n})$ were arbitrary, we
    obtain that $\phi_i$ is havoc invariant.
  \item For $i \in \interv{k+1}{\ell-1}$, the proof is by reversed
    induction on $i$. In the base case, $\phi_{\ell-1}$ is havoc
    invariant, by the hypothesis. For the inductive step, let $\config
    \models \phi_{i-1}$ and $(\comps, \interacs, \statemap, \store)
    \Arrow{}{} \ldots \Arrow{}{} (\comps, \interacs, \statemap',
    \store)$ be a sequence of state changes, induced by the executions
    of some interactions $(c_{i_1}, p_{i_1}, c'_{i_1}, p'_{i_1})$,
    $\ldots$, $(c_{i_n}, p_{i_n}, c'_{i_n}, p'_{i_n}) \in \interacs$.
    Since $\models\hoare{\phi_{i-1}}{\connect(x_i.p_i,
      x'_i.p_i)}{\phi_i}$, by the hypothesis, there exists a model
    $(\comps,\interacs',\statemap,\store)$ of $\phi_{i}$, such that
    $(c_{i_1}, p_{i_1}, c'_{i_1}, p'_{i_1})$, $\ldots$, $(c_{i_n},
    p_{i_n}, c'_{i_n}, p'_{i_n}) \in \interacs'$. Since $\phi_i$ is
    havoc invariant, by the inductive hypothesis, we have
    $(\comps,\interacs',\statemap',\store) \models \phi_i$ and, since
    $\interacs =
    \interacs'\setminus\set{(\store(x_i),p_i,\store(x'_i),p'_i)}$, we
    have $(\store(x_i),p_i,\store(x'_i),p'_i) \not\in \{(c_{i_1},
    p_{i_1}, c'_{i_1}, p'_{i_1})$, $\ldots$, $(c_{i_n}, p_{i_n},
    c'_{i_n}, p'_{i_n})\}$ and $(\comps,\interacs,\statemap',\store)
    \models \phi_{i-1}$. Since the choices of $\config$ and $(c_{i_1},
    p_{i_1}, c'_{i_1}, p'_{i_1}), \ldots, (c_{i_n}, p_{i_n}, c'_{i_n},
    p'_{i_n})$ were arbitrary, $\phi_{i-1}$ is havoc invariant. \qed
  \end{itemize}}

\section{Proof from Section \ref{sec:havoc}}

\PropChainPrecise* \proof{ Let $\varphi_i \isdef \phi_i *
  \Asterisk_{j=1}^{k_i} \chain{h_{i,j}}{t_{i,j}}(x_{i,j},y_{i,j})$ be
  symbolic configurations, where $\phi_i$ is a predicate-free symbolic
  configuration and $h_{i,j}, t_{i,j} \geq 0$ are integers, for all $j
  \in \interv{1}{k_i}$ and $i = 1,2$. We prove that $\varphi_1$ is
  precise on $\sem{\varphi_2}{}$. Let $\aconfig = \config \in
  \sem{\varphi_2}{}$ be a configuration and suppose that there exist
  configurations $\aconfig' = (\comps',\interacs',\statemap',\store)$
  and $\aconfig'' = (\comps'',\interacs'',\statemap'',\store)$, such
  that $\aconfig' \isubstreq \aconfig$, $\aconfig'' \isubstreq
  \aconfig$, $\aconfig' \models \varphi_1$ and $\aconfig'' \models
  \varphi_1$. Then there exist configurations $\aconfig'_0 \isdef
  (\comps'_0,\interacs'_0,\statemap,\store), \ldots, \aconfig'_{k_1}
  \isdef (\comps'_{k_1},\interacs'_{k_1},\statemap,\store)$ and
  $\aconfig''_0 \isdef (\comps''_0,\interacs''_0,\statemap,\store),
  \ldots, \aconfig''_{k_1} \isdef (\comps''_{k_1}, \interacs''_{k_1},
  \statemap, \store)$, such that: \begin{itemize}
  \item $\aconfig' = \bigcomp_{j=0}^{k_1} ~\aconfig'_j$ and
    $\aconfig'' = \bigcomp_{j=0}^{k_1} ~\aconfig''_j$,
  \item $\aconfig'_0 \models \phi_1$ and $\aconfig''_0 \models
    \phi_1$, and
  \item $\aconfig'_j \models
    \chain{h_{1,j}}{t_{1,j}}(x_{1,j},y_{1,j})$ and $\aconfig''_j
    \models \chain{h_{1,j}}{t_{1,j}}(x_{1,j},y_{1,j})$, for all
    $j \in \interv{1}{k_1}$.
  \end{itemize}
  Since $\phi_1$ is a predicate-free symbolic configuration, we have
  $\comps'_0=\comps''_0$ and $\interacs'_0=\interacs''_0$, thus
  $\aconfig'_0 = \aconfig''_0$. Moreover, for each $j \in
  \interv{1}{k_1}$, we have $\comps'_j = \comps''_j$ and $\interacs'_j
  = \interacs''_j$, because both configurations consist of the tight
  interactions $(c_1, \mathit{out}, c_2, \mathit{in})$, $\ldots$,
  $(c_{\ell-1}, \mathit{out}, c_\ell, \mathit{in})$, such that
  $\store(x_{1,j}) = c_1$ and $\store(y_{1,j}) = c_\ell$. Thus, we
  obtain $\aconfig'_j = \aconfig''_j$, for all $j \in \interv{0}{k}$,
  leading to $\aconfig' = \aconfig''$. \qed}

\PropHavocValid* \proof{ Let $\aconfig=\config$ be a model of $\phi$
  i.e., $\aconfig \in \sem{\phi}{}$. It is sufficient to prove that
  $\ahavoc(\aconfig) \subseteq
  \set{(\comps,\interacs,\statemap',\store) \mid \aconfig \Open{w}{}
    (\comps,\interacs,\statemap',\store),~ w \in
    \semlang{\interof{\phi}^*}{\aconfig}}$, because
  $\set{(\comps,\interacs,\statemap',\store) \mid \statemap \Open{w}{}
    (\comps,\interacs,\statemap',\store)} \subseteq \sem{\psi}{}$ for
  each $w \in \semlang{\interof{\phi}^*}{\aconfig}$, by the hypothesis
  $\models \havoctriple{\aenv}{\phi}{\interof{\phi}^*}{\psi}$
  (Def. \ref{def:havoc-valid}). Let $\aconfig' =
  (\comps,\interacs,\statemap',\store)\in \ahavoc(\aconfig)$ be a
  configuration. Then there exists a finite sequence of interactions,
  say $w = (c_1,p_1,c'_1,p'_1) \ldots (c_n,p_n,c'_n,p'_n) \in
  \interacs^*$, such that $(\comps,\interacs,\statemap,\store)
  \Arrow{w}{} (\comps,\interacs,\statemap',\store)$, by
  Def. \ref{def:havoc}.  Note that
  $\set{(\comps,\interacs,\statemap'',\store) \mid
    (\comps,\interacs,\statemap,\store) \Arrow{(c_1,p_1,c_2,p_2)}{}
    (\comps,\interacs,\statemap'',\store)} \subseteq
  \set{(\comps,\interacs,\statemap'',\store) \mid \statemap
    \Open{(c_1,p_1,c_2,p_2)}{}
    (\comps,\interacs,\statemap'',\store)}$, for each interaction
  $(c_1,p_1,c_2,p_2) \in \interacs$, by Def. \ref{def:havoc} and
  Def. \ref{def:open}. It remains to show that $w \in
  \semlang{\interof{\phi}^*}{\aconfig}$. Since $\config \models \phi$,
  for any interaction $(c_k,p_k,c'_k,p'_k)\in\interacs$, for $k \in
  \interv{1}{n}$, we distinguish two cases, either: \begin{itemize}
  \item there exists an interaction atom $\alpha =
    \interac{x_k}{p_k}{x'_k}{p'_k} \in \iatoms{\phi}$ and a
    configuration $(\comps'',\interacs'',\statemap'',\store)
    \isubstreq \aconfig$, such that
    $\interacs''=\set{(c_k,p_k,c'_k,p'_k)}$ and
    $(\comps'',\interacs'',\statemap'',\store) \models \alpha$, or
  \item there exists a predicate atom $\alpha \in \patoms{\phi}$ and a
    configuration $(\comps'',\interacs'',\statemap'',\store)
    \isubstreq \aconfig$, such that $(c_k,p_k,c'_k,p'_k)\in\interacs''$ and
    $(\comps'',\interacs'',\statemap'',\store) \models \alpha$. 
  \end{itemize}
  In both cases, we have $(c_k,p_k,c'_k,p'_k) \in
  \semlang{\interof{\alpha}}{\aconfig}$, for some $\alpha \in
  \atoms{\phi}$, thus $(c_k,p_k,c'_k,p'_k) \in
  \semlang{\interof{\phi}}{\aconfig}$, because $\interof{\phi} =
  \bigcup_{\alpha\in\atoms{\phi}} \interof{\alpha}$. Since the choice
  of $k \in \interv{1}{n}$ is arbitrary, we obtain that $w \in
  \semlang{\interof{\phi}^*}{\aconfig}$. \qed}

\LemmaDistinctive* \proof{ The proof goes by induction on the
  structure of the proof tree. For the base case, the tree consists of
  a single root node and let $\havoctriple{\aenv}{\phi}{\are}{\psi}$
  be the label of the root node. By Assumption \ref{ass:root}, $\phi$
  is a symbolic configuration and $\aenv = \set{\interof{\alpha_1},
    \ldots, \interof{\alpha_k}}$, where $\atoms{\phi} = \set{\alpha_1,
    \ldots, \alpha_k}$ is the set of interaction and predicate atoms
  from $\phi$. Let $\aconfig$ be a model of $\phi$, hence there exist
  configurations $\aconfig_0, \aconfig_1, \ldots, \aconfig_k$, such
  that $\aconfig = \bigcomp_{i=0}^k ~\aconfig_i$ and $\aconfig_i
  \models \alpha_i$, for all $i \in \interv{1}{k}$. Because the
  composition $\aconfig_i \comp \aconfig_j$ is defined, we obtain that
  $\semlang{\interof{\alpha_i}}{\aconfig_i} \cap
  \semlang{\interof{\alpha_j}}{\aconfig_j} = \emptyset$, for all $i
  \neq j \in \interv{1}{k}$. Moreover, since each formula
  $\alpha_i\in\atoms{\phi}$ is precise on $\sem{\phi}{}$, we have
  $\semlang{\interof{\alpha_i}}{\aconfig_i} =
  \semlang{\interof{\alpha_i}}{\aconfig}$, hence
  $\semlang{\interof{\alpha_i}}{\aconfig} \cap
  \semlang{\interof{\alpha_j}}{\aconfig} = \emptyset$, for all $i \neq
  j \in \interv{1}{k}$. For the inductive step, we distinguish the
  cases below, based on the type of the inference rule that expands
  the root: \begin{itemize}
  \item (\ii) Let $\havoctriple{\aenv}{\phi *
    \interac{x_1}{p_1}{x_2}{p_2}}{\are}{\psi *
    \interac{x_1}{p_1}{x_2}{p_2}}$ be label of the root of the proof
    tree and $\aconfig \in \sem{\phi *
      \interac{x_1}{p_1}{x_2}{p_2}}{}$ be a configuration. Then there
    exists configurations $\aconfig_0$ and $\aconfig_1$, such that
    $\aconfig = \aconfig_0 \comp \aconfig_1$, $\aconfig_0 \models
    \phi$ and $\aconfig_1 \models \interac{x_1}{p_1}{x_2}{p_2}$. By
    Assumption \ref{ass:root}, we have $\aenv = \interof{\phi} \cup
    \set{\interof{\interac{x_1}{p_1}{x_2}{p_2}}}$. By the inductive
    hypothesis, the premiss $\havoctriple{\aenv \setminus
      \set{\interac{x_1}{p_1}{x_2}{p_2}}}{\phi}{\are}{\psi}$ of the
    rule is distinctive, hence the interpretations of the atoms in the
    environment $\set{\semlang{\interof{\alpha}}{\aconfig_0} \mid
      \alpha \in \atoms{\phi}}$ are pairwise disjoint. Since each
    predicate atom $\alpha \in \atoms{\phi}$ is precise on
    $\sem{\phi}{}$, the sets $\semlang{\interof{\alpha}}{\aconfig}$,
    $\alpha \in \atoms{\phi}$ are also pairwise disjoint. Since
    $\interac{x_1}{p_1}{x_2}{p_2}$ is precise on $\configset$, we
    obtain that $\semlang{\interac{x_1}{p_1}{x_2}{p_2}}{\aconfig_1} =
    \semlang{\interac{x_1}{p_1}{x_2}{p_2}}{\aconfig}$ and, since
    $\aconfig = \aconfig_0 \comp \aconfig_1$, the set
    $\semlang{\interac{x_1}{p_1}{x_2}{p_2}}{\aconfig}$ is disjoint
    from the sets $\semlang{\interof{\alpha}}{\aconfig}$, $\alpha \in
    \atoms{\phi}$, thus $\havoctriple{\aenv}{\phi *
      \interac{x_1}{p_1}{x_2}{p_2}}{\are}{\psi *
      \interac{x_1}{p_1}{x_2}{p_2}}$ is distinctive.
  \item (\ie) Let $\havoctriple{\aenv}{\phi}{\are}{\psi}$ be the root
    label and let $\aconfig \in \sem{\phi}{}$ be a configuration. By
    Assumption \ref{ass:root}, we have $\aenv = \interof{\phi}$ and
    let $\interac{x_1}{p_1}{x_2}{p_2}$ be an interaction atom, such
    that $\excluded{\phi}{\interac{x_1}{p_1}{x_2}{p_2}}$. Let
    $\aconfig'$ be any model of $\interac{x_1}{p_1}{x_2}{p_2}$. By
    $\excluded{\phi}{\interac{x_1}{p_1}{x_2}{p_2}}$, we have
    $\interac{x_1}{p_1}{x_2}{p_2} \not\in \aenv$ and, moreover, the
    composition $\aconfig \comp \aconfig'$ is defined, thus $\aconfig
    \comp \aconfig' \in \sem{\phi *
      \interac{x_1}{p_1}{x_2}{p_2}}{}$. By the inductive hypothesis,
    $\havoctriple{\aenv \cup
      \set{\interof{\interac{x_1}{p_1}{x_2}{p_2}}}}{\phi *
      \interac{x_1}{p_1}{x_2}{p_2}}{\are}{\psi *
      \interac{x_1}{p_1}{x_2}{p_2}}$ is distinctive, hence
    $\semlang{\interof{\alpha_1}}{\aconfig \comp \aconfig'} \cap
    \semlang{\interof{\alpha_2}}{\aconfig \comp \aconfig'} =
    \emptyset$, for all $\interof{\alpha_1}, \interof{\alpha_2} \in
    \aenv$. Since $\aconfig' \models \interac{x_1}{p_1}{x_2}{p_2}$,
    $\aenv = \interof{\phi}$ and
    $\excluded{\phi}{\interac{x_1}{p_1}{x_2}{p_2}}$, we obtain
    $\semlang{\interof{\alpha}}{\aconfig \comp \aconfig'} =
    \semlang{\interof{\alpha}}{\aconfig}$, for all $\interof{\alpha}
    \in \aenv$, thus $\havoctriple{\aenv}{\phi}{\are}{\psi}$ is
    distinctive.
  \item (\parr) Let $\havoctriple{\aenv_1\cup\aenv_2}{\phi_1 *
    \phi_2}{\are_1 \parcomp_{\aenv_1,\aenv_2} \are_2}{\psi_1 *
    \psi_2}$ be the label of the root, $\aenv_i = \interof{\phi_i *
    \front{\phi_i}{\phi_{3-i}}}$, for $i=1,2$, and let $\aconfig$ be a
    model of the precondition of this havoc triple. Then there exists
    two configurations $\aconfig_1$, $\aconfig_2$, such that $\aconfig
    = \aconfig_1 \comp \aconfig_2$ and $\aconfig_i \models \phi_i$,
    for $i = 1,2$. By Assumption \ref{ass:root}, we have $\aenv_1 \cup
    \aenv_2 = \interof{\phi_1} \cup \interof{\phi_2}$. Let
    $\aconfig'_i$ be a structure, such that $\aconfig'_i \isubstreq
    \aconfig_{3-i}$ and $\aconfig'_i \models
    \front{\phi_i}{\phi_{3-i}}$, for $i = 1,2$. By the definition of
    $\front{\phi_i}{\phi_{3-i}}$, as separated conjunction of
    interaction atoms from $\phi_{3-i}$, these substructures exist,
    and moreover, because each interaction atom is precise on
    $\configset$, they are unique. Then we have $\aconfig_i \comp
    \aconfig'_i \models \phi_i * \front{\phi_i}{\phi_{3-i}}$, for $i =
    1,2$. By the inductive hypothesis, since each havoc triple
    $\havoctriple{\aenv_i}{\phi_i *
      \front{\phi_i}{\phi_{3-i}}}{\are_i}{\psi_i *
      \front{\phi_i}{\phi_{3-i}}}$ is distinctive, the sets
    $\set{\semlang{\interof{\alpha}}{\aconfig_i \comp \aconfig'_i}
      \mid \alpha \in \atoms{\phi_i}}$ are pairwise disjoint, for $i =
    1,2$. Since each predicate atom $\alpha \in \atoms{\phi_1 *
      \phi_2}$ is precise on $\sem{\phi_1 * \phi_2}{}$, hence the sets
    $\set{\semlang{\interof{\alpha}}{\aconfig} \mid \alpha \in
      \atoms{\phi_i}}$ are pairwise disjoint as well, for $i =
    1,2$. Since the configurations $\aconfig_1$ and $\aconfig_2$ share
    no interactions, the havoc triple
    $\havoctriple{\aenv_1\cup\aenv_2}{\phi_1 * \phi_2}{\are_1
      \parcomp_{\aenv_1, \aenv_2} \are_2}{\psi_1 * \psi_2}$ is
    distinctive.
  \item (\lu) Let $\havoctriple{\aenv}{\phi *
    \apred(y_1,\ldots,y_{\#(\apred)})}{\are}{\psi}$ be the label of
    the root let $\aconfig$ be a model of the precondition of this
    havoc triple. Then there exist configurations $\aconfig_0 =
    (\comps_0, \interacs_0, \statemap_0, \store)$ and $\aconfig_1 =
    (\comps_1, \interacs_1, \statemap_1, \store)$, such that $\aconfig
    = \aconfig_0 \comp \aconfig_1$, $\aconfig_0 \models \phi$ and
    $\aconfig_1 \models \apred(y_1,\ldots,y_{\#(\apred)})$. By
    Assumption \ref{ass:root}, we have $\aenv = \interof{\phi} \cup
    \set{\interof{\apred(y_1,\ldots,y_{\#(\apred)})}}$. Since
    $\aconfig_1 \models \apred(y_1,\ldots,y_{\#(\apred)})$, there
    exists a rule $\apred(x_1, \ldots, x_{\#(\apred)}) \unfoldrule
    \exists z_1 \ldots \exists z_{h} ~.~ \varphi$ in the SID, where
    $\varphi$ is a symbolic configuration, such that $(\comps_1, \interacs_1, \statemap_1, 
    \store[z_1 \leftarrow c_1, \ldots, z_{h} \leftarrow c_{h}]) \models \varphi$, for some components $c_1, \ldots, c_{h}
    \in \universe$, and let $\aenv' \isdef (\aenv \setminus
    \set{\interof{\apred(y_1,\ldots,y_{\#(\apred)})}}) \cup
    \interof{\varphi}$. We can assume w.l.o.g. that $\set{z_1, \ldots,
      z_h} \cap \fv{\phi} = \emptyset$ (if necessary, by an
    $\alpha$-renaming of existentially quantified variables), hence
    $\aconfig'_0 \models \phi$ and $\aconfig'_1 \models \apred(y_1,
    \ldots, y_{\#(\apred)})$, where $\aconfig'_0 \isdef (\comps_0,\interacs_0,\statemap_0,
    \store[z_1 \leftarrow c_1, \ldots, z_{h} \leftarrow c_{h}])$ and $\aconfig'_1 \isdef (\comps_1, \interacs_1, \statemap_1, \store[z_1
      \leftarrow c_1, \ldots, z_{h} \leftarrow c_{h}])$,
    thus $\aconfig' \models \phi * \apred(y_1, \ldots,
    y_{\#(\apred)})$, where $\aconfig' = \aconfig'_0 \comp
    \aconfig'_1$. By the inductive hypothesis, the premiss
    $\havoctriple{\aenv'}{\exists z_1 \ldots \exists z_{h} ~.~ \phi *
      \varphi[x_1/y_1 \ldots
        x_{\#(\apred)}/y_{\#(\apred)}]}{\are'}{\psi}$ is distinctive
    and, moreover, $\aconfig' \models \exists z_1 \ldots \exists z_{h}
    ~.~ \phi * \varphi[x_1/y_1 \ldots x_{\#(\apred)}/y_{\#(\apred)}]$,
    hence $\semlang{\interof{\alpha_1}}{\aconfig'} \cap
    \semlang{\interof{\alpha_2}}{\aconfig'} = \emptyset$, for all
    $\alpha_1 \in \atoms{\phi}$ and $\alpha_2 \in \atoms{\varphi}$.
    Since $\apred(y_1, \ldots, y_{\#(\apred)})$ is precise on
    $\sem{\phi * \apred(y_1, \ldots, y_{\#(\apred)})}{}$, we have
    $\semlang{\interof{\apred(y_1, \ldots,
        y_{\#(\apred)})}}{\aconfig'} = \bigcup_{\alpha \in
      \atoms{\varphi}} \semlang{\interof{\alpha}}{\aconfig'}$. Since
    $\semlang{\interof{\alpha}}{\aconfig'} =
    \semlang{\interof{\alpha}}{\aconfig}$, for each $\alpha \in
    \atoms{\phi} \cup \set{\apred(y_1, \ldots, y_{\#(\apred)})}$, we
    obtain that $\semlang{\interof{\alpha_1}}{\aconfig} \cap
    \semlang{\interof{\alpha_2}}{\aconfig} = \emptyset$, for all
    $\alpha_1, \alpha_2 \in \atoms{\phi} \cup \set{\apred(y_1, \ldots,
      y_{\#(\apred)})}$, leading to the fact that
    $\havoctriple{\aenv}{\phi *
      \apred(y_1,\ldots,y_{\#(\apred)})}{\are}{\psi}$ is distinctive.
  \item ($\vee$) Let $\havoctriple{\aenv}{\bigvee_{i=1}^k \phi \wedge
    \delta_i}{\are}{\bigvee_{i=1}^k\psi_i}$ be the label of the root
    and let $\aconfig$ be a model of the precondition of this
    triple. Then $\aconfig \models \phi \wedge \delta_i$, for some
    $i \in \interv{1}{k}$. By the inductive hypothesis, the triple
    $\havoctriple{\aenv}{\phi \wedge \delta_i}{\are}{\psi_i}$ is
    distinctive, hence $\semlang{\interof{\alpha_1}}{\aconfig} \cap
    \semlang{\interof{\alpha_1}}{\aconfig} = \emptyset$, for all
    $\alpha_1, \alpha_2 \in \atoms{\phi}$. Since $\atoms{\phi} =
    \atoms{\bigvee_{i=1}^k \phi \wedge \delta_i}$, we obtain that
    $\havoctriple{\aenv}{\bigvee_{i=1}^k \phi \wedge
      \delta_i}{\are}{\bigvee_{i=1}^k\psi_i}$ is distinctive.
  \item ($\wedge$) and ($\cdot$): these cases are similar to ($\vee$).
  \item (\conseq), ($*$), (\ui) and (\ue): these cases are trivial,
    because the precondition and the environment does not change
    between the conclusion and the premisses of these rules. \qed
  \end{itemize}} 

\ThmHavocSoundness*
\proof{ For each axiom and inference rule in
  Fig. \ref{fig:havoc-rules}, with premisses
  $\havoctriple{\aenv_i}{\phi_i}{\are_i}{\psi_i}$, for $i = 1, \ldots,
  k$, $k \geq 0$, and conclusion
  $\havoctriple{\aenv}{\phi}{\are}{\psi}$, we prove that:
  \[(\star)~ \models \havoctriple{\aenv}{\phi}{\are}{\psi} \text{, if }
  \models \havoctriple{\aenv_i}{\phi_i}{\are_i}{\psi_i} \text{, for
    all $i \in \interv{1}{k}$}\] Let us show first that ($\star$) is a
  sufficient condition. If $\Vdash
  \havoctriple{\aenv}{\phi}{\are}{\psi}$ then there exists a cyclic
  proof whose root is labeled by
  $\havoctriple{\aenv}{\phi}{\are}{\psi}$ and we apply the principle
  of infinite descent to prove that $\models
  \havoctriple{\aenv}{\phi}{\are}{\psi}$. Suppose, for a
  contradiction, that this is not the case. Assuming that ($\star$)
  holds, each invalid node, with label
  $\havoctriple{\aenv_i}{\phi_i}{\are_i}{\psi_i}$ and counterexample
  $\aconfig_i$, not on the frontier of the proof tree, has a
  successor, whose label is invalid, for all $i \geq 0$. Let
  $\patoms{\phi_i} = \set{\apred^i_1(\vec{y}^i_1), \ldots,
    \apred^i_{k_i}(\vec{y}^i_{k_i})}$ be the set of predicate atoms
  from $\phi_i$, for each $i \geq 0$. Consequently, there exists a set
  of configurations $\Gamma_i = \set{\aconfig^i_0, \ldots,
    \aconfig^i_{k_i}}$, such that $\aconfig_i = \aconfig^i_0 \comp
  \ldots \comp \aconfig^i_{k_i}$ and $\aconfig^i_j \models
  \apred^i_j(\vec{y}^i_j)$, for all $j \in \interv{1}{k_i}$ and all $i
  \geq 0$.

  \begin{fact}\label{fact:cex}
    For each $i \geq 0$, either $\Gamma_{i+1} \subseteq \Gamma_i$ or
    there exists $j \in \interv{1}{k_i}$, such that $\Gamma_{i+1} =
    \big(\Gamma_i \setminus \set{\aconfig^i_j}\big) \cup
    \set{\aconfig' \in \sem{\apred(x_1, \ldots, x_{\#(\apred)})}{}
      \mid \aconfig' \substreq \aconfig^i_j,~ \apred(x_1, \ldots,
      x_{\#(\apred)}) \in
      \patoms{\varphi^i_j[\vec{x}^i_j/\vec{y}^i_j]}}$, where
    $\apred^i_j(\vec{x}^i_j) \unfoldrule \exists \vec{z}^i_j ~.~
    \varphi^i_j$ is a rule of the SID and $\varphi^i_j$ is a symbolic
    configuration.
  \end{fact}
  \proof{ By inspection of the inference rules in
    Fig. \ref{fig:havoc-rules}b-e. The only interesting cases are: \begin{itemize}
    \item (\parr) in this case $\Gamma_{i+1}\subseteq\Gamma_i$,
      because the models of the preconditions from the premises are
      subconfigurations of the model of the precondition in the
      conclusion,
    \item (\lu) in this case $\Gamma_{i+1}$ is obtained by replacing
      an element $\aconfig^i_j$ from $\Gamma_i$ with a set of
      configurations $\aconfig'$, such that $\aconfig' \substreq
      \aconfig^i_j$ and $\aconfig'$ is a model of a predicate atom
      from an unfolding of the predicate atom for which $\aconfig^i_j$
      is a model. \qed
    \end{itemize}}

  For a configuration $\aconfig^i_j \in
  \sem{\apred^i_j(\vec{x}^i_j)}{}$, we denote by $\stageno(i,j)$ the
  minimum number of steps needed to evaluate the $\models$ relation in
  the given SID. Since $\apred(y_1, \ldots, y_{\#(\apred)})$ is
  precise on $\sem{\phi * \apred(y_1, \ldots, y_{\#(\apred)})}{}$, we
  have $\semlang{\interof{\apred(y_1, \ldots,
      y_{\#(\apred)})}}{\aconfig'} = \bigcup_{\alpha \in
    \atoms{\varphi}} \semlang{\interof{\alpha}}{\aconfig'}$. Since
  $\semlang{\interof{\alpha}}{\aconfig'} =
  \semlang{\interof{\alpha}}{\aconfig}$, for each $\alpha \in
  \atoms{\phi} \cup \set{\apred(y_1, \ldots, y_{\#(\apred)})}$, we
  obtain that $\semlang{\interof{\alpha_1}}{\aconfig} \cap
  \semlang{\interof{\alpha_2}}{\aconfig} = \emptyset$, for all
  $\alpha_1, \alpha_2 \in \atoms{\phi} \cup \set{\apred(y_1, \ldots,
    y_{\#(\apred)})}$, leading to the fact that
  $\havoctriple{\aenv}{\phi *
    \apred(y_1,\ldots,y_{\#(\apred)})}{\are}{\psi}$ is
  distinctive. Let $\stagemset_i$ be the multiset of numbers
  $\stageno(i,j)$, for all $j \in \interv{1}{k_i}$ and $i \geq 0$. By
  Fact \ref{fact:cex}, the sequence of multisets $\stagemset_0,
  \stagemset_1, \ldots$ is such that either $\stagemset_i =
  \stagemset_{i+1}$ or $\stagemset_i \succ \stagemset_{i+1}$, where
  the Dershowitz-Manna multiset ordering $\prec$ is defined as
  $\stagemset \prec \stagemset'$ if and only if there exist two
  multisets $X$ and $Y$, such that $X \neq \emptyset$, $X \subseteq
  \stagemset'$, $\stagemset = (\stagemset' \setminus X) \cup Y$, and
  for all $y \in Y$ there exists some $x \in X$, such that $y < x$. By
  the fact that the cyclic proof tree is a cyclic proof, the infinite
  path goes infinitely often via a node whose label is the conclusion
  of the application of (\lu). Then the infinite sequence of multisets
  $\stagemset_0, \stagemset_1, \ldots$ contains a strictly decreasing
  subsequence in the multiset order, which contradicts the fact that
  $\prec$ is well-founded.

  Let $\wopen{w}(\aconfig) \isdef \set{\aconfig' \mid
    \aconfig\Open{w}{}\aconfig'}$, where the relation $\aconfig
  \Open{w}{} \aconfig'$ is defined in Def. \ref{def:open}. We are left
  with proving ($\star$) for each type of axiom and inference rule in
  Fig. \ref{fig:havoc-rules}: \begin{itemize}
  \item (\epsilonr) For each configuration $\aconfig$, we have
    $\semlang{\epsilon}{\aconfig} = \set{\epsilon}$ and
    $\wopen{\epsilon}\config = \set{\config}$. 
  \item (\disr) In any model $\config$ of $\phi$, we have
    $\wopen{(\store(x_1), p_1, \store(x_2), p_2)}\config = \emptyset$,
    because of the side condition
    $\disabled{\phi}{\interac{x_1}{p_1}{x_2}{p_2}}$
    (Def. \ref{def:disabled-excluded}).
  \item (\botr) Because the precondition has no models.
  \item (\inter) By an application of Def. \ref{def:open}.
  \item (\supr) Let $\aconfig = \config \in \sem{\phi *
    \interac{x_1}{p_1}{x_2}{p_2}}{}$ be a configuration. By Lemma
    \ref{lemma:distinctive}, we have that $\havoctriple{\aenv}{\phi *
      \interac{x_1}{p_1}{x_2}{p_2}}{\are}{\psi *
      \interac{x_1}{p_1}{x_2}{p_2}}$ is distinctive. By the side
    condition $\interof{\interac{x_1}{p_1}{x_2}{p_2}} \in \aenv
    \setminus \supp{\are}$, it follows that
    $\semlang{\interac{x_1}{p_1}{x_2}{p_2}}{\aconfig}$ is disjoint
    from the interpretation $\semlang{\interof{\alpha}}{\aconfig}$ of
    any alphabet symbol $\interof{\alpha} \in \supp{\are}$, hence the
    interaction $(\store(x_1), p_1, \store(x_2), p_2)$ does not occur
    in $\semlang{\are}{\aconfig}$. By the inductive hypothesis, we
    have $\models \havoctriple{\aenv}{\phi}{\are}{\psi}$, which leads
    to the required $\models \havoctriple{\aenv}{\phi *
      \interac{x_1}{p_1}{x_2}{p_2}}{\are}{\psi *
      \interac{x_1}{p_1}{x_2}{p_2}}$.
  \item (\ie) Let $\aconfig = \config \in \sem{\phi}{}$ be a
    configuration and $\omega \isdef \interof{\alpha_1} \cdot \ldots
    \cdot \interof{\alpha_k}$ be a finite concatenation of alphabet
    symbols from $\supp{\are}$. If $\alpha_i =
    \interac{x_1}{p_1}{x_2}{p_2}$, for some $i \in \interv{1}{k}$,
    then we have $\semlang{\interof{\alpha_i}}{\aconfig} = \emptyset$,
    because of the side condition
    $\excluded{\phi}{\interac{x_1}{p_1}{x_2}{p_2}}$. Then
    $\wopen{w}(\aconfig) = \emptyset \subseteq \sem{\psi}{}$, for each
    $w \in \semlang{\omega}{\aconfig}$. Otherwise, if
    $\interac{x_1}{p_1}{x_2}{p_2}$ does not occur on $\omega$, then
    $\wopen{w}(\aconfig) \subseteq \sem{\psi}{}$, for each $w \in
    \semlang{\are}{\aconfig} \cap \semlang{\omega}{\aconfig}$, by the
    inductive hypothesis.
    %
    %
  \item (\parr) Let $\aconfig = \config \in \sem{\phi_1 * \phi_2}{}$
    be a configuration. Then there exist configurations $\aconfig_i =
    (\comps_i, \interacs_i, \statemap_i, \store) \in \sem{\phi_i}{}$,
    for $i = 1,2$, such that $\aconfig = \aconfig_1 \comp
    \aconfig_2$. Let $\aconfig'_i$ be configurations such that
    $\aconfig'_i \isubstreq \aconfig_{3-i}$ and $\aconfig'_i \models
    \front{\phi_i}{\phi_{3-i}}$, for $i = 1,2$. Because
    $\front{\phi_i}{\phi_{3-i}}$ is a separated conjunction of
    interaction atoms, each of which is precise on $\configset$, it
    follows that $\front{\phi_i}{\phi_{3-i}}$ is precise on
    $\configset$, thus $\aconfig'_i$ are unique, for $i = 1,2$. Let
    $\aconfig''_i \isdef \aconfig_i \comp \aconfig'_{3-i}$, for $i =
    1,2$.  Moreover, since $\aenv_i = \interof{\phi_i *
      \front{\phi_i}{\phi_{3-i}}}$, the only interactions in
    $\semlang{\aenv_i}{\aconfig}$ are the ones in
    $\semlang{\aenv_i}{\aconfig''_i}$, hence
    $\semlang{\aenv_i}{\aconfig} = \semlang{\aenv_i}{\aconfig''_i}$,
    for $i = 1,2$. Let $w \in \semlang{\are_1 \parcomp_{\aenv_1,
        \aenv_2} \are_2}{\aconfig}$ be a word. Then
    $\proj{w}{\semlang{\aenv_i}{\aconfig}} \in
    \semlang{\are_i}{\aconfig}$, for $i = 1,2$. Because
    $\semlang{\aenv_i}{\aconfig} = \semlang{\aenv_i}{\aconfig''_i}$,
    we obtain $\proj{w}{\semlang{\aenv_i}{\aconfig}} =
    \proj{w}{\semlang{\aenv_i}{\aconfig''_i}} \in
    \semlang{\are_i}{\aconfig''_i}$, for $i = 1,2$. Since, moreover,
    $\aconfig''_i \models \phi_i * \front{\phi_i}{\phi_{3-i}}$, by the
    inductive hypothesis we obtain that
    $\wopen{\proj{w}{\semlang{\aenv_i}{\aconfig''_i}}}(\aconfig''_i)
    \subseteq \sem{\psi_i * \front{\phi_i}{\phi_{3-i}}}{}$, for $i =
    1,2$. We partition $w = w_1 w'_1 w''_1 w'''_1 \ldots w_k w'_k
    w''_k w'''_k$, for some $k \geq 1$, in three types of (possibly
    empty) blocks, such that, for all $j \in \interv{1}{k}$, we
    have: \begin{itemize}
    \item $w_j \in \big(\semlang{\aenv_1}{\aconfig''_1} \setminus
      \semlang{\aenv_2}{\aconfig''_2}\big)^*$,
    \item $w'_j, w'''_j \in \big(\semlang{\aenv_1}{\aconfig''_1} \cap
      \semlang{\aenv_2}{\aconfig''_2}\big)^*$, and
    \item $w''_j \in \big(\semlang{\aenv_2}{\aconfig''_2} \setminus
      \semlang{\aenv_1}{\aconfig''_1}\big)^*$.
    \end{itemize}
    If $\wopen{w}(\aconfig) = \emptyset$, there is nothing to
    prove. Otherwise, let $\aconfig' =
    (\comps,\interacs,\statemap',\store) \in \wopen{w}(\aconfig)$ and
    $\statemap_1 \isdef \statemap, \statemap'_1, \statemap''_1,
    \statemap'''_1, \ldots, \statemap_k, \statemap'_k, \statemap''_k,
    \statemap'''_k \isdef \statemap'$ be a sequence of state maps such
    that, for all $j \in \interv{1}{k}$, we have
    $(\comps,\interacs,\statemap'_j,\store) \in
    \wopen{w_j}(\comps,\interacs,\statemap_j,\store)$,
    $(\comps,\interacs,\statemap''_j,\store) \in
    \wopen{w'_j}(\comps,\interacs,\statemap'_j,\store)$,
    $(\comps,\interacs,\statemap'''_j,\store) \in
    \wopen{w''_j}(\comps,\interacs,\statemap''_j,\store)$, and
    $(\comps,\interacs,\statemap_{j+1},\store) \in
    \wopen{w'''_j}(\comps,\interacs,\statemap'''_j,\store)$, if $j <
    k$, in particular. Let $\statemap_{i,j}$, $\statemap'_{i,j}$,
    $\statemap''_{i,j}$ and $\statemap'''_{i,j}$ be the restrictions
    of $\statemap_j$, $\statemap'_j$, $\statemap''_j$ and
    $\statemap'''_j$ to $\comps_i$, for $i = 1,2$, respectively. We
    prove the following: \begin{enumerate}
    \item\label{it1:parcomp-soundness} $\statemap_{2,j} =
      \statemap'_{2,j}$, for all $j \in \interv{1}{k}$, and
    \item\label{it2:parcomp-soundness} $\statemap''_{1,j} =
      \statemap'''_{1,j}$, for all $j \in \interv{1}{k}$.
    \end{enumerate}
    We prove the first point, the argument for the second point being
    symmetric. It is sufficient to prove that the state of the
    components with indices in $\comps_2$, which are the only ones
    $\statemap_{2,j}$ and $\statemap'_{2,j}$ account for, is not
    changed by $w_j$, for all $j \in \interv{1}{k}$. Since $w_j \in
    \big(\semlang{\aenv_1}{\aconfig''_1} \setminus
    \semlang{\aenv_2}{\aconfig''_2}\big)^*$, the only interactions on
    $w_j$ are the ones from $\aconfig''_1 = \aconfig_1 \comp
    \aconfig'_1$ that do not occur in $\aconfig''_2 = \aconfig_2 \comp
    \aconfig'_2$, where $\aconfig'_1 \isubstreq \aconfig_2$ and
    $\aconfig'_2 \isubstreq \aconfig_1$. It follows that the
    interactions occurring on $w_j$ are the ones from $\aconfig_1$
    that do not occur in $\aconfig''_2$. Since $\aconfig''_2 \models
    \Asterisk_{\alpha \in \iatoms{\phi_1} \setminus
      (\iatoms{\overline{\phi}_1} \cup \iatoms{\phi_2})} ~\alpha$ [$=
      \front{\phi_2}{\phi_1}$], the interactions occurring on $w_j$
    must occur in some model $\overline{\aconfig}$ of a tight
    subformula of $\phi_1$. Hence, the interactions from
    $\overline{\aconfig}$ can only change the state of a component
    from $\overline{\aconfig}$. Since $\overline{\aconfig} \substreq
    \aconfig_1$ and $\aconfig_1 \comp \aconfig_2$ is defined, there
    can be no component indexed by some element of $\mathfrak{C}_2$,
    whose state is changed by an interaction from
    $\overline{\aconfig}$, thus $\statemap_{2,j} = \statemap'_{2,j}$
    (\ref{it1:parcomp-soundness}). Consequently, we obtain two
    sequences of words and finite state maps: \begin{itemize}
    \item $w_1, w'_1, w'''_1, \ldots, w_k, w'_k, w'''_k$ and
      $\statemap_{1,1}, \statemap'_{1,1}, \statemap''_{1,1}, \ldots,
      \statemap_{1,k}, \statemap'_{1,k}, \statemap''_{1,k}$,
      where: \begin{itemize}
      \item $(\comps_1,\interacs_1,\statemap'_{1,j},\store) \in
        \wopen{w_1}(\comps_1,\interacs_1,\statemap_{1,j},\store)$,
      \item $(\comps_1,\interacs_1,\statemap''_{1,j},\store) \in
        \wopen{w'_1}(\comps_1,\interacs_1,\statemap'_{1,j},\store)$ and
      \item $(\comps_1,\interacs_1,\statemap_{1,j+1},\store) \in
        \wopen{w'''_1}(\comps_1,\interacs_1,\statemap''_{1,j},\store)$,
      \end{itemize}
      for all $j \in \interv{1}{k-1}$, and
    \item $w'_1, w''_1, w'''_1, \ldots, w'_k, w''_k, w'''_k$ and
      $\statemap'_{2,1}, \statemap''_{2,1}, \statemap'''_{2,1},
      \ldots, \statemap'_{2,k}, \statemap''_{2,k},
      \statemap'''_{2,k}$, where: \begin{itemize}
      \item $(\comps_2,\interacs_2,\statemap''_{2,j},\store) \in
        \wopen{w'_1}(\comps_2,\interacs_2,\statemap'_{2,j},\store)$,
      \item $(\comps_2,\interacs_2,\statemap'''_{2,j},\store) \in
        \wopen{w''_1}(\comps_2,\interacs_2,\statemap''_{2,j},\store)$
        and
      \item $(\comps_2,\interacs_2,\statemap'_{2,j+1},\store) \in
        \wopen{w'''_1}(\comps_2,\interacs_2,\statemap'''_{2,j},\store)$,
      \end{itemize}
      for all $j \in \interv{1}{k-1}$.
    \end{itemize}
    Note that $\proj{w}{\semlang{\aenv_1}{\aconfig''_1}} = w_1 w'_1
    w'''_1 \ldots w_k w'_k w'''_k$ and
    $\proj{w}{\semlang{\aenv_2}{\aconfig''_2}} = w'_1 w''_1 w'''_1
    \ldots w'_k w''_k w'''_k$. By the inductive hypothesis, we have
    $\wopen{\proj{w}{\semlang{\aenv_i}{\aconfig''_i}}}(\aconfig''_i)
    \subseteq \sem{\psi_i * \front{\phi_i}{\phi_{3-i}}}{}$, hence
    $(\comps_i,\interacs_i,\statemap'''_{i,k},\store) \in
    \wopen{\proj{w}{\semlang{\aenv_i}{\aconfig''_i}}}(\aconfig''_i)$,
    for $i = 1,2$. Moreover, $\comps_1 \uplus \comps_2 = \comps$,
    $\interacs_1 \uplus \interacs_2 = \interacs$ and the state maps
    $\statemap'$ and $\statemap'''_{1,k} \cup \statemap'''_{2,k}$
    agree on all $c \in \comps$, hence $\aconfig' =
    (\comps,\interacs,\statemap',\store) \in \sem{\psi_1 * \psi_2}{}$.
  \end{itemize}
  Proving ($\star$) for the rest of the rules is a standard check,
  left to the reader. \qed}

\section{Proofs from Section \ref{sec:trees}}

\PropTreePrecise* \proof{ Let $\varphi_i \isdef \phi_i *
  \Asterisk_{j=1}^{k_i} \treestar(x_{i,j}) * \Asterisk_{j = k_i +
    1}^{\ell_i} \treeseg(x_{i,j},y_{i,j})$ be symbolic configurations,
  where $\phi_i$ is a predicate-free symbolic configuration and
  $\treestar(x)$ is either $\treeidle(x)$, $\treenotidle(x)$ or
  $\tree(x)$, for all $j \in \interv{1}{\ell_i}$ and $i = 1,2$. We
  prove that $\varphi_1$ is precise on $\sem{\varphi_2}{}$.  Let
  $\aconfig = \config \in \sem{\varphi_2}{}$ be a configuration and
  suppose that there exist configurations $\aconfig' =
  (\comps',\interacs',\statemap,\store), \aconfig'' =
  (\comps'',\interacs'',\statemap,\store)$, such that $\aconfig'
  \isubstreq \aconfig$, $\aconfig'' \isubstreq \aconfig$, $\aconfig'
  \models \varphi_1$ and $\aconfig'' \models \varphi_1$. Then there
  exist configurations $\aconfig'_0 \isdef
  (\comps'_0,\interacs'_0,\statemap'_0,\store), \ldots,
  \aconfig'_{\ell_1} \isdef
  (\comps'_{\ell_1},\interacs'_{\ell_1},\statemap'_{\ell_1},\store)$
  and $\aconfig''_0 \isdef
  (\comps''_0,\interacs''_0,\statemap''_0,\store), \ldots,
  \aconfig''_{\ell_1} \isdef
  (\comps''_{\ell_1},\interacs''_{\ell_1},\statemap''_{\ell_1},\store)$,
  such that: \begin{itemize}
  \item $\aconfig' = \Asterisk_{j=0}^{\ell_1} \aconfig'_j$ and
    $\aconfig'' = \Asterisk_{j=0}^{\ell_1} \aconfig''_j$,
  \item $\aconfig'_0 \models \phi_1$ and $\aconfig''_0 \models
    \phi_1$,
  \item $\aconfig'_j \models \treestar(x_{1,j})$ and
    $\aconfig''_j \models \treestar(x_{2,j})$, for all $j \in
    \interv{1}{k_1}$, and 
  \item $\aconfig'_j \models \treeseg(x_{1,j},y_{1,j})$ and
    $\aconfig''_j \models \treeseg(x_{2,j},y_{2,j})$, for all $j
    \in \interv{k_1+1}{\ell_1}$. 
  \end{itemize}
  Since $\phi_1$ is a predicate-free symbolic configuration, we have
  $\comps'_0 = \comps''_0$, $\interacs'_0 = \interacs''_0$ and
  $\statemap'_0 = \statemap''_0$, thus $\aconfig'_0 =
  \aconfig''_0$. Next, for each $j \in \interv{1}{k_1}$, we have
  $\comps'_j = \comps''_j$, because these sets of components
  correspond to the vertices of the same tree, whose root is
  $\store(x_{1,j})$ and whose frontier contains only indices $c \in
  \comps'_j \cap \comps''_j$, such that $\statemap(c) \in
  \set{\stateleafidle,\stateleafbusy}$. Finally, for each $j \in
  \interv{k_1+1}{\ell_1}$, we have $\interacs'_j = \interacs''_j$,
  because these sets of interactions correspond to the edges of the
  same tree, whose root is $\store(x_{1,j})$ and whose frontier
  contains $\store(y_{1,j})$ together with indices $c \in \comps'_j
  \cap \comps''_j$, such that $\statemap(c) \in
  \set{\stateleafidle,\stateleafbusy}$. We obtain, consequently, that
  $\aconfig'_j = \aconfig''_j$, for all $j \in \interv{1}{\ell_1}$,
  leading to $\aconfig' = \aconfig''$. \qed}

\section{Havoc Invariance Proofs from Section \ref{sec:trees}}
\label{app:trees}

In order to shorten the following proofs, we introduce the rule
(\iidisr) that allows us to remove a disabled interaction atom
$\alpha$ from the pre- and postcondition, the environment and the
language if certain conditions hold.

\begin{lemma}
	Using the notation in \S\ref{sec:havoc}, the following rule is sound:
	{ \scriptsize
		\begin{prooftree}
	    \AxiomC{$\havoctriple{\aenv \setminus \set{\interof{\alpha}}}{\phi}{\are}{\psi}$}
	    \LeftLabel{(\iidisr)}
	    \RightLabel{$\begin{array}{l}
	        \scriptstyle{\alpha = \interac{x_1}{p_1}{x_2}{p_2}} \\[-.5mm]
	        \scriptstyle{\interof{\alpha} \in \aenv \setminus \supp{\are}} \\[-.5mm]
	        \scriptstyle{\disabled{\phi}{\alpha}.}
	      \end{array}$}
	    \UnaryInfC{$\havoctriple{\aenv}{\phi * \alpha}{\interof{\alpha} \cup \are}{\psi * \alpha}$}
		\end{prooftree}
	}
\end{lemma}

\begin{proof}
	We assume that $\phi$ and $\psi$ are two symbolic
        configurations, $\aenv$ is an environment and $\alpha =
        \interac{x_1}{p_1}{x_2}{p_2}$ an interaction
        atom. Furthermore, $\interof {\alpha} \in \aenv\setminus
        \supp{\are}$ and $\disabled{\phi}{\alpha}$. Then we can apply
        the rule ($\cup$) first and the rules (\conseq), (\inter) and
        (\ii) on the subtrees and obtain: { \tiny
		\begin{prooftree}
			\AxiomC{}
			\LeftLabel{(\disr)}
		    \havocline{\aenv}{\phi * \alpha}{\interof{\alpha}}{\predfalse}
		    \LeftLabel{(\conseq)}
		    \havocline{\aenv}{\phi * \alpha}{\interof{\alpha}}{\psi * \alpha}

		    \AxiomC{$\havoctriple{\aenv \setminus \set{\interof{\alpha}}}{\phi}{\are}{\psi}$}
		    \LeftLabel{(\ii)}
		    \havocline{\aenv}{\phi * \alpha}{\are}{\psi * \alpha}
		    \LeftLabel{($\cup$)}
		    \BinaryInfC{$\havoctriple{\aenv}{\phi * \alpha}{\interof{\alpha} \cup \are}
		    	{\psi * \alpha}$.}
		\end{prooftree}
	}
	\noindent Hence the rule can by derived from the rules in Fig.~\ref{fig:havoc-rules}.
\end{proof}

\subsection{Havoc Invariance of the Predicate Atom $\tree(x)$}

The invariance of the predicate $\tree(x)$ is proven via the rules in
Fig.~\ref{fig:havoc-rules}. The proof is divided into subtrees labeled by
letters. Backlinks are indicated by numbers and in each cycle in the proof tree
the rule (\lu) is applied at least once.

{ \tiny
	\begin{prooftree}
		\AxiomC{}
		\LeftLabel{(\epsilonr)}
		\havoctwolines{\emptyset}
			{\compin{x}{\stateleafidle}}
			{\epsilon}
			{\tree(x)}

		\AxiomC{}
		\LeftLabel{(\epsilonr)}
		\havoctwolines{\emptyset}
			{\compin{x}{\stateleafbusy}}
			{\epsilon}
			{\tree(x)}

		\AxiomC{\textbf{(A)}}
		\havocfourlines{\set{\interof{\interac{y}{\lrecv}{x}{\send}}, \interof{\interac{z}{\rrecv}{x}{\send}},
				\interoftree{y}, \interoftree{z}}}
			{\company{x} \ast \interac{y}{\lrecv}{x}{\send} \ast \interac{z}{\rrecv}{x}{\send}
				\ast \tree(y) \ast \tree(z)}
			{\interof{\interac{y}{\lrecv}{x}{\send}} \cup \interof{\interac{z}{\rrecv}{x}{\send}} \cup
				\interoftree{y} \cup \interoftree{z}}
			{\tree(x)}
		\LeftLabel{(\lu)}
		\trihavocline{\text{\textbf{(1)} }\set{\interoftree{x}}}
			{\tree(x)}
			{\interoftree{x}}
			{\tree(x)}
		\LeftLabel{($\ast$)}
		\havocline{\set{\interoftree{x}}}
			{\tree(x)}
			{\interoftree{x}^*}
			{\tree(x)}
	\end{prooftree}
}

\hspace{3mm}

{ \tiny
	\begin{prooftree}
    \def\defaultHypSeparation{\hskip .05in}
		\AxiomC{}
		\havocfourlines{\textbf{(B) }\set{\interof{\interac{y}{\lrecv}{x}{\send}}, \interof{\interac{z}{\rrecv}{x}{\send}},
				\interoftree{y}, \interoftree{z}}}
			{\company{x} \ast \interac{y}{\lrecv}{x}{\send} \ast \interac{z}{\rrecv}{x}{\send}
				\ast \tree(y) \ast \tree(z)}
			{\interof{\interac{y}{\lrecv}{x}{\send}}}
			{\company{x} \ast \interac{y}{\lrecv}{x}{\send} \ast \interac{z}{\rrecv}{x}{\send}
				\ast \tree(y) \ast \tree(z)}

		\AxiomC{similar to \textbf{(B)}}

		\AxiomC{}
		\havocfourlines{\textbf{(C) }\set{\interof{\interac{y}{\lrecv}{x}{\send}}, \interof{\interac{z}{\rrecv}{x}{\send}},
				\interoftree{y}, \interoftree{z}}}
			{\company{x} \ast \interac{y}{\lrecv}{x}{\send} \ast \interac{z}{\rrecv}{x}{\send}
				\ast \tree(y) \ast \tree(z)}
			{\interoftree{y}}
			{\company{x} \ast \interac{y}{\lrecv}{x}{\send} \ast \interac{z}{\rrecv}{x}{\send}
				\ast \tree(y) \ast \tree(z)}

		\AxiomC{similar to \textbf{(C)}}

		\LeftLabel{($\cup$)}
		\quahavoctwolines{\set{\interof{\interac{y}{\lrecv}{x}{\send}}, \interof{\interac{z}{\rrecv}{x}{\send}},
				\interoftree{y}, \interoftree{z}}}
			{\company{x} \ast \interac{y}{\lrecv}{x}{\send} \ast \interac{z}{\rrecv}{x}{\send}
				\ast \tree(y) \ast \tree(z)}
			{\interof{\interac{y}{\lrecv}{x}{\send}} \cup \interof{\interac{z}{\rrecv}{x}{\send}} \cup
				\interoftree{y} \cup \interoftree{z}}
			{\company{x} \ast \interac{y}{\lrecv}{x}{\send} \ast \interac{z}{\rrecv}{x}{\send}
				\ast \tree(y) \ast \tree(z)}
		\LeftLabel{(\conseq)}
		\havoctwolines{\textbf{(A) }\set{\interof{\interac{y}{\lrecv}{x}{\send}}, \interof{\interac{z}{\rrecv}{x}{\send}},
				\interoftree{y}, \interoftree{z}}}
			{\company{x} \ast \interac{y}{\lrecv}{x}{\send} \ast \interac{z}{\rrecv}{x}{\send}
				\ast \tree(y) \ast \tree(z)}
			{\interof{\interac{y}{\lrecv}{x}{\send}} \cup \interof{\interac{z}{\rrecv}{x}{\send}} \cup
				\interoftree{y} \cup \interoftree{z}}
			{\tree(x)}
	\end{prooftree}
}

\hspace{3mm}

{ \tiny
	\begin{prooftree}
		\AxiomC{\textbf{(D)}}
		\AxiomC{\textbf{(E)}}
		\AxiomC{\textbf{(F)}}
		\LeftLabel{(\lu)}
		\trihavoctwolines{\set{\interof{\interac{y}{\lrecv}{x}{\send}},
				\interoftree{y}}}
			{\company{x} \ast \interac{y}{\lrecv}{x}{\send} \ast \tree(y)}
			{\interof{\interac{y}{\lrecv}{x}{\send}}}
			{\company{x} \ast \interac{y}{\lrecv}{x}{\send} \ast \tree(y)}

		\AxiomC{}
		\LeftLabel{(\epsilonr)}
		\havoctwolines{\set{\interoftree{z}}}
			{\tree(z)}
			{\epsilon}
			{\tree(z)}
		\LeftLabel{(\parr)}
		\binhavocline{\set{\interof{\interac{y}{\lrecv}{x}{\send}},
				\interoftree{y}, \interoftree{z}}}
			{\company{x} \ast \interac{y}{\lrecv}{x}{\send} \ast \tree(y) \ast \tree(z)}
			{\interof{\interac{y}{\lrecv}{x}{\send}}}
			{\company{x} \ast \interac{y}{\lrecv}{x}{\send} \ast \tree(y) \ast \tree(z)}
		\LeftLabel{(\ii)}
		\havoctwolines{\textbf{(B) }\set{\interof{\interac{y}{\lrecv}{x}{\send}}, \interof{\interac{z}{\rrecv}{x}{\send}},
				\interoftree{y}, \interoftree{z}}}
			{\company{x} \ast \interac{y}{\lrecv}{x}{\send} \ast \interac{z}{\rrecv}{x}{\send}
				\ast \tree(y) \ast \tree(z)}
			{\interof{\interac{y}{\lrecv}{x}{\send}}}
			{\company{x} \ast \interac{y}{\lrecv}{x}{\send} \ast \interac{z}{\rrecv}{x}{\send}
				\ast \tree(y) \ast \tree(z)}
	\end{prooftree}
}

\hspace{3mm}

{ \tiny
	\begin{prooftree}
		\AxiomC{backlink to \textbf{(1)}}
		\havocline{\set{\interoftree{y}}}
			{\tree(y)}
			{\interoftree{y}}
			{\tree(y)}

		\AxiomC{}
		\LeftLabel{(\epsilonr)}
		\havocline{\emptyset}
			{\company{x} \ast \tree(z)}
			{\epsilon}
			{\company{x} \ast \tree(z)}
		\LeftLabel{(\parr)}
		\binhavocline{\set{\interoftree{y}}}
			{\company{x} \ast \tree(y) \ast \tree(z)}
			{\interoftree{y}}
			{\company{x} \ast \tree(y) \ast \tree(z)}
		\LeftLabel{(\ii)}
		\havoctwolines{\textbf{(C) }\set{\interof{\interac{y}{\lrecv}{x}{\send}}, \interof{\interac{z}{\rrecv}{x}{\send}},
				\interoftree{y}, \interoftree{z}}}
			{\company{x} \ast \interac{y}{\lrecv}{x}{\send} \ast \interac{z}{\rrecv}{x}{\send}
				\ast \tree(y) \ast \tree(z)}
			{\interoftree{y}}
			{\company{x} \ast \interac{y}{\lrecv}{x}{\send} \ast \interac{z}{\rrecv}{x}{\send}
				\ast \tree(y) \ast \tree(z)}
	\end{prooftree}
}

\hspace{3mm}

{ \tiny
	\begin{prooftree}
		\AxiomC{}
		\LeftLabel{(\disr)}
		\havocline{\set{\interof{\interac{y}{\lrecv}{x}{\send}}}}
			{\company{x} \ast \interac{y}{\lrecv}{x}{\send} \ast \compin{y}{\stateleafidle}}
			{\interof{\interac{y}{\lrecv}{x}{\send}}}
			{\predfalse}
		\LeftLabel{(\conseq)}
		\havocline{\textbf{(D) }\set{\interof{\interac{y}{\lrecv}{x}{\send}}}}
			{\company{x} \ast \interac{y}{\lrecv}{x}{\send} \ast \compin{y}{\stateleafidle}}
			{\interof{\interac{y}{\lrecv}{x}{\send}}}
			{\company{x} \ast \interac{y}{\lrecv}{x}{\send} \ast \tree(y)}
	\end{prooftree}
}

\hspace{3mm}

{ \tiny
	\begin{prooftree}
		\AxiomC{}
		\LeftLabel{(\inter)}
		\havoctwolines{\set{\interof{\interac{y}{\lrecv}{x}{\send}}}}
			{\compin{x}{\stateidle} \ast \interac{y}{\lrecv}{x}{\send} \ast \compin{y}{\stateleafbusy}}
			{\interof{\interac{y}{\lrecv}{x}{\send}}}
			{\compin{x}{\stateleft} \ast \interac{y}{\lrecv}{x}{\send} \ast \compin{y}{\stateleafidle}}
		\LeftLabel{(\conseq)}
		\havoctwolines{\set{\interof{\interac{y}{\lrecv}{x}{\send}}}}
			{\compin{x}{\stateidle} \ast \interac{y}{\lrecv}{x}{\send} \ast \compin{y}{\stateleafbusy}}
			{\interof{\interac{y}{\lrecv}{x}{\send}}}
			{\company{x} \ast \interac{y}{\lrecv}{x}{\send} \ast \tree(y)}

		\AxiomC{}
		\LeftLabel{(\disr)}
		\havoctwolines{\set{\interof{\interac{y}{\lrecv}{x}{\send}}}}
			{\compin{x}{q} \ast \interac{y}{\lrecv}{x}{\send} \ast \compin{y}{\stateleafbusy}}
			{\interof{\interac{y}{\lrecv}{x}{\send}}}
			{\predfalse}
		\LeftLabel{(\conseq)}
		\RightLabel{for $q\neq \stateidle$}
		\havoctwolines{\set{\interof{\interac{y}{\lrecv}{x}{\send}}}}
			{\compin{x}{q} \ast \interac{y}{\lrecv}{x}{\send} \ast \compin{y}{\stateleafbusy}}
			{\interof{\interac{y}{\lrecv}{x}{\send}}}
			{\company{x} \ast \interac{y}{\lrecv}{x}{\send} \ast \tree(y)}
		\LeftLabel{($\vee$)}
		\binhavocline{\textbf{(E) }\set{\interof{\interac{y}{\lrecv}{x}{\send}}}}
			{\company{x} \ast \interac{y}{\lrecv}{x}{\send} \ast \compin{y}{\stateleafbusy}}
			{\interof{\interac{y}{\lrecv}{x}{\send}}}
			{\company{x} \ast \interac{y}{\lrecv}{x}{\send} \ast \tree(y)}
	\end{prooftree}
}

\hspace{3mm}

{ \tiny
	\begin{prooftree}
		\AxiomC{backlink to \textbf{(D)}}
		\AxiomC{backlink to \textbf{(E)}}

		\LeftLabel{($\vee$)}
		\binhavoctwolines{\set{\interof{\interac{y}{\lrecv}{x}{\send}}}}
			{\company{x} \ast \interac{y}{\lrecv}{x}{\send} \ast \company{y}}
			{\interof{\interac{y}{\lrecv}{x}{\send}}}
			{\company{x} \ast \interac{y}{\lrecv}{x}{\send} \ast \company{y}}

		\AxiomC{}
		\LeftLabel{(\epsilonr)}
		\havoctwolines{\set{\interof{\tree(v)}, \interof{\tree(w)}}}
			{\tree(v) \ast \tree(w)}
			{\epsilon}
			{\tree(v) \ast \tree(w)}

		\LeftLabel{(\parr)}
		\binhavoctwolines{\set{\interof{\interac{y}{\lrecv}{x}{\send}}, \interof{\tree(v)},
				\interof{\tree(w)}}}
			{\company{x} \ast \interac{y}{\lrecv}{x}{\send} \ast \company{y}
				\ast \tree(v) \ast \tree(w)}
			{\interof{\interac{y}{\lrecv}{x}{\send}}}
			{\company{x} \ast \interac{y}{\lrecv}{x}{\send} \ast \company{y} \ast \interac{v}{\lrecv}{y}{\send}
				\ast \interac{w}{\rrecv}{y}{\send} \ast \tree(v) \ast \tree(w)}
		\LeftLabel{(\ii)}
		\havoctwolines{\set{\interof{\interac{y}{\lrecv}{x}{\send}}, \interof{\interac{v}{\lrecv}{y}{\send}},
			\interof{\interac{w}{\rrecv}{y}{\send}}, \interof{\tree(v)}, \interof{\tree(w)}}}
			{\company{x} \ast \interac{y}{\lrecv}{x}{\send} \ast \company{y} \ast \interac{v}{\lrecv}{y}{\send}
				\ast \interac{w}{\rrecv}{y}{\send} \ast \tree(v) \ast \tree(w)}
			{\interof{\interac{y}{\lrecv}{x}{\send}}}
			{\company{x} \ast \interac{y}{\lrecv}{x}{\send} \ast \company{y} \ast \interac{v}{\lrecv}{y}{\send}
				\ast \interac{w}{\rrecv}{y}{\send} \ast \tree(v) \ast \tree(w)}
		\LeftLabel{(\conseq)}
		\havoctwolines{\textbf{(F) }\set{\interof{\interac{y}{\lrecv}{x}{\send}}, \interof{\interac{v}{\lrecv}{y}{\send}},
			\interof{\interac{w}{\rrecv}{y}{\send}}, \interof{\tree(v)}, \interof{\tree(w)}}}
			{\company{x} \ast \interac{y}{\lrecv}{x}{\send} \ast \company{y} \ast \interac{v}{\lrecv}{y}{\send}
				\ast \interac{w}{\rrecv}{y}{\send} \ast \tree(v) \ast \tree(w)}
			{\interof{\interac{y}{\lrecv}{x}{\send}}}
			{\company{x} \ast \interac{y}{\lrecv}{x}{\send} \ast \tree(y)}
	\end{prooftree}
}

\hspace{3mm}

\subsection{Havoc Invariance of the Predicate Atom $\treeidle(x)$}

The invariance of the predicate $\treeidle(x)$ is proven via the rules in
Fig.~\ref{fig:havoc-rules} and the proof is structured similar to the previous
invariance proof.

{ \tiny
	\begin{prooftree}
		\AxiomC{}
		\LeftLabel{(\epsilonr)}
		\havoctwolines{\emptyset}
			{\compin{x}{\stateleafidle}}
			{\epsilon}
			{\treeidle(x)}

		\AxiomC{\textbf{(A)}}
		\havoctwolines{\set{\interof{\interac{y}{\lrecv}{x}{\send}}, \interof{\interac{z}{\rrecv}{x}{\send}},
			\interof{\treeidle(y)}, \interof{\treeidle(z)}}}
			{\compin{x}{\stateidle} \ast \interac{y}{\lrecv}{x}{\send} \ast \interac{z}{\rrecv}{x}{\send}
				\ast \treeidle(y) \ast \treeidle(z)}
			{\interof{\interac{y}{\lrecv}{x}{\send}} \cup \interof{\interac{z}{\rrecv}{x}{\send}} \cup
				\interof{\treeidle(y)} \cup \interof{\treeidle(z)}}
			{\treeidle(x)}

		\LeftLabel{(\lu)}
		\binhavocline{\textbf{(2) }\set{\interof{\treeidle(x)}}}
			{\treeidle(x)}
			{\interof{\treeidle(x)}}
			{\treeidle(x)}
		\LeftLabel{($\ast$)}
		\havocline{\set{\interof{\treeidle(x)}}}
			{\treeidle(x)}
			{\interof{\treeidle(x)}^*}
			{\treeidle(x)}
	\end{prooftree}
}

\hspace{3mm}

{ \tiny
	\begin{prooftree}
		\AxiomC{}
		\havocfourlines{\textbf{(B) }\set{\interof{\interac{y}{\lrecv}{x}{\send}},
			\interof{\treeidle(y)}, \interof{\treeidle(z)}}}
			{\compin{x}{\stateidle} \ast \interac{y}{\lrecv}{x}{\send} \ast \treeidle(y) \ast \treeidle(z)}
			{\interof{\interac{y}{\lrecv}{x}{\send}}}
			{\compin{x}{\stateidle} \ast \interac{y}{\lrecv}{x}{\send} \ast \treeidle(y) \ast \treeidle(z)}

		\AxiomC{}
		\havocfourlines{\textbf{(C) }\set{\interof{\interac{y}{\lrecv}{x}{\send}},
			\interof{\treeidle(y)}, \interof{\treeidle(z)}}}
			{\compin{x}{\stateidle} \ast \interac{y}{\lrecv}{x}{\send} \ast \treeidle(y) \ast \treeidle(z)}
			{\interof{\treeidle(y)}}
			{\compin{x}{\stateidle} \ast \interac{y}{\lrecv}{x}{\send} \ast \treeidle(y) \ast \treeidle(z)}

		\AxiomC{similar to \textbf{(C)}}
		\LeftLabel{($\cup$)}
		\trihavoctwolines{\set{\interof{\interac{y}{\lrecv}{x}{\send}},
			\interof{\treeidle(y)}, \interof{\treeidle(z)}}}
			{\compin{x}{\stateidle} \ast \interac{y}{\lrecv}{x}{\send} \ast \treeidle(y) \ast \treeidle(z)}
			{\interof{\interac{y}{\lrecv}{x}{\send}} \cup
				\interof{\treeidle(y)} \cup \interof{\treeidle(z)}}
			{\compin{x}{\stateidle} \ast \interac{y}{\lrecv}{x}{\send} \ast \treeidle(y) \ast \treeidle(z)}
		\LeftLabel{(\iidisr)}
		\havoctwolines{\set{\interof{\interac{y}{\lrecv}{x}{\send}}, \interof{\interac{z}{\rrecv}{x}{\send}},
			\interof{\treeidle(y)}, \interof{\treeidle(z)}}}
			{\compin{x}{\stateidle} \ast \interac{y}{\lrecv}{x}{\send} \ast \interac{z}{\rrecv}{x}{\send}
				\ast \treeidle(y) \ast \treeidle(z)}
			{\interof{\interac{y}{\lrecv}{x}{\send}} \cup \interof{\interac{z}{\rrecv}{x}{\send}} \cup
				\interof{\treeidle(y)} \cup \interof{\treeidle(z)}}
			{\compin{x}{\stateidle} \ast \interac{y}{\lrecv}{x}{\send} \ast \interac{z}{\rrecv}{x}{\send}
				\ast \treeidle(y) \ast \treeidle(z)}
			\LeftLabel{(\conseq)}
		\havoctwolines{\textbf{(A) }\set{\interof{\interac{y}{\lrecv}{x}{\send}}, \interof{\interac{z}{\rrecv}{x}{\send}},
			\interof{\treeidle(y)}, \interof{\treeidle(z)}}}
			{\compin{x}{\stateidle} \ast \interac{y}{\lrecv}{x}{\send} \ast \interac{z}{\rrecv}{x}{\send}
				\ast \treeidle(y) \ast \treeidle(z)}
			{\interof{\interac{y}{\lrecv}{x}{\send}} \cup \interof{\interac{z}{\rrecv}{x}{\send}} \cup
				\interof{\treeidle(y)} \cup \interof{\treeidle(z)}}
			{\treeidle(x)}
	\end{prooftree}
}

\hspace{3mm}

{ \tiny
	\begin{prooftree}
		\AxiomC{\textbf{(D)\qquad (E)}}
		\LeftLabel{(\lu)}
		\havoctwolines{\set{\interof{\interac{y}{\lrecv}{x}{\send}},
			\interof{\treeidle(y)}}}
			{\compin{x}{\stateidle} \ast \interac{y}{\lrecv}{x}{\send} \ast \treeidle(y)}
			{\interof{\interac{y}{\lrecv}{x}{\send}}}
			{\compin{x}{\stateidle} \ast \interac{y}{\lrecv}{x}{\send} \ast \treeidle(y)}

		\AxiomC{}
		\LeftLabel{(\epsilonr)}
		\havocline{\set{\interof{\treeidle(z)}}}
			{\treeidle(z)}
			{\epsilon}
			{\treeidle(z)}
		\LeftLabel{(\parr)}
		\binhavoctwolines{\textbf{(B) }\set{\interof{\interac{y}{\lrecv}{x}{\send}},
			\interof{\treeidle(y)}, \interof{\treeidle(z)}}}
			{\compin{x}{\stateidle} \ast \interac{y}{\lrecv}{x}{\send} \ast \treeidle(y) \ast \treeidle(z)}
			{\interof{\interac{y}{\lrecv}{x}{\send}}}
			{\compin{x}{\stateidle} \ast \interac{y}{\lrecv}{x}{\send} \ast \treeidle(y) \ast \treeidle(z)}
	\end{prooftree}
}

\hspace{3mm}

{ \tiny
	\begin{prooftree}
		\AxiomC{}
		\LeftLabel{(\epsilonr)}
		\havocline{\set{\interof{\treeidle(z)}}}
			{\compin{x}{\stateidle} \ast \treeidle(z)}
			{\epsilon}
			{\compin{x}{\stateidle} \ast \treeidle(z)}

		\AxiomC{backlink to \textbf{(2)}}
		\havocline{\set{\interof{\treeidle(y)}}}
			{\treeidle(y)}
			{\interof{\treeidle(y)}}
			{\treeidle(y)}
		\LeftLabel{(\parr)}
		\binhavocline{\set{\interof{\treeidle(y)}, \interof{\treeidle(z)}}}
			{\compin{x}{\stateidle} \ast \treeidle(y) \ast \treeidle(z)}
			{\interof{\treeidle(y)}}
			{\compin{x}{\stateidle} \ast \treeidle(y) \ast \treeidle(z)}
		\LeftLabel{(\ii)}
		\havoctwolines{\textbf{(C) }\set{\interof{\interac{y}{\lrecv}{x}{\send}},
			\interof{\treeidle(y)}, \interof{\treeidle(z)}}}
			{\compin{x}{\stateidle} \ast \interac{y}{\lrecv}{x}{\send} \ast \treeidle(y) \ast \treeidle(z)}
			{\interof{\treeidle(y)}}
			{\compin{x}{\stateidle} \ast \interac{y}{\lrecv}{x}{\send} \ast \treeidle(y) \ast \treeidle(z)}
	\end{prooftree}
}

\hspace{3mm}

{ \tiny
	\begin{prooftree}
		\AxiomC{}
		\LeftLabel{(\disr)}
		\havocline{\set{\interof{\interac{y}{\lrecv}{x}{\send}}}}
			{\compin{x}{\stateidle} \ast \interac{y}{\lrecv}{x}{\send} \ast \compin{y}{\stateleafidle}}
			{\interof{\interac{y}{\lrecv}{x}{\send}}}
			{\predfalse}
		\LeftLabel{(\conseq)}
		\havocline{\textbf{(D) }\set{\interof{\interac{y}{\lrecv}{x}{\send}}}}
			{\compin{x}{\stateidle} \ast \interac{y}{\lrecv}{x}{\send} \ast \compin{y}{\stateleafidle}}
			{\interof{\interac{y}{\lrecv}{x}{\send}}}
			{\compin{x}{\stateidle} \ast \interac{y}{\lrecv}{x}{\send} \ast \treeidle(y)}
	\end{prooftree}
}

\hspace{3mm}

{ \tiny
	\begin{prooftree}
		\AxiomC{}
		\LeftLabel{(\disr)}
		\havoctwolines{\set{\interof{\interac{y}{\lrecv}{x}{\send}}}}
			{\compin{x}{\stateidle} \ast \interac{y}{\lrecv}{x}{\send} \ast \compin{y}{\stateidle}}
			{\interof{\interac{y}{\lrecv}{x}{\send}}}
			{\compin{x}{\stateidle} \ast \interac{y}{\lrecv}{x}{\send} \ast \compin{y}{\stateidle}}

		\AxiomC{}
		\LeftLabel{(\epsilonr)}
		\havoctwolines{\set{\interof{\treeidle(v)}, \interof{\treeidle(w)}}}
			{\treeidle(v) \ast \treeidle(w)}
			{\epsilon}
			{\treeidle(v) \ast \treeidle(w)}
		\LeftLabel{(\parr)}
		\binhavoctwolines{\set{\interof{\interac{y}{\lrecv}{x}{\send}}, \interof{\treeidle(v)},
				\interof{\treeidle(w)}}}
			{\compin{x}{\stateidle} \ast \interac{y}{\lrecv}{x}{\send} \ast \compin{y}{\stateidle}
				\ast \treeidle(v) \ast \treeidle(w)}
			{\interof{\interac{y}{\lrecv}{x}{\send}}}
			{\compin{x}{\stateidle} \ast \interac{y}{\lrecv}{x}{\send} \ast \compin{y}{\stateidle}
				\ast \treeidle(v) \ast \treeidle(w)}
		\LeftLabel{(\ii)}
		\havocfourlines{\set{\interof{\interac{y}{\lrecv}{x}{\send}}, \interof{\interac{v}{\lrecv}{y}{\send}},
			\interof{\interac{w}{\rrecv}{y}{\send}}, \interof{\treeidle(v)}, \interof{\treeidle(w)}}}
			{\compin{x}{\stateidle} \ast \interac{y}{\lrecv}{x}{\send} \ast \compin{y}{\stateidle}
				\ast \interac{v}{\lrecv}{y}{\send} \ast \interac{w}{\rrecv}{y}{\send} \ast \treeidle(v) \ast \treeidle(w)}
			{\interof{\interac{y}{\lrecv}{x}{\send}}}
			{\compin{x}{\stateidle} \ast \interac{y}{\lrecv}{x}{\send} \ast \compin{y}{\stateidle}
				\ast \interac{v}{\lrecv}{y}{\send} \ast \interac{w}{\rrecv}{y}{\send} \ast \treeidle(v) \ast \treeidle(w)}
		\LeftLabel{(\conseq)}
		\havocfourlines{\textbf{(E) }\set{\interof{\interac{y}{\lrecv}{x}{\send}}, \interof{\interac{v}{\lrecv}{y}{\send}},
			\interof{\interac{w}{\rrecv}{y}{\send}}, \interof{\treeidle(v)}, \interof{\treeidle(w)}}}
			{\compin{x}{\stateidle} \ast \interac{y}{\lrecv}{x}{\send} \ast \compin{y}{\stateidle}
				\ast \interac{v}{\lrecv}{y}{\send} \ast \interac{w}{\rrecv}{y}{\send} \ast \treeidle(v) \ast \treeidle(w)}
			{\interof{\interac{y}{\lrecv}{x}{\send}}}
			{\compin{x}{\stateidle} \ast \interac{y}{\lrecv}{x}{\send} \ast \treeidle(y)}
	\end{prooftree}
}

\hspace{3mm}

\subsection{Havoc Invariance of the Predicate Atom $\treenotidle(x)$}

The proof of the invariance of the predicate $\treenotidle(x)$ is
similar to the previous proofs.

{ \tiny
	\begin{prooftree}
		\AxiomC{}
		\LeftLabel{(\epsilonr)}
		\havocline{\emptyset}
			{\compin{x}{\stateleafbusy}}
			{\epsilon}
			{\treenotidle(x)}

		\AxiomC{\textbf{(A)\qquad(B)}\qquad similar to \textbf{(A)}}

		\LeftLabel{(\lu)}
		\binhavocline{\textbf{(3) }\set{\interof{\treenotidle(x)}}}
			{\treenotidle(x)}
			{\interof{\treenotidle(x)}}
			{\treenotidle(x)}
		\LeftLabel{($\ast$)}
		\havocline{\set{\interof{\treenotidle(x)}}}
			{\treenotidle(x)}
			{\interof{\treenotidle(x)}^*}
			{\treenotidle(x)}
	\end{prooftree}
}

\hspace{3mm}

{ \tiny
	\begin{prooftree}
    \def\defaultHypSeparation{\hskip .03in}
    \def\ScoreOverhang{0pt}
    \def\labelSpacing{2pt}
		\AxiomC{\textbf{(C)}}
		\havocfourlines{\set{\interof{\interac{z}{\rrecv}{x}{\send}},
				\interof{\treeidle(y)}, \interof{\treenotidle(z)}}}
			{\compin{x}{\stateleft} \ast \interac{z}{\rrecv}{x}{\send}
				\ast \treeidle(y) \ast \treenotidle(z)}
			{\interof{\interac{z}{\rrecv}{x}{\send}}}
			{\compin{x}{\stateleft} \ast \interac{z}{\rrecv}{x}{\send}
				\ast \treeidle(y) \ast \treenotidle(z)}

		\AxiomC{\textbf{(D)}}
		\havocfourlines{\set{\interof{\interac{z}{\rrecv}{x}{\send}},
				\interof{\treeidle(y)}, \interof{\treenotidle(z)}}}
			{\compin{x}{\stateleft} \ast \interac{z}{\rrecv}{x}{\send}
				\ast \treeidle(y) \ast \treenotidle(z)}
			{\interof{\treeidle(y)}}
			{\compin{x}{\stateleft} \ast \interac{z}{\rrecv}{x}{\send}
				\ast \treeidle(y) \ast \treenotidle(z)}

		\AxiomC{\textbf{(E)}}
		\havocfourlines{\set{\interof{\interac{z}{\rrecv}{x}{\send}},
				\interof{\treeidle(y)}, \interof{\treenotidle(z)}}}
			{\compin{x}{\stateleft} \ast \interac{z}{\rrecv}{x}{\send}
				\ast \treeidle(y) \ast \treenotidle(z)}
			{\interof{\treenotidle(z)}}
			{\compin{x}{\stateleft} \ast \interac{z}{\rrecv}{x}{\send}
				\ast \treeidle(y) \ast \treenotidle(z)}

		\LeftLabel{($\cup$)}
		\trihavoctwolines{\set{\interof{\interac{z}{\rrecv}{x}{\send}},
				\interof{\treeidle(y)}, \interof{\treenotidle(z)}}}
			{\compin{x}{\stateleft} \ast \interac{z}{\rrecv}{x}{\send}
				\ast \treeidle(y) \ast \treenotidle(z)}
			{\interof{\interac{z}{\rrecv}{x}{\send}} \cup
				\interof{\treeidle(y)} \cup \interof{\treenotidle(z)}}
			{\compin{x}{\stateleft} \ast \interac{z}{\rrecv}{x}{\send}
				\ast \treeidle(y) \ast \treenotidle(z)}
		\LeftLabel{(\iidisr)}
		\havoctwolines{\set{\interof{\interac{y}{\lrecv}{x}{\send}}, \interof{\interac{z}{\rrecv}{x}{\send}},
				\interof{\treeidle(y)}, \interof{\treenotidle(z)}}}
			{\compin{x}{\stateleft} \ast \interac{y}{\lrecv}{x}{\send} \ast \interac{z}{\rrecv}{x}{\send}
				\ast \treeidle(y) \ast \treenotidle(z)}
			{\interof{\interac{y}{\lrecv}{x}{\send}} \cup \interof{\interac{z}{\rrecv}{x}{\send}} \cup
				\interof{\treeidle(y)} \cup \interof{\treenotidle(z)}}
			{\compin{x}{\stateleft} \ast \interac{y}{\lrecv}{x}{\send} \ast \interac{z}{\rrecv}{x}{\send}
				\ast \treeidle(y) \ast \treenotidle(z)}
		\LeftLabel{(\conseq)}
		\havoctwolines{\textbf{(A) }\set{\interof{\interac{y}{\lrecv}{x}{\send}}, \interof{\interac{z}{\rrecv}{x}{\send}},
				\interof{\treeidle(y)}, \interof{\treenotidle(z)}}}
			{\compin{x}{\stateleft} \ast \interac{y}{\lrecv}{x}{\send} \ast \interac{z}{\rrecv}{x}{\send}
				\ast \treeidle(y) \ast \treenotidle(z)}
			{\interof{\interac{y}{\lrecv}{x}{\send}} \cup \interof{\interac{z}{\rrecv}{x}{\send}} \cup
				\interof{\treeidle(y)} \cup \interof{\treenotidle(z)}}
			{\treenotidle(x)}
	\end{prooftree}
}

\hspace{3mm}

{ \tiny
	\begin{prooftree}
		\AxiomC{}
		\havoctwolines{\emptyset}
			{\compin{x}{\stateright}}
			{\epsilon}
			{\compin{x}{\stateright}}

		\AxiomC{backlink to \textbf{(2)}}
		\havoctwolines{\set{\interof{\treeidle(y)}}}
			{\treeidle(y)}
			{\interof{\treeidle(y)}}
			{\treeidle(y)}

		\AxiomC{backlink to \textbf{(2)}}
		\havoctwolines{\set{\interof{\treeidle(z)}}}
			{\treeidle(z)}
			{\interof{\treeidle(z)}}
			{\treeidle(z)}
		\LeftLabel{(\parr)}
		\trihavocline{\set{\interof{\treeidle(y)}, \interof{\treeidle(z)}}}
			{\compin{x}{\stateright} \ast \treeidle(y) \ast \treeidle(z)}
			{\interof{\treeidle(y)} \cup \interof{\treeidle(z)}}
			{\compin{x}{\stateright} \ast \treeidle(y) \ast \treeidle(z)}
		\LeftLabel{(\iidisr)}
		\havoctwolines{\set{\interof{\interac{y}{\lrecv}{x}{\send}}, \interof{\interac{z}{\rrecv}{x}{\send}},
				\interof{\treeidle(y)}, \interof{\treeidle(z)}}}
			{\compin{x}{\stateright} \ast \interac{y}{\lrecv}{x}{\send} \ast \interac{z}{\rrecv}{x}{\send}
				\ast \treeidle(y) \ast \treeidle(z)}
			{\interof{\interac{y}{\lrecv}{x}{\send}} \cup \interof{\interac{z}{\rrecv}{x}{\send}} \cup
				\interof{\treeidle(y)} \cup \interof{\treeidle(z)}}
			{\compin{x}{\stateright} \ast \interac{y}{\lrecv}{x}{\send} \ast \interac{z}{\rrecv}{x}{\send}
				\ast \treeidle(y) \ast \treeidle(z)}
		\LeftLabel{(\conseq)}
		\havoctwolines{\textbf{(B) }\set{\interof{\interac{y}{\lrecv}{x}{\send}}, \interof{\interac{z}{\rrecv}{x}{\send}},
				\interof{\treeidle(y)}, \interof{\treeidle(z)}}}
			{\compin{x}{\stateright} \ast \interac{y}{\lrecv}{x}{\send} \ast \interac{z}{\rrecv}{x}{\send}
				\ast \treeidle(y) \ast \treeidle(z)}
			{\interof{\interac{y}{\lrecv}{x}{\send}} \cup \interof{\interac{z}{\rrecv}{x}{\send}} \cup
				\interof{\treeidle(y)} \cup \interof{\treeidle(z)}}
			{\treenotidle(x)}
	\end{prooftree}
}

\hspace{3mm}

{ \tiny
	\begin{prooftree}
		\AxiomC{\textbf{(F)\qquad(G)\qquad(H)\qquad(I)}}
		\LeftLabel{(\lu)}
		\havoctwolines{\set{\interof{\interac{z}{\rrecv}{x}{\send}},
				\interof{\treenotidle(z)}}}
			{\compin{x}{\stateleft} \ast \interac{z}{\rrecv}{x}{\send}
				\ast \treenotidle(z)}
			{\interof{\interac{z}{\rrecv}{x}{\send}}}
			{\compin{x}{\stateleft} \ast \interac{z}{\rrecv}{x}{\send}
				\ast \treenotidle(z)}

		\AxiomC{}
		\LeftLabel{(\epsilonr)}
		\havoctwolines{\set{\interof{\treeidle(y)}}}
			{\treeidle(y)}
			{\epsilon}
			{\treeidle(y)}
		\LeftLabel{(\parr)}
		\binhavoctwolines{\textbf{(C) }\set{\interof{\interac{z}{\rrecv}{x}{\send}},
				\interof{\treeidle(y)}, \interof{\treenotidle(z)}}}
			{\compin{x}{\stateleft} \ast \interac{z}{\rrecv}{x}{\send}
				\ast \treeidle(y) \ast \treenotidle(z)}
			{\interof{\interac{z}{\rrecv}{x}{\send}}}
			{\compin{x}{\stateleft} \ast \interac{z}{\rrecv}{x}{\send}
				\ast \treeidle(y) \ast \treenotidle(z)}
	\end{prooftree}
}

\hspace{3mm}

{ \tiny
	\begin{prooftree}
		\AxiomC{}
		\LeftLabel{(\epsilonr)}
		\havoctwolines{\set{\interof{\interac{z}{\rrecv}{x}{\send}}, \interof{\treenotidle(z)}}}
			{\compin{x}{\stateleft} \ast \interac{z}{\rrecv}{x}{\send} \ast \treenotidle(z)}
			{\epsilon}
			{\compin{x}{\stateleft} \ast \interac{z}{\rrecv}{x}{\send} \ast \treenotidle(z)}

		\AxiomC{backlink to \textbf{(2)}}
		\havoctwolines{\set{\interof{\treeidle(y)}}}
			{\treeidle(y)}
			{\interof{\treeidle(y)}}
			{\treeidle(y)}
		\LeftLabel{(\parr)}
		\binhavoctwolines{\textbf{(D) }\set{\interof{\interac{z}{\rrecv}{x}{\send}},
				\interof{\treeidle(y)}, \interof{\treenotidle(z)}}}
			{\compin{x}{\stateleft} \ast \interac{z}{\rrecv}{x}{\send}
				\ast \treeidle(y) \ast \treenotidle(z)}
			{\interof{\treeidle(y)}}
			{\compin{x}{\stateleft} \ast \interac{z}{\rrecv}{x}{\send}
				\ast \treeidle(y) \ast \treenotidle(z)}
	\end{prooftree}
}

\hspace{3mm}

{ \tiny
	\begin{prooftree}
		\AxiomC{}
		\LeftLabel{(\epsilonr)}
		\havocline{\set{\interof{\treeidle(y)}}}
			{\compin{x}{\stateleft} \ast \treeidle(y)}
			{\epsilon}
			{\compin{x}{\stateleft} \ast \treeidle(y)}

		\AxiomC{backlink to \textbf{(3)}}
		\havocline{\set{\interof{\treenotidle(z)}}}
			{\treenotidle(z)}
			{\interof{\treenotidle(z)}}
			{\treenotidle(z)}
		\LeftLabel{(\parr)}
		\binhavocline{\set{\interof{\treeidle(y)}, \interof{\treenotidle(z)}}}
			{\compin{x}{\stateleft} \ast \treeidle(y) \ast \treenotidle(z)}
			{\interof{\treenotidle(z)}}
			{\compin{x}{\stateleft} \ast \treeidle(y) \ast \treenotidle(z)}
		\LeftLabel{(\ii)}
		\havoctwolines{\textbf{(E) }\set{\interof{\interac{z}{\rrecv}{x}{\send}},
				\interof{\treeidle(y)}, \interof{\treenotidle(z)}}}
			{\compin{x}{\stateleft} \ast \interac{z}{\rrecv}{x}{\send}
				\ast \treeidle(y) \ast \treenotidle(z)}
			{\interof{\treenotidle(z)}}
			{\compin{x}{\stateleft} \ast \interac{z}{\rrecv}{x}{\send}
				\ast \treeidle(y) \ast \treenotidle(z)}
	\end{prooftree}
}

\hspace{3mm}

{ \tiny
	\begin{prooftree}
		\AxiomC{}
		\LeftLabel{(\inter)}
		\havocline{\set{\interof{\interac{z}{\rrecv}{x}{\send}}}}
			{\compin{x}{\stateleft} \ast \interac{z}{\rrecv}{x}{\send}
				\ast \compin{z}{\stateleafbusy}}
			{\interof{\interac{z}{\rrecv}{x}{\send}}}
			{\compin{x}{\stateright} \ast \interac{z}{\rrecv}{x}{\send}
				\ast \compin{z}{\stateleafidle}}
		\LeftLabel{(\conseq)}
		\havocline{\textbf{(F) }\set{\interof{\interac{z}{\rrecv}{x}{\send}}}}
			{\compin{x}{\stateleft} \ast \interac{z}{\rrecv}{x}{\send}
				\ast \compin{z}{\stateleafbusy}}
			{\interof{\interac{z}{\rrecv}{x}{\send}}}
			{\compin{x}{\stateleft} \ast \interac{z}{\rrecv}{x}{\send}
				\ast \treenotidle(z)}
	\end{prooftree}
}

\hspace{3mm}

{ \tiny
	\begin{prooftree}
		\AxiomC{}
		\LeftLabel{(\disr)}
		\havoctwolines{\set{\interof{\interac{z}{\rrecv}{x}{\send}}}}
			{\compin{x}{\stateleft} \ast \interac{z}{\rrecv}{x}{\send}
				\ast \compin{z}{\stateleft}}
			{\interof{\interac{z}{\rrecv}{x}{\send}}}
			{\predfalse}
		\LeftLabel{(\conseq)}
		\havoctwolines{\set{\interof{\interac{z}{\rrecv}{x}{\send}}}}
			{\compin{x}{\stateleft} \ast \interac{z}{\rrecv}{x}{\send}
				\ast \compin{z}{\stateleft}}
			{\interof{\interac{z}{\rrecv}{x}{\send}}}
			{\compin{x}{\stateleft} \ast \interac{z}{\rrecv}{x}{\send}
				\ast \compin{z}{\stateleft}}

		\AxiomC{}
		\LeftLabel{(\epsilonr)}
		\havoctwolines{\set{\interof{\treeidle(v)}, \interof{\treenotidle(w)}}}
			{\treeidle(v) \ast \treenotidle(w)}
			{\epsilon}
			{\treeidle(v) \ast \treenotidle(w)}
		\LeftLabel{(\parr)}
		\binhavoctwolines{\set{\interof{\interac{z}{\rrecv}{x}{\send}},
				\interof{\treeidle(v)}, \interof{\treenotidle(w)}}}
			{\compin{x}{\stateleft} \ast \interac{z}{\rrecv}{x}{\send}
				\ast \compin{z}{\stateleft} \ast \treeidle(v) \ast \treenotidle(w)}
			{\interof{\interac{z}{\rrecv}{x}{\send}}}
			{\compin{x}{\stateleft} \ast \interac{z}{\rrecv}{x}{\send}
				\ast \compin{z}{\stateleft} \ast \treeidle(v) \ast \treenotidle(w)}
		\LeftLabel{(\ii)}
		\havocfourlines{\set{\interof{\interac{z}{\rrecv}{x}{\send}},
				\interof{\interac{v}{\lrecv}{z}{\send}}, \interof{\interac{w}{\rrecv}{z}{\send}},
				\interof{\treeidle(v)}, \interof{\treenotidle(w)}}}
			{\compin{x}{\stateleft} \ast \interac{z}{\rrecv}{x}{\send}
				\ast \compin{z}{\stateleft} \ast \interac{v}{\lrecv}{z}{\send}
				\ast \interac{w}{\rrecv}{z}{\send} \ast \treeidle(v) \ast \treenotidle(w)}
			{\interof{\interac{z}{\rrecv}{x}{\send}}}
			{\compin{x}{\stateleft} \ast \interac{z}{\rrecv}{x}{\send}
				\ast \compin{z}{\stateleft} \ast \interac{v}{\lrecv}{z}{\send}
				\ast \interac{w}{\rrecv}{z}{\send} \ast \treeidle(v) \ast \treenotidle(w)}
		\LeftLabel{(\conseq)}
		\havocfourlines{\textbf{(G) }\set{\interof{\interac{z}{\rrecv}{x}{\send}},
				\interof{\interac{v}{\lrecv}{z}{\send}}, \interof{\interac{w}{\rrecv}{z}{\send}},
				\interof{\treeidle(v)}, \interof{\treenotidle(w)}}}
			{\compin{x}{\stateleft} \ast \interac{z}{\rrecv}{x}{\send}
				\ast \compin{z}{\stateleft} \ast \interac{v}{\lrecv}{z}{\send}
				\ast \interac{w}{\rrecv}{z}{\send} \ast \treeidle(v) \ast \treenotidle(w)}
			{\interof{\interac{z}{\rrecv}{x}{\send}}}
			{\compin{x}{\stateleft} \ast \interac{z}{\rrecv}{x}{\send}
				\ast \treenotidle(z)}
	\end{prooftree}
}

\hspace{3mm}

{ \tiny
	\begin{prooftree}
		\AxiomC{}
		\LeftLabel{(\inter)}
		\havoctwolines{\set{\interof{\interac{z}{\rrecv}{x}{\send}}}}
			{\compin{x}{\stateleft} \ast \interac{z}{\rrecv}{x}{\send}
				\ast \compin{z}{\stateright}}
			{\interof{\interac{z}{\rrecv}{x}{\send}}}
			{\compin{x}{\stateright} \ast \interac{z}{\rrecv}{x}{\send}
				\ast \compin{z}{\stateidle}}
		\LeftLabel{(\conseq)}
		\havoctwolines{\set{\interof{\interac{z}{\rrecv}{x}{\send}}}}
			{\compin{x}{\stateleft} \ast \interac{z}{\rrecv}{x}{\send}
				\ast \compin{z}{\stateright}}
			{\interof{\interac{z}{\rrecv}{x}{\send}}}
			{\compin{x}{\stateright} \ast \interac{z}{\rrecv}{x}{\send}
				\ast \compin{z}{\stateright}}

		\AxiomC{}
		\LeftLabel{(\epsilonr)}
		\havoctwolines{\set{\interof{\treeidle(v)}, \interof{\treeidle(w)}}}
			{\treeidle(v) \ast \treeidle(w)}
			{\epsilon}
			{\treeidle(v) \ast \treeidle(w)}
		\LeftLabel{(\parr)}
		\binhavoctwolines{\set{\interof{\interac{z}{\rrecv}{x}{\send}},
				\interof{\treeidle(v)}, \interof{\treeidle(w)}}}
			{\compin{x}{\stateleft} \ast \interac{z}{\rrecv}{x}{\send}
				\ast \compin{z}{\stateright} \ast \treeidle(v) \ast \treeidle(w)}
			{\interof{\interac{z}{\rrecv}{x}{\send}}}
			{\compin{x}{\stateleft} \ast \interac{z}{\rrecv}{x}{\send}
				\ast \compin{z}{\stateright} \ast \treeidle(v) \ast \treeidle(w)}
		\LeftLabel{(\ii)}
		\havocfourlines{\set{\interof{\interac{z}{\rrecv}{x}{\send}},
				\interof{\interac{v}{\lrecv}{z}{\send}}, \interof{\interac{w}{\rrecv}{z}{\send}},
				\interof{\treeidle(v)}, \interof{\treeidle(w)}}}
			{\compin{x}{\stateleft} \ast \interac{z}{\rrecv}{x}{\send}
				\ast \compin{z}{\stateright} \ast \interac{v}{\lrecv}{z}{\send}
				\ast \interac{w}{\rrecv}{z}{\send} \ast \treeidle(v) \ast \treeidle(w)}
			{\interof{\interac{z}{\rrecv}{x}{\send}}}
			{\compin{x}{\stateleft} \ast \interac{z}{\rrecv}{x}{\send}
				\ast \compin{z}{\stateright} \ast \interac{v}{\lrecv}{z}{\send}
				\ast \interac{w}{\rrecv}{z}{\send} \ast \treeidle(v) \ast \treeidle(w)}
		\LeftLabel{(\conseq)}
		\havocfourlines{\textbf{(H) }\set{\interof{\interac{z}{\rrecv}{x}{\send}},
				\interof{\interac{v}{\lrecv}{z}{\send}}, \interof{\interac{w}{\rrecv}{z}{\send}},
				\interof{\treeidle(v)}, \interof{\treeidle(w)}}}
			{\compin{x}{\stateleft} \ast \interac{z}{\rrecv}{x}{\send}
				\ast \compin{z}{\stateright} \ast \interac{v}{\lrecv}{z}{\send}
				\ast \interac{w}{\rrecv}{z}{\send} \ast \treeidle(v) \ast \treeidle(w)}
			{\interof{\interac{z}{\rrecv}{x}{\send}}}
			{\compin{x}{\stateleft} \ast \interac{z}{\rrecv}{x}{\send}
				\ast \treeidle(z)}
	\end{prooftree}
}

\hspace{3mm}

{ \tiny
	\begin{prooftree}
		\AxiomC{}
		\LeftLabel{(\disr)}
		\havoctwolines{\set{\interof{\interac{z}{\rrecv}{x}{\send}}}}
			{\compin{x}{\stateleft} \ast \interac{z}{\rrecv}{x}{\send}
				\ast \compin{z}{\stateidle}}
			{\interof{\interac{z}{\rrecv}{x}{\send}}}
			{\predfalse}
		\LeftLabel{(\conseq)}
		\havoctwolines{\set{\interof{\interac{z}{\rrecv}{x}{\send}}}}
			{\compin{x}{\stateleft} \ast \interac{z}{\rrecv}{x}{\send}
				\ast \compin{z}{\stateidle}}
			{\interof{\interac{z}{\rrecv}{x}{\send}}}
			{\compin{x}{\stateleft} \ast \interac{z}{\rrecv}{x}{\send}
				\ast \compin{z}{\stateidle}}

		\AxiomC{}
		\LeftLabel{(\epsilonr)}
		\havoctwolines{\set{\interof{\treenotidle(v)}, \interof{\treenotidle(w)}}}
			{\treenotidle(v) \ast \treenotidle(w)}
			{\epsilon}
			{\treenotidle(v) \ast \treenotidle(w)}
		\LeftLabel{(\parr)}
		\binhavoctwolines{\set{\interof{\interac{z}{\rrecv}{x}{\send}},
				\interof{\treenotidle(v)}, \interof{\treenotidle(w)}}}
			{\compin{x}{\stateleft} \ast \interac{z}{\rrecv}{x}{\send}
				\ast \compin{z}{\stateidle} \ast \treenotidle(v) \ast \treenotidle(w)}
			{\interof{\interac{z}{\rrecv}{x}{\send}}}
			{\compin{x}{\stateleft} \ast \interac{z}{\rrecv}{x}{\send}
				\ast \compin{z}{\stateidle} \ast \treenotidle(v) \ast \treenotidle(w)}
		\LeftLabel{(\ii)}
		\havocfourlines{\set{\interof{\interac{z}{\rrecv}{x}{\send}},
				\interof{\interac{v}{\lrecv}{z}{\send}}, \interof{\interac{w}{\rrecv}{z}{\send}},
				\interof{\treenotidle(v)}, \interof{\treenotidle(w)}}}
			{\compin{x}{\stateleft} \ast \interac{z}{\rrecv}{x}{\send}
				\ast \compin{z}{\stateidle} \ast \interac{v}{\lrecv}{z}{\send}
				\ast \interac{w}{\rrecv}{z}{\send} \ast \treenotidle(v) \ast \treenotidle(w)}
			{\interof{\interac{z}{\rrecv}{x}{\send}}}
			{\compin{x}{\stateleft} \ast \interac{z}{\rrecv}{x}{\send}
				\ast \compin{z}{\stateidle} \ast \interac{v}{\lrecv}{z}{\send}
				\ast \interac{w}{\rrecv}{z}{\send} \ast \treenotidle(v) \ast \treenotidle(w)}
		\LeftLabel{(\conseq)}
		\havocfourlines{\textbf{(I) }\set{\interof{\interac{z}{\rrecv}{x}{\send}},
				\interof{\interac{v}{\lrecv}{z}{\send}}, \interof{\interac{w}{\rrecv}{z}{\send}},
				\interof{\treenotidle(v)}, \interof{\treenotidle(w)}}}
			{\compin{x}{\stateleft} \ast \interac{z}{\rrecv}{x}{\send}
				\ast \compin{z}{\stateidle} \ast \interac{v}{\lrecv}{z}{\send}
				\ast \interac{w}{\rrecv}{z}{\send} \ast \treenotidle(v) \ast \treenotidle(w)}
			{\interof{\interac{z}{\rrecv}{x}{\send}}}
			{\compin{x}{\stateleft} \ast \interac{z}{\rrecv}{x}{\send}
				\ast \treenotidle(z)}
	\end{prooftree}
}

\hspace{3mm}

\subsection{Havoc Invariance of the Predicate Atom $\treeseg(x,u)$}

Lastly, we prove the invariance of the predicate $\treeseg(x,u)$
via the rules in Fig.~\ref{fig:havoc-rules}.

{ \tiny
	\begin{prooftree}
		\AxiomC{}
		\LeftLabel{(\epsilonr)}
		\havoctwolines{\emptyset}
			{\company{x}}
			{\epsilon}
			{\treeseg(x,u)}

		\AxiomC{\textbf{(A)}}
		\havocfourlines{\set{\interof{\interac{y}{\lrecv}{x}{\send}}, \interof{\interac{z}{\rrecv}{x}{\send}},
			\interof{\treeseg(y,u)}, \interof{\tree(z)}}}
			{\company{x} \ast \interac{y}{\lrecv}{x}{\send} \ast \interac{z}{\rrecv}{x}{\send} \ast
				\treeseg(y,u) \ast \tree(z)}
			{\interof{\interac{y}{\lrecv}{x}{\send}} \cup \interof{\interac{z}{\rrecv}{x}{\send}} \cup
				\interof{\treeseg(y,u)} \cup \interof{\tree(z)}}
			{\treeseg(x,u)}

		\AxiomC{similar to \textbf{(A)}}
		\havocfourlines{\set{\interof{\interac{y}{\lrecv}{x}{\send}}, \interof{\interac{z}{\rrecv}{x}{\send}},
			\interof{\tree(y)}, \interof{\treeseg(z,u)}}}
			{\company{x} \ast \interac{y}{\lrecv}{x}{\send} \ast \interac{z}{\rrecv}{x}{\send} \ast
				\tree(y) \ast \treeseg(z,u)}
			{\interof{\interac{y}{\lrecv}{x}{\send}} \cup \interof{\interac{z}{\rrecv}{x}{\send}} \cup
				\interof{\tree(y)} \cup \interof{\treeseg(z,u)}}
			{\treeseg(x,u)}

		\LeftLabel{(\lu)}
		\trihavocline{\textbf{(4) }\set{\interof{\treeseg(x,u)}}}
			{\treeseg(x,u)}
			{\interof{\treeseg(x,u)}}
			{\treeseg(x,u)}
		\LeftLabel{($\ast$)}
		\havocline{\set{\interof{\treeseg(x,u)}}}
			{\treeseg(x,u)}
			{\interof{\treeseg(x,u)}^*}
			{\treeseg(x,u)}
	\end{prooftree}
}

\hspace{3mm}

{ \tiny
	\begin{prooftree}
		\AxiomC{\textbf{(B)}\qquad similar to \textbf{(B)\qquad (C)\qquad (D)}}
		\LeftLabel{($\cup$)}
		\havoctwolines{\set{\interof{\interac{y}{\lrecv}{x}{\send}}, \interof{\interac{z}{\rrecv}{x}{\send}},
			\interof{\treeseg(y,u)}, \interof{\tree(z)}}}
			{\company{x} \ast \interac{y}{\lrecv}{x}{\send} \ast \interac{z}{\rrecv}{x}{\send} \ast
				\treeseg(y,u) \ast \tree(z)}
			{\interof{\interac{y}{\lrecv}{x}{\send}} \cup \interof{\interac{z}{\rrecv}{x}{\send}} \cup
				\interof{\treeseg(y,u)} \cup \interof{\tree(z)}}
			{\company{x} \ast \interac{y}{\lrecv}{x}{\send} \ast \interac{z}{\rrecv}{x}{\send} \ast
				\treeseg(y,u) \ast \tree(z)}
		\LeftLabel{(\conseq)}
		\havoctwolines{\textbf{(A) }\set{\interof{\interac{y}{\lrecv}{x}{\send}}, \interof{\interac{z}{\rrecv}{x}{\send}},
			\interof{\treeseg(y,u)}, \interof{\tree(z)}}}
			{\company{x} \ast \interac{y}{\lrecv}{x}{\send} \ast \interac{z}{\rrecv}{x}{\send} \ast
				\treeseg(y,u) \ast \tree(z)}
			{\interof{\interac{y}{\lrecv}{x}{\send}} \cup \interof{\interac{z}{\rrecv}{x}{\send}} \cup
				\interof{\treeseg(y,u)} \cup \interof{\tree(z)}}
			{\treeseg(x,u)}
	\end{prooftree}
}

\hspace{3mm}

{ \tiny
	\begin{prooftree}
		\AxiomC{\textbf{(E)\qquad(F)}\qquad similar to \textbf{(F)}}
		\LeftLabel{(\lu)}
		\havoctwolines{\set{\interof{\interac{y}{\lrecv}{x}{\send}},
			\interof{\treeseg(y,u)}}}
			{\company{x} \ast \interac{y}{\lrecv}{x}{\send} \ast \treeseg(y,u)}
			{\interof{\interac{y}{\lrecv}{x}{\send}}}
			{\company{x} \ast \interac{y}{\lrecv}{x}{\send} \ast \treeseg(y,u)}

		\AxiomC{}
		\LeftLabel{(\epsilonr)}
		\havocline{\set{\interof{\tree(z)}}}
			{\tree(z)}
			{\epsilon}
			{\tree(z)}
		\LeftLabel{(\parr)}
		\binhavocline{\set{\interof{\interac{y}{\lrecv}{x}{\send}},
			\interof{\treeseg(y,u)}, \interof{\tree(z)}}}
			{\company{x} \ast \interac{y}{\lrecv}{x}{\send} \ast \treeseg(y,u) \ast \tree(z)}
			{\interof{\interac{y}{\lrecv}{x}{\send}}}
			{\company{x} \ast \interac{y}{\lrecv}{x}{\send} \ast \treeseg(y,u) \ast \tree(z)}
		\LeftLabel{(\ii)}
		\havoctwolines{\textbf{(B) }\set{\interof{\interac{y}{\lrecv}{x}{\send}}, \interof{\interac{z}{\rrecv}{x}{\send}},
			\interof{\treeseg(y,u)}, \interof{\tree(z)}}}
			{\company{x} \ast \interac{y}{\lrecv}{x}{\send} \ast \interac{z}{\rrecv}{x}{\send} \ast
				\treeseg(y,u) \ast \tree(z)}
			{\interof{\interac{y}{\lrecv}{x}{\send}}}
			{\company{x} \ast \interac{y}{\lrecv}{x}{\send} \ast \interac{z}{\rrecv}{x}{\send} \ast
				\treeseg(y,u) \ast \tree(z)}
	\end{prooftree}
}

\hspace{3mm}

{ \tiny
	\begin{prooftree}
		\AxiomC{}
		\LeftLabel{(\epsilonr)}
		\havoctwolines{\set{\interof{\tree(z)}}}
			{\company{x} \ast \tree(z)}
			{\epsilon}
			{\company{x} \ast \tree(z)}

		\AxiomC{backlink to \textbf{(4)}}
		\havoctwolines{\set{\interof{\treeseg(y,u)}}}
			{\treeseg(y,u)}
			{\interof{\treeseg(y,u)}}
			{\treeseg(y,u)}
		\LeftLabel{(\parr)}
		\binhavoctwolines{\set{\interof{\treeseg(y,u)}, \interof{\tree(z)}}}
			{\company{x} \ast \treeseg(y,u) \ast \tree(z)}
			{\interof{\treeseg(y,u)}}
			{\company{x} \ast \treeseg(y,u) \ast \tree(z)}
		\LeftLabel{(\ii)}
		\havoctwolines{\textbf{(C) }\set{\interof{\interac{y}{\lrecv}{x}{\send}}, \interof{\interac{z}{\rrecv}{x}{\send}},
			\interof{\treeseg(y,u)}, \interof{\tree(z)}}}
			{\company{x} \ast \interac{y}{\lrecv}{x}{\send} \ast \interac{z}{\rrecv}{x}{\send} \ast
				\treeseg(y,u) \ast \tree(z)}
			{\interof{\treeseg(y,u)}}
			{\company{x} \ast \interac{y}{\lrecv}{x}{\send} \ast \interac{z}{\rrecv}{x}{\send} \ast
				\treeseg(y,u) \ast \tree(z)}
	\end{prooftree}
}

\hspace{3mm}

{ \tiny
	\begin{prooftree}
		\AxiomC{}
		\LeftLabel{(\epsilonr)}
		\havocline{\set{\interof{\treeseg(y,u)}}}
			{\company{x} \ast \treeseg(y,u)}
			{\epsilon}
			{\company{x} \ast \treeseg(y,u)}

		\AxiomC{backlink to \textbf{(1)}}
		\havocline{\set{\interof{\tree(z)}}}
			{\tree(z)}
			{\interof{\tree(z)}}
			{\tree(z)}
		\LeftLabel{(\parr)}
		\binhavocline{\set{\interof{\treeseg(y,u)}, \interof{\tree(z)}}}
			{\company{x} \ast \treeseg(y,u) \ast \tree(z)}
			{\interof{\tree(z)}}
			{\company{x} \ast \treeseg(y,u) \ast \tree(z)}
		\LeftLabel{(\ii)}
		\havoctwolines{\textbf{(D) }\set{\interof{\interac{y}{\lrecv}{x}{\send}}, \interof{\interac{z}{\rrecv}{x}{\send}},
			\interof{\treeseg(y,u)}, \interof{\tree(z)}}}
			{\company{x} \ast \interac{y}{\lrecv}{x}{\send} \ast \interac{z}{\rrecv}{x}{\send} \ast
				\treeseg(y,u) \ast \tree(z)}
			{\interof{\tree(z)}}
			{\company{x} \ast \interac{y}{\lrecv}{x}{\send} \ast \interac{z}{\rrecv}{x}{\send} \ast
				\treeseg(y,u) \ast \tree(z)}
	\end{prooftree}
}

\hspace{3mm}

{ \tiny
	\begin{prooftree}
    \def\defaultHypSeparation{\hskip .1in}
    \def\ScoreOverhang{0pt}
    \def\labelSpacing{4pt}

		\AxiomC{}
		\LeftLabel{(\inter)}
		\havoctwolines{\set{\interof{\interac{y}{\lrecv}{x}{\send}}}}
			{\compin{x}{\stateidle} \ast \interac{y}{\lrecv}{x}{\send} \ast \compin{y}{p}}
			{\interof{\interac{y}{\lrecv}{x}{\send}}}
			{\compin{x}{\stateleft} \ast \interac{y}{\lrecv}{x}{\send} \ast \compin{y}{p'}}
		\LeftLabel{(\conseq)}
		\RightLabel{\shortstack[l]{for ($p = \stateright$, \\ $p' = \stateidle$) \\
			or ($p = \stateleafbusy$,\\ $p' = \stateleafidle$)}}
		\havoctwolines{\set{\interof{\interac{y}{\lrecv}{x}{\send}}}}
			{\compin{x}{\stateidle} \ast \interac{y}{\lrecv}{x}{\send} \ast \compin{y}{p}}
			{\interof{\interac{y}{\lrecv}{x}{\send}}}
			{\company{x} \ast \interac{y}{\lrecv}{x}{\send} \ast \treeseg(y,u)}

		\AxiomC{}
		\LeftLabel{(\disr)}
		\havoctwolines{\set{\interof{\interac{y}{\lrecv}{x}{\send}}}}
			{\compin{x}{q} \ast \interac{y}{\lrecv}{x}{\send} \ast \compin{y}{p}}
			{\interof{\interac{y}{\lrecv}{x}{\send}}}
			{\predfalse}
		\LeftLabel{(\conseq)}
		\RightLabel{\shortstack[l]{for $(q,p) \neq (\stateidle,\stateright),$ \\
			$ (\stateidle,\stateleafbusy)$}}
		\havoctwolines{\set{\interof{\interac{y}{\lrecv}{x}{\send}}}}
			{\company{x} \ast \interac{y}{\lrecv}{x}{\send} \ast \company{y}}
			{\interof{\interac{y}{\lrecv}{x}{\send}}}
			{\company{x} \ast \interac{y}{\lrecv}{x}{\send} \ast \treeseg(y,u)}

		\LeftLabel{($\vee$)}
		\binhavoctwolines{\textbf{(E) }\set{\interof{\interac{y}{\lrecv}{x}{\send}}}}
			{\company{x} \ast \interac{y}{\lrecv}{x}{\send} \ast \company{y}}
			{\interof{\interac{y}{\lrecv}{x}{\send}}}
			{\company{x} \ast \interac{y}{\lrecv}{x}{\send} \ast \treeseg(y,u)}
	\end{prooftree}
}

\hspace{3mm}

{ \tiny
	\begin{prooftree}
		\AxiomC{backlink to \textbf{(E)}}
		\havoctwolines{\set{\interof{\interac{y}{\lrecv}{x}{\send}}}}
			{\company{x} \ast \interac{y}{\lrecv}{x}{\send} \ast \company{y}}
			{\interof{\interac{y}{\lrecv}{x}{\send}}}
			{\company{x} \ast \interac{y}{\lrecv}{x}{\send} \ast \company{y}}

		\AxiomC{}
		\LeftLabel{(\epsilonr)}
		\havoctwolines{\set{\interof{\treeseg(v,u)}, \interof{\tree(w)}}}
			{\treeseg(v,u) \ast \tree(w)}
			{\epsilon}
			{\treeseg(v,u) \ast \tree(w)}
		\LeftLabel{(\parr)}
		\binhavoctwolines{\set{\interof{\interac{y}{\lrecv}{x}{\send}},
				\interof{\treeseg(v,u)}, \interof{\tree(w)}}}
			{\company{x} \ast \interac{y}{\lrecv}{x}{\send} \ast \company{y}
				\ast \treeseg(v,u) \ast \tree(w)}
			{\interof{\interac{y}{\lrecv}{x}{\send}}}
			{\company{x} \ast \interac{y}{\lrecv}{x}{\send} \ast \company{y}
				\ast \treeseg(v,u) \ast \tree(w)}
		\LeftLabel{(\ii)}
		\havoctwolines{\set{\interof{\interac{y}{\lrecv}{x}{\send}},
				\interof{\interac{v}{\lrecv}{y}{\send}}, \interof{\interac{w}{\rrecv}{y}{\send}}, \interof{\treeseg(v,u)},
				\interof{\tree(w)}}}
			{\company{x} \ast \interac{y}{\lrecv}{x}{\send} \ast \company{y}
				\ast \interac{v}{\lrecv}{y}{\send} \ast \interac{w}{\rrecv}{y}{\send} \ast \treeseg(v,u) \ast \tree(w)}
			{\interof{\interac{y}{\lrecv}{x}{\send}}}
			{\company{x} \ast \interac{y}{\lrecv}{x}{\send} \ast \company{y}
				\ast \interac{v}{\lrecv}{y}{\send} \ast \interac{w}{\rrecv}{y}{\send} \ast \treeseg(v,u) \ast \tree(w)}
		\LeftLabel{(\conseq)}
		\havoctwolines{\textbf{(F) }\set{\interof{\interac{y}{\lrecv}{x}{\send}},
				\interof{\interac{v}{\lrecv}{y}{\send}}, \interof{\interac{w}{\rrecv}{y}{\send}}, \interof{\treeseg(v,u)},
				\interof{\tree(w)}}}
			{\company{x} \ast \interac{y}{\lrecv}{x}{\send} \ast \company{y}
				\ast \interac{v}{\lrecv}{y}{\send} \ast \interac{w}{\rrecv}{y}{\send} \ast \treeseg(v,u) \ast \tree(w)}
			{\interof{\interac{y}{\lrecv}{x}{\send}}}
			{\company{x} \ast \interac{y}{\lrecv}{x}{\send} \ast \treeseg(y,u)}
	\end{prooftree}
}

\hspace{3mm}

{ \tiny
	\begin{prooftree}
		\AxiomC{}
	\end{prooftree}
}

\hspace{3mm}

{ \tiny
	\begin{prooftree}
		\AxiomC{}
	\end{prooftree}
}

\hspace{3mm}

\end{document}